\documentclass[12pt]{article}
\usepackage{amsmath,amssymb,amsthm}
\usepackage{eurosym}
\usepackage[left=2.0cm, right=1.75cm, top=2.2cm, bottom=2.0cm]{geometry}
\usepackage{lscape}
\usepackage{graphicx}
\usepackage{pgfplots}
\usepackage[round]{natbib}
\usepackage{multirow}
\usepackage[flushleft]{threeparttable}
\usepackage{float}
\usepackage{caption, subcaption}
\usepackage{array}
\usepackage{booktabs}
\usepackage{verbatim}
\usepackage{tikz}
\usepackage{pdfpages}
\usepackage{rotating}
\usepackage{longtable}

\DeclareMathOperator*{\plim}{plim}
\DeclareMathOperator*{\argmin}{argmin}
\newtheorem{theorem}{Theorem}

\newtheorem{lemma}{Lemma}

\newtheorem{proposition}{Proposition}

\newtheorem{remark}{Remark}

\def\E{\mathbb{E}}

\begin{document}

\title{Weak Identification in Discrete Choice Models}
\author{
	David T. Frazier\footnote{
		Department of Econometrics and Business Statistics, Monash University, and the Australian Center for Excellence in Mathematics and Statistics (ACEMS)
		(\texttt{david.frazier@monash.edu}).}
	\hspace{0.3cm}
	Eric Renault\footnote{Department of Economics, University of Warwick.}
	\hspace{0.3cm}
	Lina Zhang\footnote{Department of Econometrics and Business Statistics, Monash University.}
	\hspace{0.3cm}
	Xueyan Zhao\footnote{Department of Econometrics and Business Statistics, Monash University.}
	\vspace{0.7cm}
}

\maketitle

\begin{abstract}
We study the impact of weak identification in discrete choice models, and provide insights into the determinants of identification strength in these models. Using these insights, we propose a novel test that can consistently detect weak identification in commonly applied discrete choice models, such as probit, logit, and many of their extensions. Furthermore, we demonstrate that when the null hypothesis of weak identification is rejected, Wald-based inference can be carried out using standard formulas and critical values. A Monte Carlo study compares our proposed testing approach against commonly applied weak identification tests. The results simultaneously demonstrate the good performance of our approach and the fundamental failure of using conventional weak identification tests for linear models in the discrete {choice} model context. Furthermore, we compare our approach against those commonly applied in the literature in two empirical examples: married women labor force participation, and US food aid and civil conflicts.
\end{abstract}
\hspace{6cm}


\hspace{20cm}

\section{Introduction}
A prevalent aspect of econometric research concerns  
estimating the causal impact of 
some policy relevant treatment variable $y_{2}$ on an outcome variable of interest $y_{1}$. The outcome $y_{1}$ is often qualitative in nature, and the treatment $y_{2}$ is often endogenous when using observational data in empirical studies. For {example}, there is a growing body of research that studies the causal effect
of {certain} economic conditions on the incidence of civil conflict in developing countries. {In this context,} economic
conditions may be summarized by a state variable such as ``economic
growth" (see e.g. \citealp{miguel2004economic}) or 
by a policy tool {such as} US Food Aid (see \citealp{nunn2014us}).
In such settings, the most common modelling strategy is to characterize the
qualitative outcome variable $y_{1}$ {as a known function of a latent quantitative variable
$y_{1}^{\ast }$, and with $y_{1}^{\ast }$ driven by a regression equation:}
\begin{equation}
	y_{1i}^{\ast }=\alpha y_{2i}+x_{i}^{\prime }\beta +u_{i},\;i=1,\dots,n,
	\label{struct1}
\end{equation}
{where $y_{2i}$ denotes the (scalar) variable whose causal impact is of interest}, $x_{i}$\ denotes a vector of $k_{x}$ exogenous variables and, for the sake of
expositional simplicity, {we consider $i=1,2,...,n$ as indicating independent and identically distributed (i.i.d.) cross-sectional realizations of the respective random variables.}\footnote{In the introduction, we use the terminology  ``exogenous'' to refer to the explanatory variables $x_i$ and to the instrumental variables $z_i$. In Section \ref{sec:cfa} we define, following \cite{newey1999nonparametric}, a precise concept of control variables that is 
	related, but not equivalent, to the common concept of exogeneity.} {We restrict our attention to settings where the relationship between the unobservable $y_{1i}^{\ast}$ and the observable $y_{1i}$ is given by a threshold crossing mechanism and conventionally specified as follows:}\footnote{At the cost of more involved notations, the methodology developed in this paper can easily be extended to a wide variety of multinomial models, such as, for instance, ordered probit models. To some extent, the binary case considered here is the most extreme case of information loss with respect to the observability of the latent variable.} 
\begin{equation*}
	y_{1i}=1[y_{1i}^{\ast }>0].
\end{equation*}

{The causal analysis of interest is conducted through} statistical inference on the true unknown value of the causal parameter $\alpha $ that must be carefully defined in order to account {for the (possible) presence of simultaneity. {However, more often than not, the treatment variable $y_{2}$ is not exogenous, which means that the structural model (\ref{struct1}) {can} not be interpreted as a model for the
conditional expectation of $y_{1}^{\ast }$\ given $y_{2}$ and $x$. For this reason, identification of the structural parameters in \eqref{struct1}, {and in particular the causal effect $\alpha$,} requires a set of valid (i.e., exogenous) instrumental variables (hereafter, IVs) denoted throughout by $z$.}

Critically, identification of the causal effect relies on the relevance of the underlying instruments to the treatment variable, i.e., the ``strength" of the IVs. {The consequences and detection of weak IVs has been extensively studied in} linear models, but it is currently unclear how the instrument strength in binary models affects identification of $\alpha$ and therefore any causal interpretation we may obtain in a given analysis.  


To illustrate this point, consider the concrete example given by \citet{nunn2014us} for estimating the impact of US food aid on the incidence of civil conflicts. Let $y_{2i}$ denote the amount of US food aid to
	country $i$, and assume we are interested in analyzing if $y_{2i}$ 
	has a causal impact on the probability of civil conflict, with the incidence of conflict denoted by a binary variable $y_{1i}$. In this setting, one must be concerned about the existence of reverse
	causality (``Do countries receive US aid precisely because they are not doing
	well?") or common cause (``Could US strategic objectives be a common cause for both
	conflict and food aid receipts?'') regarding these two variables, {which leads \cite{nunn2014us} to use lagged US wheat production as an IV to identify the causal impact of US food aid.} Whilst \citet{nunn2014us} consider various versions of linear probability and hazard models involving different definitions for the binary outcome of civil conflict, {across the various specifications the only measures of identification strength used by \citet{nunn2014us} to assess the validity of their conclusions are those explicitly designed for \textit{linear models}, such as the Kleibergen-Paap F-statistics from the first stage regression (\citealp{kleibergen2006generalized}), which are not statistically valid in either binary or hazard models.} 

The goal of this paper is to understand, characterize, and quantify the concept of identification strength as it pertains to discrete choice models. We make three primary contributions. First, we give a novel characterization of identification strength in endogenous discrete choice models which demonstrates that identification can be significantly impacted by factors other than the linear correlation between the instruments {and the endogenous variables.} Our second contribution is to use this characterization of identification strength to propose a consistent test for the null hypothesis that ``identification is so weak that point estimators are inconsistent," while under the alternative consistent estimation is warranted. Our final contribution is to demonstrate that, under the alternative, we can carry out Wald-based inference in the standard manner.

We now discuss these contributions in more detail, and place them into the broader literature on weak identification.

\subsection*{Testing Identification Strength: Existing Literature and Contributions}


Since the analysis of \citet{staiger1997instrumental}, practitioners have used the well-regarded ``rule-of-thumb" to measure instrument strength in the case where $y_{1i}^\ast$ is observed. The
magnitude of the $F$-statistic from the reduced form regression equation is arguably the most common measure for determining instrument strength in the linear regression model. Subsequent to the development of the rule-of-thumb, several influential refinements of this measure, and indeed the very concept of weak instruments in the linear model, have been put forward. \citet{stock2005testing} provide a quantitative definition of weak instruments in the
linear model, and use this definition
to propose a formal test for instrument weakness. While the approach of \citet{stock2005testing} relies on conditionally homoskedastic and serially uncorrelated regression
errors, an extension of the \citet{stock2005testing} testing strategy to
heteroskedastic and serially correlated errors is devised in \citet{olea2013robust}. 

However, when one moves to general nonlinear models, the impact of instrument weakness on the resulting estimates is more difficult to ascertain. As presented in \citet{antoine2009efficient,antoine2012efficient}, and following the work of \citet{hahn2002discontinuities} and \citet{caner2009testing}, there can exist a range of identification strengths in nonlinear models, between the extreme
cases of weak identification (when estimators are not consistent) and
strong identification (when estimators are consistent and asymptotically
normal at the $n^{1/2}$ rate). Indeed, these authors have shown that generalized method of moments (GMM) estimators can be consistent at a rate slower than the canonical rate of $n^{1/2}$,
but only in the case of a convergence rate strictly larger than $n^{1/4}$ is standard
inference based on the normal distribution approximation warranted. The key issue is that, when convergence is too slow and the
model is nonlinear, second-order terms in Taylor expansions, which govern the behavior of the estimator, may not be
negligible 
relative to the first-order terms, so that standard asymptotic
inference may no longer be valid. Such slow rates of convergence have also
been documented in the case of many weak instruments (see \citealp{newey2009generalized} and references therein) while a general study of nearly
strong instruments is available in \citet{andrews2012estimation}.

Using this characterization of varying identification strength, \cite{antoine2017testing} have devised a testing strategy that is capable of detecting (certain levels of) instrument strength in nonlinear models estimated by GMM.
The proposed test, dubbed the
distorted J-test (DJ test), is based on computing the GMM J-test statistic at a
perturbed value of the continuously updated GMM (CUGMM) estimator. The logic behind the test is that, if identification is truly weak, a small perturbation of the moments within the J-statistic will not significantly alter its value, while if identification is not weak this perturbation will result in a significant increase in the value of the J-statistic. Similar to other inference strategies robust to weak identification, the approach explicitly relies on the nature of the CUGMM objective function, which, as originally pointed out by \cite{stock2000gmm}, automatically controls the behaviour of the GMM objective function under weak identification.

Interestingly, \citet{antoine2017testing} have
demonstrated that their DJ test is akin to the standard rule-of-thumb when the
model is linear and homoskedastic. In contrast, they stress (see also
\citealp{windmeijer2019two} for related work in the context of clustering) that this
DJ test differs from standard ``robustified'' versions of the rule-of-thumb in
case of a heteroskedastic linear model. We note, in particular, that when
using linear probability models, one is faced (besides the well-known
criticisms of this approach) with a severely heteroskedastic linear model. 

Herein, we adapt the general testing strategy of \cite{antoine2017testing} to the case of discrete choice models and construct a consistent test for the null hypothesis that the instruments are too weak to allow consistent point estimation. Following the nomenclature of \cite{antoine2017testing}, we also refer to this test as a distorted J-test (DJ test) in this binary model context.  Similar to \cite{antoine2017testing}, we
demonstrate that our DJ test can be interpreted as a natural ``generalized rule-of-thumb" in the
context of discrete choice models, in the sense that {this test appropriately modifies the standard approach to account for both heteroskedasticity and non-linearity.
	
	We compare the performance of this test with the aforementioned existing} approaches both through Monte Carlo experiments and the analyses of two empirical examples. Monte Carlo results show that our DJ test, albeit conservative, {has respectable power. However, the crucial feature of this approach is its ability to discern} that the underlying estimator may not be reliable, while in contrast, the standard rule-of-thumb, because it overlooks information lost due to the nonlinearity of the model, will severely over-reject the null of weak identification. When applied to the two examples with real data, our DJ test is able to unambiguously determine when the null of weak identification should be rejected (as in the textbook
example of the causal effect of education of married women on their labor
force participation, with strong instruments like parents education), while it rightly questions the use of {standard inference approaches} when identification appears weak, as in the second empirical example. In
particular, contrary to the naive rule-of-thumb and Stock and Yogo test results, the DJ test casts some
doubt on the strength of the IV used in \citet{nunn2014us} and thus on the consistency of the 
estimated negative effect on war offset and the conclusion that food aid may prolong the duration of conflict.

In addition to the development of our DJ test, this paper also reinforces the asymptotic theory developed in \cite{antoine2009efficient,antoine2012efficient} regarding inference with nearly-strong instruments. By characterizing the strength of instruments in terms of a drifting data generating process, a la \cite{staiger1997instrumental} and \cite{stock2000gmm}, we demonstrate that once the null hypothesis of estimator inconsistency has been rejected, Wald-based inference can be performed as normal, up to the effect of pretesting.\footnote{For simplicity, and following \cite{antoine2017testing}, we choose to overlook the effect of pretesting on the resulting inferences in this current work.} This result is in stark contrast to the existing results for general nonlinear models under weak identification, where it has been shown that standard inference is only warranted once the rate of convergence is strictly larger than $n^{1/4}$. {The ability to perform standard inference in this setting stems from the fact that discrete choice models, while nonlinear, are built from
latent linear models, which ensures that they are close enough to linear models to permit} standard inference once the underlying estimator is consistent.
{While the convergence rate of the resulting estimator may be very slow, the studentization performed in computing Wald test statistics make their behavior consistent with the standard critical values.} In short, if our DJ test rejects the null of estimator inconsistency (which will be accomplished asymptotically with probability one under
the alternative), the practitioner can safely apply standard inference
procedures.

In this respect, our recommendation remains true to the widespread practice
of a two-stage decision rule: a pretest for weak IV followed by standard
inference when the null of weak identification is rejected. Of course, an
alternative would be to use more computationally demanding inference
strategies that are robust to weak identification. The robust approach
proposed by \citet{kleibergen2005testing} has been extended by \citet{magnusson2010inference} to the
context of limited dependent variable models. More generally, while the
existence of weak IV is a common phenomena, there is little theoretical
evidence regarding the properties of GMM estimators in endogenous discrete
choice models. Using Monte Carlo simulations, \citet{dufour2013weak}
demonstrate the poor behavior of Wald and Likelihood Ratio tests in the presence of weak
instruments. \cite{finlay2009implementing} analyze the Wald test
in probit models with weak instruments, and find that the test can significantly
over-reject the null hypothesis.

We note that the development of a consistent test for weak instrument  in discrete choice
models is particularly important since the similarity between linear
models and common discrete choice models has led
researcher to apply tests that are appropriate for linear models in this
nonlinear context. In particular, it is relatively common to see researchers apply the
rule-of-thumb developed for the linear model to
detect the presence of weak instruments in discrete choice models: see, e.g.,
\citet{miguel2004economic}, \citet{arendt2005does}, \citet{mckenzie2011can}, \citet{cawley2012medical}, \citet{block2013education} and \citet{goto2016cartel}.
However, the above studies do not question the validity of the
rule-or-thumb when it is applied in discrete choice models. Other researchers prefer
to abandon the discrete choice framework in favor of the linear probability
model; see, e.g., \citet{lochner2004effect}, \citet{powell2005importance}, \citet{kinda2010investment}, \citet{ruseski2014sport}. Besides the fact that they are heavily
heteroskedastic, linear probability models are by definition misspecified.
Since our DJ test is based on a distortion of the standard J-test statistic for
misspecification, it should not be used in the context of
misspecified moment models.

The remainder of the paper is organized as follows.

Section 2 introduces our model setup and assumptions. The key maintained
assumption is the existence of a control function, in which the conditional
probability distribution of the structural error term, given all the
variables in the reduced form regression,
coincides with the conditional distribution of the structural error term conditional on the reduced form error term. The control function approach for probit with endogeneity
has been pioneered by \citet{rivers1988limited} and led them to put forward a
{two-stage conditional maximum likelihood (2SCML) approach, and \cite{blundell2004endogeneity} propose a nonparametric extension that does not require certain of the parametric assumptions underlying the 2SCML approach.} In this section, we note that a GMM
framework allows us to obtain asymptotically equivalent estimators for the
structural parameters without necessarily resorting to a two-stage approach.
Moreover, we show that our GMM approach is also versatile enough to encompass the
Quasi-LIML approach of \citet{wooldridge2014quasi}.

In Section 3, we present our DJ test and prove its asymptotic properties: size control (under the null of weak identification)
and consistency (under the alternative). We further demonstrate that as long as the estimators are consistent (i.e., under the alternative to the null hypothesis of weak identification), standard Wald-style inference can be applied. This stands in contrast to the
general case of identification strength for nonlinear models considered in \citet{antoine2009efficient,antoine2012efficient} and \cite{andrews2014gmm}, where it is shown that in nonlinear models standard inference approaches are warranted only when the rate of convergence is faster than $n^{1/4}$. Lastly, we demonstrate that, in the context of a
discrete choice model, the DJ test can be interpreted as a generalized
rule-of-thumb that accounts for the nonlinear nature of the probit model.

Monte Carlo experiments in Section 4 compare the finite-sample properties of
our proposed test as well as the performance of other weak IV tests. Section 5 applies our test in two empirical examples: married women labor force participation (\citealp{wooldridge2010econometric}), and US
food aid and civil conflicts (\citealp{nunn2014us}). Section 6 concludes.

\section{General Framework}
{\cite{blundell2004endogeneity} propose  a control function (hereafter, CF) approach to conduct inference on the structural parameters of endogenous binary choice models. In this and the next section, we examine the impact of weak instruments on such a CF approach to inference. However, we first demonstrate the general point that a CF approach allows us to see} both the 2SCML of \cite{rivers1988limited} and the Quasi-LIML
approach of \cite{wooldridge2014quasi} as particular cases of a class of
GMM estimators, which we discuss in Section \ref{sec:2.2}. While these GMM
estimators can always be characterized by a one-step minimization problem, {using similar arguments to those in Section 6 of \cite{newey1994large},} we can also interpret the
estimator of the structural parameters as a two-step estimator, {whereby a preliminary plug-in estimator (obtained from a reduced form regression equation) is used within the moments. After establishing the general framework, in Section \ref{sec:2.3} we then sketch the weak IV issue in the context of probit models.}
	
\subsection{Model and Control Function Approach}\label{sec:cfa}
\cite{newey1999nonparametric} suggest that the key for a
CF approach is to start from a triangular simultaneous
equations model. In the context of endogenous binary choice models, this entails specifying structural and reduced form regression equations, and the mechanism generating the binary responses.

{The structural equation characterizes the response of an unobservable} endogenous
variable $y_{1i}^{\ast }$, {conditional on a scalar-valued endogenous variable $y_{2i}$ and a $k_x$-dimensional vector of explanatory variables $x_{i}$}, as the sum of an {unknown structural} function $g\left(
y_{2i},x_{i}\right) $ and a structural error term $u_{i}$:\footnote{While \citet{imbens2009identification} propose an even more general structural model
	where the error term $u_{i}$ may not be additively separable at the cost of
	more restrictive independence assumptions, such an extension is beyond
	the scope of this paper.}%
\begin{equation}
y_{1i}^{\ast }=g\left( y_{2i},x_{i}\right) +u_{i},\;\E\left[ u_{i}\right] =0.
\label{struct}
\end{equation}{For sake of expositional simplicity, we will
maintain the following linear specification for the structural function }
\begin{equation*}
g\left( y_{2i},x_{i}\right) =\alpha y_{2i}+x_{i}^{\prime }\beta ,
\end{equation*}
but we note that the analysis remains applicable to any situation where $g\left(
y_{2i},x_{i}\right) $ is a parametric function of $\left( y_{2i},x_{i}\right) $; the case of nonparametric $g(\cdot)$ is beyond the scope of this current paper, and is left for future research.
{Our primary focus of interest is the case where only the sign of the quantitative structural variable $y_{1i}^{\ast }$ is observable, which yields the structural equation defining the observed binary outcome $y_{1i}$:}\footnote{The binary choice model allows
	us to address the issue of weak identification in the case of maximum
	information loss going from the quantitative latent variable $y_{1i}^{\ast }
	$ to the observed variable $y_{1i}$. However, we note that the general methodology developed in
	this paper would be similarly relevant for any observation scheme that would
	define $y_{1i}$ as a known function of $y_{1i}^{\ast }$ and $x_{i}$\ (see
	e.g Tobit model, Gompit model, disequilibrium model, etc.).}
\begin{equation*}
y_{1i}=1[ y_{1i}^{\ast }>0].
\end{equation*}

A reduced form, or first stage, regression equation relates the endogenous
explanatory variable $y_{2i}$\ {to a $k_z$-dimensional vector of valid instrumental variables, $z_i$,
and the explanatory variables $x_{i}$}:
\begin{equation}
y_{2i}=\pi \left( x_{i},z_{i}\right) +v_{i},\;\E[v_{i}\mid x_{i},z_{i}
] =0  \label{first}.
\end{equation}
\begin{remark}
\normalfont 	
While we have chosen {to view the reduced form}
regression equation (\ref{first}) as the specification of a conditional
expectation, we could alternatively follow the quasi-LIML estimation
approach of \cite{wooldridge2014quasi}. In his approach, the
reduced form regression equation is only required to be a linear projection
of $y_{2i}$ onto $x_{i}$ and $z_{i}$. We will always assume that $%
x_{i}$ includes a constant, so that the reduced form error term $v_{i}$ has
a zero mean. That is, instead of (\ref{first}), we could have assumed
\begin{equation}
y_{2i}=x_{i}^{\prime }\pi +z_{i}^{\prime }\xi +v_{i},\;\E[ v_{i}] =0
\label{sec},
\text{ with }
\text{Cov}\left( \left[
\begin{array}{c}
x_{i} \\
z_{i}%
\end{array}%
\right] ,v_{i}\right) =0 .
\end{equation}
\end{remark}
\begin{remark}
\normalfont
As noted by \cite{blundell2004endogeneity}, the reduced form error term $v_{i}$
often appears to be conditionally heteroskedastic. Taking this possibility
into account will allow us to devise more efficient estimators when the
reduced form error term is deduced from a conditional expectation rather
then from only a linear projection. We will actually combine the advantages
of both approaches (\ref{first}) and (\ref{sec}) by assuming
that:%
\begin{equation}
y_{2i}=x_{i}^{\prime }\pi +z_{i}^{\prime }\xi +v_{i},\;\E[v_{i}\left\vert
x_{i},z_{i}\right] =0  \label{firstlin}
\end{equation}
{However, it must be acknowledged that the linearity assumption for the
conditional expectation is restrictive, and prevents us from considering cases} where the endogenous explanatory variable $y_{2i}$ is itself
qualitative.\footnote{{We also note that, while \citet{blundell2004endogeneity} propose a nonparametric
	estimator of the possibly nonlinear regression function $\pi \left(
	x_{i},z_{i}\right) $, a given nonlinear parametric form of this regression
	function would not result either in a significant change in our proposed methodology.}}
\end{remark}

{As stressed by \cite{newey1999nonparametric}, the CF approach does not assume that $%
x_{i} $ and $z_{i}$ are valid instruments, in that the approach does not require}
\begin{equation}
\E[u_{i}\left\vert x_{i},z_{i}\right] =0  \label{exo},
\end{equation}but instead only that
\begin{equation}
\E[u_{i}\left\vert v_{i},x_{i},z_{i}\right] =\E[u_{i}\left\vert v_{i}\right]
\label{control}.
\end{equation}
Moreover, it is worth realizing that {neither equation (\ref{exo}) or equation (%
\ref{control}) implies the other.} While we will {eventually
maintain a stronger version of equation (\ref{control})}, i.e., $u_i$ conditionally independent of $x_i,z_i$ given $v_i$, there is no
reason to believe that $v_{i}$ is itself independent of $x_{i},z_{i}$ ,
which jointly with the former conditional independence would be tantamount
to joint independence of $(u_{i},v_{i})$ and $(x_{i},z_{i})$, and would in turn imply (\ref{exo}). {In particular, such independence would rule out the possibility of conditional heteroskedasticity for the error term
$v_{i}$ in the reduced form regression equation (\ref{firstlin}).}

As clearly defined by \citet{wooldridge2015control}, ``a control function is a variable
that, when added to a regression, renders a policy variable appropriately
exogenous." {Typically, the restriction in (\ref{control})} allows us to rewrite equation (\ref{struct})
as
\begin{equation}
y_{1i}^{\ast }=g\left( y_{2i},x_{i}\right) +\E[u_{i}\left\vert v_{i}\right]
+\varepsilon_{i}  \label{controlreg},
\end{equation}
where
\[
\varepsilon_{i}=y_{1i}^{\ast }-\E\left[y_{1i}^{\ast }\mid v_{i},x_{i},z_{i}
\right] =u_{i}-\E[u_{i}\left\vert v_{i}\right],
\]
{which ensures, by definition, that the policy variable is appropriately exogenous; i.e.,}
\[
\E[\varepsilon_{i}\left\vert y_{2i},x_{i},v_{i}\right] =0 .
\]

In their seminal work, {\cite{rivers1988limited} note that the only assumption needed to obtain valid inference in the probit model} is that the
conditional distribution of $u_{i}$ given $v_{i}$ is
normal with a mean that is linear in $v_{i}$ and with a fixed variance. While this condition is satisfied if $(u_{i},v_{i})$ is jointly normal, joint
normality is not required in general. Similarly, for general discrete choice models, a CF approach can be constructed by assuming that $\E\left[u_{i}\mid v_{i}%
\right] $ is linear in $v_{i}$\ and that $\varepsilon_{i}=u_{i}-\E\left[u_{i}\mid v_{i}%
\right] $ \ is independent of $v_{i}$, along with an assumption that $\varepsilon_i$ has a known continuous cumulative distribution function denoted by $\Phi $. We assume that this probability distribution is symmetric, i.e., $\Phi(\varepsilon)=1-\Phi(-\varepsilon)$, which, together with \eqref{control}, allows us to write
\begin{flalign*}
	\Pr \left[y_{1i}=1 \mid v_{i},x_{i},z_{i}\right] &=\Pr \left\{\varepsilon
	_{i}>-g\left( y_{2i},x_{i}\right) -\E\left[u_{i}\mid v_{i}\right]
	\mid v_{i},x_{i},z_{i}\right\} \\
	&=\Phi \left\{ g\left( y_{2i},x_{i}\right) +\E[u_{i}\left\vert v_{i}\right] %
	\right\}.
\end{flalign*}

We now collect the maintained assumptions on the general model in \eqref{struct}-\eqref{first}.

\medskip

\noindent\textbf{Assumption 1:} The following conditions are satisfied. \medskip

\noindent\textbf{(A.1)} (\textit{Observation scheme}) The observed data $\left\{
s_{i}\right\} _{i=1}^{n}=\left\{ (y_{1i},y_{2i},x_{i}',z_{i}')'\right\}
_{i=1}^{n}$ are from an i.i.d. sample and for some $\kappa >0,\E\left[ \left\Vert
s_{i}\right\Vert ^{2+\kappa }\right] <\infty .$
\medskip

\noindent(\textbf{A.2}) (\textit{Reduced form regression}): $
y_{2i}=\pi \left( x_{i},z_{i}\right) +v_{i},$ and $\E[v_{i}\mid x_{i},z_{i} ] =0.$
\medskip

\noindent(\textbf{A.3}) (\textit{Structural equation}): (i) $\E[u_{i}\left\vert v_{i},x_{i},z_{i}\right] =\E[u_{i}\left\vert v_{i}\right]$; (ii) $\Phi$ is a known cumulative distribution
function, twice continuously differentiable and strictly increasing, such that $\Phi (\varepsilon)=1-\Phi (-\varepsilon)$; and (iii) for some unknown parameter $\tilde{\rho}\in\mathbb{R}$, $$
\Pr [y_{1i}=1\mid v_{i},x_{i},z_{i}] =\Phi [ g\left(
y_{2i},x_{i}\right) +\tilde{\rho}v_{i}].$$

\noindent(\textbf{A.4}) (\textit{Linearity}): The unknown functions $g(\cdot,\cdot)$ and $\pi(\cdot,\cdot)$ are linear:
\begin{itemize}
	\item[(i)] For unknown parameters $\alpha\in\mathbb{R}$ and $\beta\in\mathbb{R}^{k_x}$, $g\left( y_{2i},x_{i}\right) =\alpha y_{2i}+x_{i}^{\prime }\beta$;
	\item[(ii)] For unknown parameters $\pi\in\mathbb{R}^{k_x}$ and $\xi\in\mathbb{R}^{k_z}$, $\pi \left( x_{i},z_{i}\right)=x_{i}^{\prime }\pi +z_{i}^{\prime }\xi $.
\end{itemize}
\medskip

\noindent(\textbf{A.5})(\textit{Parameters}) The unknown parameters $\theta=(\theta_1',\theta_2')'$, where $\theta_1:=(\tilde{\rho},\alpha,\beta')'$ and $\theta_2:=(\pi',\xi')'$, are of dimension $p=2+2k_x+k_z$. We have $\theta_1\in\Theta_1\subset\mathbb{R}^{k_x+2}$, $\theta_2\in\Theta_2\subset\mathbb{R}^{k_x+k_z}$, $\Theta:=\Theta_1\times\Theta_2$ and $\Theta$ is compact.  For $\theta^0$ denoting the unknown true value of $\theta$, we have $\theta^0\in\text{Int}(\Theta)$.

\medskip

{As already mentioned, the linearity in \textbf{Assumption (A.4)} is innocuous and what follows can be extended to settings where $g\left(
y_{2i},x_{i}\right) $ has any parametric single-index structure and to cases where $\pi \left( x_{i},z_{i}\right) $ has any parametric form.} In the more general nonparametric setting, \cite{newey1999nonparametric} demonstrate that identification by {CF of the structural model is tantamount} to assuming that there is no
functional relationship between the random variables $y_{2i},x_{i}$ and $v_{i}$
(see \citealt{newey1999nonparametric} for a precise definition of this concept). With a
linear structural function $g\left( y_{2i},x_{i}\right) $, identification
of the {structural parameter $\alpha $\ is equivalent to assuming that $y_{2i}$ is not a linear combination of $x_{i}$ and $%
v_{i}$, meaning that the reduced form regression depends on $z_{i}$, i.e., $\xi \neq 0$.}

To give a more concise treatment, throughout the remainder we restrict our
analysis to the case where $\Phi $ is the CDF of the standard normal
distribution and refer to the model:%
\begin{equation*}
\Pr [y_{1i}=1\left\vert v_{i},x_{i},z_{i}\right] =\Phi \left[ \alpha
y_{2i}+x_{i}^{\prime }\beta +\tilde{\rho}v_{i}\right]
\end{equation*}
as a probit model. Since only the sign of the latent
variable $y_{1i}^{\ast }$ is observed, the probit model generally requires the normalization condition $\text{Var}(u_{i})=1$. {However, it is without loss of generality to instead consider the normalization condition}
\begin{equation*}
\text{Var}[u_{i}\left\vert v_{i}\right] =\text{Var}\left(\varepsilon_{i}\right) =1.
\end{equation*}
If $\rho $ denotes the linear correlation coefficient between $%
u_{i}$ and $v_{i}$, the above normalization ensures that
\[
\text{Var}\left( u_{i}\right) =\tilde{\rho}^{2}\text{Var}\left( v_{i}\right) +1=\rho
^{2}\text{Var}\left( u_{i}\right) +1,
\]where  $\sigma _{v}=\sqrt{\text{Var}\left( v_{i}\right) }$,
\[
\text{Var}\left( u_{i}\right) =\frac{1}{1-\rho ^{2}},\;\tilde{\rho}=\frac{\rho }{%
	\sigma _{v}\sqrt{1-\rho ^{2}}} ,
\]and where we have that $\tilde{\rho}$\ is monotonic in $\rho $. Of course, the {simultaneity/endogeneity} problem is in evidence if and only if $\rho \neq 0$ or
equivalently $\tilde{\rho}\neq 0$.

\subsection{Estimating Equations}\label{sec:2.2}
Throughout the remainder, we partition the parameter vector as $\theta=(\theta_1',\theta_2')'$, where
\begin{equation*}
\theta _{1} =\left( \tilde{\rho},\alpha ,\beta ^{\prime }\right) ^{\prime
},\quad\theta _{2}=\left( \pi ^{\prime },\xi ^{\prime }\right) ^{\prime }.
\end{equation*}
The vector $\theta _{1}$ (resp., $\theta _{2}$) represents the
vector of structural (resp., reduced-from) parameters. Following \textbf{Assumption 1}, the {true value of the reduced form parameters $\theta _{2}$ is defined by} the conditional moment restrictions
\begin{equation}
\E[r_{2i}\left( \theta _{2}\right) \left\vert x_{i},z_{i}\right] =0,\text{ where }
r_{2i}\left( \theta _{2}\right) =y_{2i}-x_{i}^{\prime }\pi -z_{i}^{\prime
}\xi \label{CMRR}.
\end{equation}
{For fixed $\theta _{2}$, the true value of the structural parameters $\theta
_{1}$\ is} defined by the conditional moment restrictions
\begin{equation}
\E[r_{1i}\left( \theta _{1},\theta _{2}\right) \left\vert
y_{2i},x_{i},z_{i} \right] =0,\text{ where }	r_{1i}\left( \theta _{1},\theta _{2}\right)=y_{1i}-\Phi \left[ \alpha
y_{2i}+x_{i}^{\prime }\beta +\tilde{\rho}v_{i}\left( \theta _{2}\right) %
\right]  \label{CMRS},
\end{equation}and where $$v_{i}\left( \theta _{2}\right) =r_{2i}\left( \theta _{2}\right)
=y_{2i}-x_{i}^{\prime }\pi -z_{i}^{\prime }\xi.$$
As usual, we will handle conditional moment restrictions by choosing vectors
of instrumental functions, denoted respectively as $\tilde{b}\left(
x_{i},z_{i}\right) $\ for (\ref{CMRR}) and $\tilde{a}\left(
y_{2i},x_{i},z_{i}\right) $ for (\ref{CMRS}), where it is assumed that the moments $\E[\|\tilde{a}( y_{2i}, x_{i},z_{i})\|^{2+\kappa}]$ and $\E[\|\tilde{b}( x_{i},z_{i})\|^{2+\kappa}]$ are finite for some $\kappa>0$. For a given choice of instrumental
functions $\tilde{a}\left( .,.,.\right) $\ and $\tilde{b}\left( .,.\right) $, we maintain the following identification assumption.

\medskip

\noindent\textbf{Assumption 2} (\textit{Identification}): The true unknown value $\theta ^{0}=(\theta_1^{0'},\theta_2^{0'})'\in\text{Int}(\Theta)$ is the
unique solution $\theta \in \Theta $ to the following moment restrictions:%
\begin{alignat*}{4}
\text{\textbf{Reduced form}:}\quad&\quad	\E[ \tilde{b}( x_{i},z_{i}) r_{2i}( \theta _{2}) %
] =0&\iff&\quad\theta_2=\theta^0_2,\\
\text{\textbf{Structural}:}\quad&\quad		\E[ \tilde{a}( y_{2i},x_{i},z_{i}) r_{1i}( \theta
	_{1},\theta _{2}^0) ] =0\quad&\iff&\quad\theta_1=\theta^0_1 .
\end{alignat*}

We can summarize the unconditional moment conditions in \textbf{Assumption 2} as follows: for $H\ge p$, and $H$-dimensional vectors $a_{i}$ and $b_{i}$ of the same dimension, define
\begin{equation*}
g_{i}(\theta ) =a_{i}r_{1i}\left( \theta _{1},\theta _{2}\right)
+b_{i}r_{2i}\left( \theta _{2}\right),\text{ where }
a_{i} =\begin{bmatrix}
\tilde{a}\left( y_{2i},x_{i},z_{i}\right) \\
\mathbf{0}%
\end{bmatrix} ,\quad b_{i}=\begin{bmatrix}
\mathbf{0} \\
\tilde{b}\left( x_{i},z_{i}\right)%
\end{bmatrix} ,
\end{equation*}
then \textbf{Assumption 2} implies that the moment function $g_i(\theta)$ satisfies
\begin{eqnarray*}
\E\left[ g_{i}(\theta )\right] =0\Longleftrightarrow \theta =\theta ^{0}.
\end{eqnarray*}
A GMM estimator of $\theta^0$ can then be constructed using the moment function
\begin{equation}\label{eq:mom}
g_i(\theta)=(g_{1i}(\theta)',g_{2i}(\theta)')',\text{ where }g_{1i}(\theta)=\tilde{a}\left( y_{2i},x_{i},z_{i}\right) r_{1i}\left( \theta \right),\;\; g_{2i}(\theta)=\tilde{b}\left( x_{i},z_{i}\right) r_{2i}\left( \theta _{2}\right).
\end{equation}
In particular, for $W_n$ a sequence of positive-definite $H\times H$ weighting matrix, we can estimate $\theta^0$ using the GMM estimator
$$
\hat\theta_n=\arg\min_{\theta\in\Theta}\bar{g}_n(\theta)'W_{n}\bar{g}_n(\theta), \text{ where }\bar{g}_{n}(\theta )=\frac{1}{n}\sum_{i=1}^{n}g_i(\theta)\equiv \begin{pmatrix}
\bar{g}_{1n}(\theta)'&\bar{g}_{2n}(\theta)'
\end{pmatrix}'.
$$

\begin{remark}
\normalfont {In general, imposing that some components of the vectors $%
a_{i}$ and $b_{i}$ are zero prevents us from choosing optimal instruments, and ultimately results in $\hat\theta_n$ being an inefficient estimator of $\theta^0$.} The characterization of optimal
instrumental functions for the joint set (\ref{CMRR}) and (\ref{CMRS}) of
conditional moment restrictions is non-standard because they correspond to
different conditioning variables. The optimal instrumental functions in this
case have been characterized by \cite{kawaguchi2017moment} (see
also \citet{ai2003efficient} for a general study). Their result implies that in
case of overidentification and simultaneity ($\tilde{\rho}\neq 0$), the
first set $r_{1i}(\theta)$\ of moment conditions is also informative about $\theta
_{2}$, so that a more efficient estimator of $\theta _{2}$ (and in turn $%
\theta _{1}$) is obtained by an appropriate choice of $a_{i}$ in which all of its components are non-zero.
\end{remark}

{While the specific choice of instrumental functions $a_{i}$ and $b_{i}$ may be sub-optimal, this choice allows us to demonstrate the equivalence between a GMM-based approach and the 2SCML approach of \cite{rivers1988limited}. In particular, for $g_{1i}(\theta)$ and $g_{2i}(\theta)$ defined as in equation \eqref{eq:mom}}, we have that
\begin{flalign*}
\text{Cov}\left[ g_{1i}(\theta ^{0}),g_{2i}(\theta ^{0})\right]
&=\E\left[ \tilde{a}\left( y_{2i},x_{i},z_{i}\right) \tilde{b}^{\prime
}\left( x_{i},z_{i}\right) r_{1i}\left( \theta^0 \right) r_{2i}\left( \theta^0_{2}\right) \right]\\&=\E\left\{ \tilde{a}\left( y_{2i},x_{i},z_{i}\right) \tilde{b}^{\prime
}\left( x_{i},z_{i}\right) r_{2i}\left( \theta _{2}^{0}\right)
\E[r_{1i}\left( \theta ^{0}\right) \left\vert y_{2i},x_{i},z_{i}\right]
\right\} =0.
\end{flalign*}
Thus, an efficient GMM estimator based on the moment functions in \eqref{eq:mom} can be defined as
\begin{eqnarray*}
\hat{\theta}_{n} &=&\arg \min_{\theta \in \Theta }\bar{g}_{n}(\theta
)^{\prime }\left[
\begin{array}{cc}
W_{1n} & 0 \\
0 & W_{2n}%
\end{array}%
\right] \bar{g}_{n}(\theta ) \\
&=&\arg \min_{\theta \in \Theta }\left\{ \bar{g}_{1n}(\theta )^{\prime
}W_{1n}\bar{g}_{1n}(\theta )+\bar{g}_{2n}(\theta )^{\prime }W_{2n}\bar{g}%
_{2n}(\theta )\right\},  \nonumber
\end{eqnarray*}for an appropriate choice of the weighting matrices $W_{1n}$ and $W_{2n}$. Consequently, the components of the first-order conditions for the structural parameters $\theta_1$ are given by
\begin{equation}
\frac{\partial \bar{g}_{1n}(\hat{\theta}_{n})^{\prime }}{\partial \theta
	_{1}}W_{1n}\bar{g}_{1n}(\hat{\theta}_{n})=0  \label{absfoc}.
\end{equation}

Equation (\ref{absfoc}) allows us to see the estimator $\hat{\theta}_{1n}$ as a two-step estimator based on the moment
conditions
\begin{equation}
\E[r_{1i}\left( \theta _{1},\theta _{2}^{0}\right) \left\vert
y_{2i},x_{i},z_{i}\right] =0  \label{structmom},
\end{equation}
where the nuisance parameter $\theta _{2}^{0}$\ is replaced by a consistent first-step
estimator $\hat{\theta}_{2n}$. From \eqref{absfoc}, we can see that the estimator $\hat{\theta}_{1n}$ is
the solution in $\theta _{1}=\left( \tilde{\rho},\alpha ,\beta
^{\prime }\right) ^{\prime }$\ to the $\left( 2+k_{x}\right) $
orthogonality conditions
\begin{equation}
\sum_{i=1}^{n}\gamma _{i,n}\left\{ y_{1i}-\Phi \left[ \alpha
y_{2i}+x_{i}^{\prime }\beta +\tilde{\rho}v_{i}\left( \hat{\theta}%
_{2n}\right) \right] \right\} =0,\text{ for }\gamma _{i,n}=\frac{\partial \bar{g}_{1n}(\hat{\theta}_{n})^{\prime }}{%
	\partial \theta _{1}}W_{1n}\tilde{a}\left( y_{2i},x_{i},z_{i}\right) .
  \label{orthocond}
\end{equation}

The optimal instruments associated with estimation of $\theta_1^0$ in equation \eqref{structmom} (i.e., where $\theta^0_2$ is known) are given by any consistent estimator of:%
\begin{eqnarray}
\gamma _{i}^{\ast } =\left[ \text{Var}(r_{1i}\left( \theta _{1}^0,\theta
_{2}^{0}\right) \left\vert y_{2i},x_{i},z_{i}\right) \right] ^{-1}\E\left[
\frac{\partial r_{1i}\left( \theta^0 _{1},\theta _{2}^{0}\right) }{\partial
	\theta _{1}}\mid y_{2i},x_{i},z_{i} \right]\equiv \frac{\phi _{i}\left( \theta ^{0}\right) }{\Phi _{i}\left( \theta
	^{0}\right) \left[ 1-\Phi _{i}\left( \theta ^{0}\right) \right] }\left[
\begin{array}{c}
v_{i}\left( \theta _{2}^{0}\right) \\
y_{2i} \\
x_{i}%
\end{array}%
\right]  \nonumber
\end{eqnarray}
where
\begin{eqnarray*}
	\Phi _{i}\left( \theta ^{0}\right) &=&\Phi \left[ \alpha
	^{0}y_{2i}+x_{i}^{\prime }\beta ^{0}+\tilde{\rho}^{0}v_{i}\left( \theta
	_{2}^{0}\right) \right] \\
	\phi _{i}\left( \theta ^{0}\right) &=&\phi \left[ \alpha
	^{0}y_{2i}+x_{i}^{\prime }\beta ^{0}+\tilde{\rho}^{0}v_{i}\left( \theta
	_{2}^{0}\right) \right]
\end{eqnarray*}
and $\phi \left( x\right) =d\Phi\left(x\right)/dx $ is the
probability density function associated to $\Phi $.

Therefore, if one were to choose a consistent estimator of $\gamma_i^*$ as instruments, the estimator $%
\hat{\theta}_{1n}$ can be seen as the solution in $\theta _{1}=\left( \tilde{%
	\rho},\alpha ,\beta ^{\prime }\right) ^{\prime }$\ to the equations:%
\begin{equation}
\sum_{i=1}^{n}\frac{\phi _{i}\left( \theta _{1},\hat{\theta}_{2n}\right)
}{\Phi _{i}\left( \theta _{1},\hat{\theta}_{2n}\right) \left[ 1-\Phi
	_{i}\left( \theta _{1},\hat{\theta}_{2n}\right) \right] }\left[
\begin{array}{c}
v_{i}\left( \hat{\theta}_{2n}\right) \\
y_{2i} \\
x_{i}%
\end{array}%
\right] \left\{ y_{1i}-\Phi \left[ \alpha y_{2i}+x_{i}^{\prime }\beta +%
\tilde{\rho}v_{i}\left( \hat{\theta}_{2n}\right) \right] \right\} =0.
\label{2SCML}
\end{equation}
Equation (\ref{2SCML}) shows that, for any choice of a consistent first-step estimator $%
\hat{\theta}_{2n}$, the estimator $\hat{\theta}_{1n}$\ is a
2SCML estimator a la \cite{rivers1988limited}.

\subsection{The Weak IV Issue in the Probit Model}\label{sec:2.3}
The representation in equation \eqref{2SCML} demonstrates that the general class of GMM estimators for $\theta_1$ defined in equation (\ref{orthocond}) contains both 2SCML and Quasi-LIML estimators as particular cases. Therefore, we can ascertain the impact of instrument weakness, on these and related methods, by studying instrument weakness in this general class of GMM estimators.

However, before moving to a general study, we give some intuition on the potential impacts of instrument weakness in the case of probit model. These implications are most easily elucidated in the infeasible case where we replace the optimal instruments in equation \eqref{2SCML} with their infeasible counterpart $\gamma_i^\ast$, and where we replace the estimator $\hat\theta_{2n}$ by the true value $\theta^0_2$.

Under these simplification, and under the one-to-one transformation of $\theta_1$ defined by
$$
\eta _{1}=\tilde{\rho},\quad\eta _{2}=\alpha +\tilde{\rho},\quad\eta _{3}^{}=\beta -%
\tilde{\rho}{\pi}^0,
$$the infeasible estimator $\tilde\eta_n$ of $\eta^0$ (and thus $\theta_1^0$) can be defined as the solution to
$$
\sum_{i=1}^{n}\gamma_i^* \left\{ y_{1i}-\Phi \left[ \eta _{1}\left( -z_{i}^{\prime }{\xi}^0\right) +\eta _{2}y_{2i}+x_{i}^{\prime }\eta _{3}^{}\right] \right\}
=\sum_{i=1}^{n}w_i D_i\left\{ y_{1i}-\Phi \left[ \eta _{1}\left( -z_{i}^{\prime }{\xi}^0\right) +\eta _{2}y_{2i}+x_{i}^{\prime }\eta _{3}^{}\right] \right\}=0,
$$where $\gamma_i^*=w_i D_i$,  $w_i={1}/{\Phi_i(\theta^0)[1-\Phi_i(\theta^0)]}$ and $D_i=\phi_i(\theta^0)(-z_i'\xi^0,y_{2i},x_i')'$.\footnote{The simplification made in the term $D_i$, i.e., replacing $v_i(\theta^0_2)$ by $-z_i'\xi^0$, follows from the row operation on $\gamma^*_i$ which does not affect the solution of the linear equations in \eqref{2SCML} asymptotically.}
A Taylor expansion allows us to heuristically write
\begin{flalign*}
y_{1i}&-\Phi \left[ \eta _{1}\left( -z_{i}^{\prime }{\xi}^{0}\right)
+\eta _{2}y_{2i}+x_{i}^{\prime }\eta _{3}^{}\right] \nonumber \\
&\approx y_{1i}-\Phi _{i}\left( \theta ^{0}\right) -\phi _{i}\left(
\theta ^{0}\right) \left[ \left( -z_{i}^{\prime }\xi ^{0}\right) \left( \eta
_{1}-\eta _{1}^{0}\right) +y_{2i}\left( \eta _{2}-\eta _{2}^{0}\right) %
 +x_{i}^{\prime }\left( \eta _{3}-\eta _{3}^{0}\right)\right].
\end{flalign*}
Using this expansion within the infeasible estimating equations, $\tilde{\eta}_n$ can be seen to solve
$$
\sum_{i=1}^{n}w_i D_i\left( \tilde{y}_{1i}-D_i'\eta \right)=0, \text{ where }\tilde{y}_{1i} =y_{1i}-\Phi(\theta^0)+D_i'\eta^0.
$$

Consequently, $\tilde{\eta}_n$ is obtained from a weighted least squares regression of $\tilde{y}_{1i}$ on the explanatory variables $D_i= \phi_i(\theta^0)(-z_i'\xi^0,y_{2i},x_i')'$. While the above estimating equations are not identical to those in equation \eqref{2SCML}, it is clear from comparing the two that they are of a similar form, and therefore whatever implications are drawn about the later will be sustained by the former.

This regression-based viewpoint yields two important, and interrelated, implications for inference in endogenous binary choice models. First, the linear regression that is considered is not the one suggested by a linear probability model, which would be based on explanatory variables $z_{i}^{\prime }\xi ^{0},y_{2i}, x_{i}$, and not the weighted versions in $D_i$. Second, since the explanatory variables in the regression are weighted by $\phi_i(\theta^0)$, it is inappropriate to focus solely on the contribution of $z_i'\xi^0$ in the reduced form regression as a measure of instrument strength.

\begin{remark}
\normalfont Before moving on, we note that the above type of estimation approach has been dubbed ``two-stage residual
inclusion" (2SRI) estimation by \cite{terza2008two}. In particular, using the
first stage consistent estimators $\hat{\theta}_{2n}=(\hat{\pi}_{n}^{\prime
},\hat{\xi}_{n}^{\prime })^{\prime },$ the estimated first stage
residual
\[
\hat{v}_{i}=y_{2i}-x_{i}^{\prime }\hat{\pi}_{n}-z_{i}^{\prime }\hat{\xi}%
_{n}
\]
is included in the computation of the generalized residual
\[
r_{1i}\left( \theta _{1},\theta _{2}\right) =y_{1i}-\Phi \left[ \alpha
y_{2i}+x_{i}^{\prime }\beta +\tilde{\rho}v_i(\theta_2)\right].
\]
We know from \cite{hausman1978specification} that, in a fully linear model and as far as
estimation of structural parameters $\alpha $\ and $\beta $ is concerned,
2SRI is equivalent to 2SLS. The inclusion of the residual $\hat{v}_{i}$\ in
the regression equation ensures that naive OLS would coincide with 2SLS. In addition,
\cite{terza2008two} dub ``Two-stage predictor substitution" (2SPS)
the direct generalization of 2SLS to our nonlinear context, meaning that in
the structural equation, the endogenous variable is simply replaced by its
first stage adjusted value, leading to the generalized residual:%
\begin{eqnarray*}
	\hat{u}_{i} &=&y_{1i}-\Phi \left[ \alpha \hat{y}_{2i}+x_{i}^{\prime }\beta %
	\right] \\
	\hat{y}_{2i} &=&x_{i}^{\prime }\hat{\pi}_{n}+z_{i}^{\prime }\hat{\xi}_{n}
\end{eqnarray*}
Not surprisingly, \cite{terza2008two} show that in a nonlinear
model, 2SPS is not equivalent anymore to 2SRI and only the latter provides a
consistent estimator of structural parameters. The intuition is quite clear.
Due to the non-linearity of the function $\Phi \left( .\right) $, plugging
in $\hat{y}_{2i}$ to instrument $y_{2i}$ does not fix satisfactorily the
endogeneity bias problem.

\end{remark}

As alluded to above, it can be misleading to set the focus on the contribution of $z_{i}^{\prime
}\xi ^{0}$\ in the reduced form regression to gauge the instrument strength,
as is done when using the standard rule-of-thumb. Doing so is akin to
overlooking the impact of nonlinearity\ in the same way as that it is wrong to
confuse the correct 2SRI and the flawed 2SPS. Indeed, as the above arguments clarify, the relevant variable for capturing
	instrument strength is not $z_i$, as in the standard linear case, but $ \phi _{i}(\theta ^{0})z_{i}$. Thus, the assessment of identification strength should rather be based
	on the variability of $\phi _{i}(\theta ^{0})z_{i}^{\prime }\xi
	^{0}$.
	
We can easily illustrate the impact of moving from $z_i'\xi^0$ to $\phi(\theta^0)z_i'\xi^0$ in terms of instrument strength in the probit model, so that $\phi (\cdot)$ is the probability density function of the Gaussian
distribution.\footnote{The conclusions given below will remain valid for any other probability
distribution with thin tails, such that the variability of the
$\phi _{i}(\theta ^{0})z_{i}$ is drastically
different from the one of $z_{i}$.} First we recall that
that for a real valued variable $\nu $ and any given number $c$, the
absolute value of the function $h(\nu )=\nu \phi (c+\nu )$ is decreasing
in $\left\vert \nu \right\vert $\ when the latter value is larger than the
absolute value of the roots of the polynomial $\left[ 1-c\nu -\nu ^{2}\right]
$. Moreover, the rate of this decrease is sharp (converging swiftly to
zero) due to the thin tails of the Gaussian distribution.

Using this argument, one may realize that the multiplication of $z_{i}^{\prime }\xi ^{0}$ by%
\begin{equation*}
\phi _{i}(\theta ^{0})=\phi \left[ \alpha ^{0}y_{2i}+x_{i}^{\prime
}\beta ^{0}+\tilde{\rho}^{0}\left( y_{2i}-x_{i}^{\prime }\pi
^{0}-z_{i}^{\prime }\xi ^{0}\right) \right]
\end{equation*}
erases the variability of $z_{i}^{\prime }\xi ^{0}$, by
pruning all its large values. For $Z\sim \mathcal{N}\left( 0,\sigma _z^{2}\right) $, it is useful to illustrate the above
point by comparing the variance of $Z\phi (1+Z)$ as a percentage of the variance of $Z$. For various values of $\sigma
_z^{2}$ , we collect these ratios in Table \ref{tab:one} below.
\begin{table}[htbp]
	\begin{center}
		\caption{Comparison of the variance of $Z$ to the variance of $W=Z\phi(1+Z)$}
		\begin{tabular}{lcccccccc}
			\hline\hline
			&  $\sigma^2_z$     & 1    &2     & 5    & 10    & 50 & 100 \\\hline
			& Rel. \%     & 100\% & 79.03\% & 30.18\%& 28.13\% & 7.42\% &3.83\%\\\hline\hline
		\end{tabular}%
		\label{tab:one}%
	\end{center}
	\footnotesize
	Note: For $\sigma^2_w=\text{Var}(W)$, we first calculate $l_z=\sigma^2_w/\sigma^2_z$, i.e., the variance of $W$ as a percentage of the variance of $Z$, for various values of $\sigma^2_z$. The value of Rel \% in the table is the value of $l_z$ expressed as a percentage of  $\sigma^2_w/1$, i.e., we report the results relative to the case where $\sigma^2_z=1$.
\end{table}%

The results in Table \ref{tab:one} constitute compelling evidence on the likely flaws of
the standard rule-of-thumb in the
probit context. It is also worth stressing that, while Table \ref{tab:one} only
displays results with the normalized function $\phi (1+Z),$ the pruning
impact of large values of $z_{i}^{\prime }\xi ^{0}$\ within the function $%
\phi (\cdot)$ may actually be magnified in finite
sample by a large value of the parameter $\tilde{\rho}^{0}$. We may then
expect that the pruning effect documented in Table \ref{tab:one} will be even more
detrimental for small values of $\sigma _{v}$\ and/or a large degree of
endogeneity $\rho ,$ with both cases corresponding to a large
value of $\tilde{\rho}$. These possible perverse effects for the naive rule-of-thumb will be confirmed by the Monte
Carlo experiments in Section \ref{sectionMC}. These experiments will show that the
standard rule-of-thumb will be more prone to over-reject the null of weak
instruments in the case of strong simultaneity ($\rho $ close to one) and/or
a large signal to noise ratio $\sigma _{z}/\sigma _{v}$\ in the reduced form
regression.

\section{A Test for Instruments Weakness}

\subsection{Intuition}\label{sec:int}

Several authors, {such as} \cite{kleibergen2005testing}, \cite{caner2009testing}, \cite{chaudhuri2017score}, \cite{stock2000gmm}, and \cite{antoine2017testing}, have
discussed the advantages of a continuously updated GMM (CUGMM) approach to efficient GMM estimation
in case of possible weak identification. Following the latter two authors, in our context the advantage of the
CUGMM approach is that, irrespective of identification weakness, the asymptotic behavior of the CUGMM criterion is always
controlled. This feature of the CUGMM criterion will ultimately allow us to obtain a test for instrument weakness that is size controlled and consistent.

To see that this key feature remains true in our setting, recall the specific moment conditions underlying this analysis given by equation \eqref{eq:mom}; namely, for $\theta_1=(\tilde\rho,\alpha,\beta')'$ and $\theta_2=(\pi',\xi')'$, and $g_{1i}(\theta)=\tilde{a}(y_{2i},x_i,z_i)r_{1i}(\theta_1,\theta_2)$, $g_{2i}(\theta)=\tilde{b}(x_i,z_i)r_{2i}(\theta_2)$,
\[
g_{i}(\theta )=r_{1i}(\theta )a\left( y_{2i},x_{i},z_{i}\right)
+r_{2i}(\theta_2 )b\left( x_{i},z_{i}\right)=\begin{pmatrix} g_{1i}(\theta)',&g_{2i}(\theta)'\end{pmatrix}'.
\]
Defining the weighting matrix
\[
S_{n}(\theta )=\begin{bmatrix}S_{11,n}(\theta)&0\\0&S_{22,n}(\theta)\end{bmatrix},\;\;S_{jj,n}(\theta)=\frac{1}{n}\sum_{i=1}^{n}\left[g_{j,i}(\theta )-\bar{g}_{j,n}(\theta )\right]\left[g_{j,i}(\theta )-\bar{g}_{j,n}(\theta )\right]^{\prime
}, \;(j=1,2),
\]we consider a CUGMM estimator (hereafter, CUE) that takes into account the block diagonal structure of the population variance matrix. Then our CUE of $\theta^0$ based on $\bar{g}_n(\theta)=(\bar{g}_{1n}(\theta)',\bar{g}_{2n}(\theta)')'$ is defined as
\begin{eqnarray*}
	\hat{\theta}_{n} =\arg\min_{\theta\in\Theta }J_{n}(\theta,\theta ), \text{ for }
	J_{n}(\theta,\tilde\theta ) =n\bar{g}_{n}(\theta )^{\prime }S_{n}^{-1}(\tilde\theta )\bar{g}%
	_{n}(\theta ),
\end{eqnarray*}where the notation $J_n(\theta,\tilde\theta)$ differentiates the occurrences of $\theta$ in the moments, $\bar{g}_n(\theta)$, from those in the weighting matrix,  $S^{-1}_n(\tilde\theta)$.

{The critical feature of the criterion $J_{n}(\theta,\tilde\theta ) $ is that, by definition,}
\begin{equation}
J_{n}( \theta ^{0},\theta ^{0})\geq
J_{n}( \hat{\theta}_{n},\hat\theta_n),
\label{upper}
\end{equation}
while, since $\text{Cov}\left[g_{1i}(\theta^0),g_{2i}(\theta^0)\right]=0$, it follows that $J_{n}\left( \theta ^{0},\theta^0\right) $ {converges in {distribution to
a chi-square random variable with $H$ degrees of freedom, denoted throughout as $\chi ^{2}(H)$}}.\footnote{We note that a similar bound remains valid for a general CUGMM setup that does not make use of the block diagonal structure. For the reasons given previously, we focus on this more particular case.}

The general validity of this upper bound, {regardless of the instrument strength, and, hence consistency of $\hat{\theta}_n$,} is the reason why we resort to CUGMM. This upper bound will allow us to control the size of our test for weak identification.\footnote{The upper bound (\ref{upper}) is generally invalid if a first-step estimator of $\theta ^{0}${ is used to estimate the optimal instrumental functions.} The only way to incorporate optimal instrumental functions for $a(y_{2i},x_i,z_i)$ and $b(x_i,z_i)$ would be
to use them with a free value of $\theta $\ like in the weighting matrix of
CUGMM. The discussion of this alternative approach is left for future research. Also, we note that in the just identified case, the minimum $J_{n}\left(\hat{\theta}_{n},\hat\theta_n\right) $ of $J_{n}(\theta )$\ is asymptotically, with
probability one, equal to zero and $S_{n}^{-1}(\theta )
$\ is immaterial. In
particular, when using the first-order conditions of some M-estimator,
including two-stage conditional maximum likelihood or quasi-LIML, the
weighting matrix is irrelevant.}

The key intuition for our test of weak identification is the following observation. {Under weak identification, there are certain directions of the parameter space where the CUGMM objective function $J_{n}( \cdot,\hat\theta_n) $ is flat in the neighbourhood of $\hat{\theta}_{n}$. In these directions, if we distort $\hat{\theta}_{n}$ by some ``small'' value, say $\Delta_n\in\mathbb{R}^{p}$, and evaluate $J_n(\cdot,\hat\theta_n)$ at $\hat{\theta}^\delta_n=\hat{\theta}_n+\Delta_n$, then the value of $J_n(\hat\theta_n^\delta,\hat\theta_n)$ should not differ ``significantly'' from that of $J_n(\hat\theta_n,\hat\theta_n)$. Herein, the concept of} ``significance" means that
$J_{n}( \hat{\theta}_{n}^{\delta },\hat\theta_n) $ exceeds some {pre-specified quantile of the $\chi ^{2}(H)$ distribution.}

{Critically, however, since the objective function scales the squared norm of the sample mean $\bar{g}%
_{n}(\theta )$, by the factor $n$, when identification is not weak the distortion introduces a wedge between $\bar{g}_n(\hat{\theta}_{n}^{\delta })$\ and $\bar{g}_n(\hat{\theta}_{n})$. Therefore, if identification is not weak, so long as the distortion goes to zero sufficiently slowly with $n$, the criterion $J_{n}( \hat{\theta}_{n}^{\delta },\hat\theta_n) $ diverges asymptotically and thus exceeds (with
probability going to one) the chosen quantile of the $\chi ^{2}(H)$ distribution. Throughout the remainder, we refer to this testing procedure as a distorted J-test.}\footnote{It is worth noting that this test is dubbed the ``distorted J-test" because
	it uses the J statistic proposed by \cite{hansen1982large} in the overidentified case to test for the validity of a set of moments. The
	terminology is a bit misleading since our test may work even in the just
	identified case ($H=p$). There are actually two possible points of view: either one chooses to perform the distorted J-test test in a just identified setting ($H=p$), or in the overidentified setting ($H>p$).}

\subsection{The null hypothesis of weak identification}\label{sec:null}

As already discussed in Section \ref{sec:2.3}, {weak instruments impact estimation of the structural parameters through the structural moment function}
\[
g_{1i}\left( \theta \right) =\tilde{a}\left( y_{2i},x_{i},z_{i}\right)
r_{1i}\left( \theta _{1},\theta _{2}\right),
\text{ where }	r_{1i}\left( \theta _{1},\theta _{2}\right)=y_{1i}-\Phi \left[ (\tilde{\rho}+\alpha)
y_{2i}+x_{i}^{\prime }(\beta-\tilde{\rho}\pi) -\tilde{\rho}z_{i}^{\prime }\xi\right].
\]
{The impact of weak instruments can be most easily disentangled under the parameterization}
\begin{equation}
\eta =\left( \eta _{1},\eta _{2},\eta _{3}^{\prime }\right) ^{\prime
}=\left( \tilde{\rho},\tilde{\rho}+\alpha ,\beta ^{\prime }-\tilde{\rho}\pi
^{\prime }\right)^{\prime },\label{basis}
\end{equation}which allows us to restate the moment function as
$$
g_{1i}(\eta,\theta_2)=\tilde{a}(y_{2i},x_i,z_i)\tilde{r}_{1i}(\eta,\theta_2), \text{ where }\tilde{r}_{1i}\left( \eta
,\theta _{2}\right)=y_{1i}-\Phi \left[ -\eta _{1}z_{i}^{\prime }\xi +\eta
_{2}y_{2i}+x_{i}^{\prime }\eta _{3}\right].
$$

Following \cite{staiger1997instrumental} and \cite{stock2000gmm}, we use a drifting data generating process (DGP) to capture instrument weakness, so that population expectations are viewed as being $n$-dependent. However, to paraphrase \cite{lewbel2019identification}, we do not actually believe that the DGP is changing as $n$ changes, but use the drifting DGP concept in order to obtain more reliable asymptotic approximations in the context of weak identification. To this end, we consider that the population expectation of $\bar{g}_{1n}(\eta,\theta_2)$ is defined as
\[
m_{1n}\left( \eta, \theta _{2}\right) =\E_{n}\left[\sum_{i=1}^{n} \tilde{a}\left(
y_{2i},x_{i},z_{i}\right) \tilde{r}_{1i}\left( \eta, \theta _{2}\right) %
\right]/n.
\] {Under this drifting DGP, we are obliged to see $\theta^0_2$, and hence $\eta^0$, as $n$-dependent, so that the maintained identification assumption should technically be recast as $$m_{1n}\left( \eta,\theta _{2}\right) =0\iff (\eta,\theta_2)=(\eta^0_n,\theta_{2n}^0).$$ However, to keep the notational burned to a minimum, we only make the true-values dependence on $n$ explicit when absolutely necessary.}

{Following the approach of \cite{stock2000gmm} (see their Section 2.3), the following decomposition of  $m_{1n}\left( \eta, \theta _{2}\right)$ will ultimately allow us to isolate the impact of instrument weakness
\begin{eqnarray*}
m_{1n}\left( \eta, \theta _{2}^{0}\right) &=&m_{1n}\left( \eta^{0}, \theta
_{2}^{0}\right) +\left[ m_{1n}\left( \eta, \theta _{2}^{0}\right)
-m_{1n}\left( \eta _{1}^{0},\eta _{2},\eta_{3}, \theta _{2}^{0}\right) %
\right] \\&+&\left[ m_{1n}\left( \eta _{1}^{0},\eta _{2},\eta_{3}, \theta
_{2}^{0}\right) -m_{1n}\left( \eta^{0}, \theta _{2}^{0}\right) \right].
\end{eqnarray*} In particular, since $m_{1n}\left( \eta^{0}, \theta _{2}^{0}\right) =0$, we have}
\begin{equation}
m_{1n}\left( \eta, \theta _{2}^{0}\right) =\left[ m_{1n}\left( \eta
,\theta _{2}^{0}\right) -m_{1n}\left( \eta _{1}^{0},\eta _{2},\eta
_{3},\theta _{2}^{0}\right) \right] +m_{1n}\left( \eta _{1}^{0},\eta
_{2},\eta_{3}, \theta _{2}^{0}\right)  \label{decomp}.
\end{equation}
As explained in Section \ref{sec:2.3}, instrument weakness is encapsulated by the
explanatory variable $\phi_i \left( \theta ^{0}\right) z_{i}^{\prime }\xi^{0}$. {The impact of this explanatory variable on instrument strength can be directly obtained by linearising $m_{1n}\left( \eta, \theta _{2}^{0}\right) $ around $\eta_1^0$ to obtain}
\begin{eqnarray}
m_{1n}\left( \eta, \theta _{2}^{0}\right) -m_{1n}\left( \eta _{1}^{0},\eta
_{2},\eta_{3}, \theta _{2}^{0}\right) &=&\left( \eta _{1}-\eta
_{1}^{0}\right) \frac{\partial m_{1n}}{\partial \eta _{1}}\left( \eta
_{1n}^{\ast },\eta _{2},\eta_{3}, \theta _{2}^{0}\right)  \label{localSW} \\
&=&\left( \eta _{1}-\eta _{1}^{0}\right) \E_{n}\left[ \sum_{i=1}^{n}\tilde{a}\left(
y_{2i},x_{i},z_{i}\right) \phi _{i}\left( \eta _{1n}^{\ast },\eta
_{2},\eta_{3}, \theta _{2}^{0}\right) z_{i}^{\prime }\xi ^{0}\right]/n,
\nonumber
\end{eqnarray}
where $\eta _{1n}^{\ast }$ denotes a component-by-component intermediate value, which  can vary according to the components of the function $\tilde{a}(.)$.

Equation \eqref{localSW} allows us to write the decomposition in equation \eqref{decomp} in the following semi-separable form, which clearly partitions the directions of weakness in the parameter space:
for some real, positive, and deterministic sequence $\varsigma_{n}\rightarrow\infty$ as $n\rightarrow\infty$, with $\varsigma_{n}=O(\sqrt{n})$, possibly $o(\sqrt{n})$,
\begin{equation}
m_{1n}\left( \eta, \theta _{2}^{0}\right) ={q_{11,n}\left( \eta
	\right) }/{\varsigma_n}+q_{12,n}\left( \eta _{2},\eta _{3}\right),  \label{SW}
\end{equation}where
\begin{flalign*}
&q_{11,n}\left( \eta \right) =\varsigma_n\left[ m_{1n}\left( \eta, \theta
_{2}^{0}\right) -m_{1n}\left( \eta _{1}^{0},\eta _{2},\eta_{3}, \theta
_{2}^{0}\right) \right],\\
&q_{12,n}\left( \eta _{2},\eta _{3}\right) =m_{1n}\left( \eta _{1}^{0},\eta
_{2},\eta_{3}, \theta _{2}^{0}\right) =\sum_{i=1}^{n}\E_n\left[\tilde{a}\left(
y_{2i},x_{i},z_{i}\right) \tilde{r}_{1i}\left( \eta _{1}^{0},\eta
_{2},\eta_{3}, \theta _{2}^{0}\right) \right]/n.
\end{flalign*}
Given this decomposition of $m_{1n}\left( \eta, \theta _{2}^{0}\right)$, the identification strength of $\eta_1$ is entirely determined by equation \eqref{localSW} and therefore $q_{11,n}(\eta)/\varsigma_n$. In particular, the rate $\varsigma_n$ can be thought of as encapsulating the speed with which the curvature of the moments approaches zero in the $\eta_1$ direction, and thus $\varsigma_n$ determines the degree of identification weakness. If $\varsigma_n$ diverges like $\sqrt{n}$, the speed at which this curvature vanishes is matched by the rate at which information accumulates in the sample, i.e., $\sqrt{n}$, and there is no hope that $\eta^0_1$ can be identified from sample information; i.e., $\eta^0_1$ is weakly identified. In contrast, the identification of $\eta_2,\eta_3$ is determined by $q_{12,n}(\eta_2,\eta_3)$ and is not afflicted by identification weakness. That is, in this rotated parameter space of $\eta$, identification weakness only occurs in the $\eta_1$ direction and does not permeate the remaining directions in the parameter space. The representation in equation \eqref{SW} is conformable, but not equivalent, to the decomposition employed by \cite{stock2000gmm} to study the behavior of GMM under weak identification (see Remark \ref{remark_SW} for details). We maintain the following conditions on $m_{1n}(\eta,\theta^0_2)$, which has the same form as Assumption C in \cite{stock2000gmm}.

\medskip
\noindent\textbf{Assumption 3:} For $\varsigma_n=O(\sqrt{n})$, possibly $o(\sqrt{n})$, $m_{1n}\left( \eta, \theta _{2}^{0}\right) ={q_{11,n}\left( \eta
	\right)/\varsigma_{n}}+q_{12,n}\left( \eta _{2},\eta _{3}\right)$:

\noindent(i) $q_{11,n}\left( \eta \right)\rightarrow q_{11}\left( \eta \right)$ as $n\rightarrow\infty$ uniformly in $\eta $, where $q_{11}\left( \eta ^{0}\right) =0$, and $q_{11}(\cdot)$ is uniformly
continuous (and hence bounded) in $\eta$.
\medskip

\noindent (ii) $q_{12,n}\left( \eta_2,
\eta_3 \right)\rightarrow q_{12}\left( \eta_2,\eta_3 \right)$ as $n\rightarrow\infty$ uniformly in $\eta_2,\eta_3 $. For all $n\ge1$, $q_{12,n}\left( \eta _{2},\eta _{3}\right) $ satisfies $q_{12,n}\left( \eta _{2},\eta _{3}\right) =0\Longleftrightarrow \left( \eta
_{2},\eta _{3}\right) =\left( \eta _{2}^{0},\eta _{3}^{0}\right)$, and is continuously differentiable, with
${\partial q_{12,n}\left( \eta _{2},\eta _{3}\right)}/\partial (\eta _{2}, \eta _{3}^{\prime })'$ full column rank at $(\eta^0_2,\eta_3^{0'})'$.

\begin{remark}{
	\normalfont \textbf{Assumption 3(i)} is justified by the decomposition in equation (\ref{localSW}) and \textbf{Assumptions 1 and 2}. Secondly, we note that \textbf{Assumption 3} is  natural in our context. \textbf{Assumption 3(ii)} enforces that, for $q_{12,n}\left( \eta _{2},\eta _{3}\right)=\sum_{i=1}^{n}\E_n\left\{ \tilde{a}\left(
y_{2i},x_{i},z_{i}\right) \left[ y_{1i}-\Phi \left(-\eta^0
_{1}z_{i}^{\prime }\xi^0 +\eta _{2}y_{2i}+x_{i}^{\prime }\eta _{3}\right)
\right] \right\}/n$,
\begin{eqnarray*}
	-\frac{\partial q_{12,n}\left( \eta _{2},\eta _{3}\right) }{\partial (\eta _{2}, \eta
		_{3}^{\prime })'}=\frac{1}{n}\E_n\left\{ \sum_{i=1}^{n}\tilde{a}\left(
	y_{2i},x_{i},z_{i}\right) \phi _{i}\left( \eta_1^0,\eta_2,\eta_3, \theta _{2}^{0}\right) %
	( y_{2i}\;\vdots\; x_{i}^{\prime }) \right\}
\end{eqnarray*}has full column rank at $(\eta_2^0,\eta_3^{0'})'$. This is tightly related to the requirement that
the components of $( y_{2i}\;\vdots\; x_{i}^{\prime })$\ be
linearly independent, since they coincide with the explanatory variables of
the latent structural equation.}
\end{remark}

For the set,  $$\Upsilon(\theta^0_2):=\left\{\eta\in\mathbb{R}^{k_x+2}\;:\;\eta=(\tilde\rho,\alpha+\tilde\rho,\beta'-\tilde\rho\pi^{0'})',\text{ for some }\theta_1=(\tilde\rho,\alpha,\beta')'\in\Theta_1\right\},$$ we state the null hypothesis of weak identification as follows.
\medskip

\noindent\textbf{Null Hypothesis of Weak Identification:}
\begin{equation}
\text{H}_{0}\left(\varsigma_{n}=\sqrt{n}\right):\sup_{\eta\in\Upsilon(\theta_2^0) }\frac{1}{n}\left\Vert \E_{n}\left[ \sum_{i=1}^{n}\tilde{a}\left(
y_{2i},x_{i},z_{i}\right) \phi _{i}\left( \eta, \theta _{2}^{0}\right)
z_{i}^{\prime }\xi ^{0}\right] \right\Vert =O_{}\left( \frac{1}{\sqrt{n}}%
\right)  \label{nullweak}.
\end{equation}
\medskip

The set $\Upsilon(\theta^0_2)$ denotes the set of structural parameters under the parametrization in \eqref{basis}, and with $\theta_2=\theta^0_2$, so that the supremum over $\eta$ in (\ref{nullweak}) is akin to a supremum over the structural parameters $\theta _{1}$, given the true value $\theta _{2}^{0}$\ of the reduced form parameters. Both sets of structural
parameters, the initial one $\Theta _{1}$ and the reparameterized one $%
\Upsilon \left( \theta _{2}^{0}\right) $ are compact subsets of $
\mathbb{R}^{k_{x}+2}$. Based on the decomposition of \eqref{SW}, the identification strength of $\eta_1$ is determined by the rate $\varsigma_n$, and $\varsigma_n=O(\sqrt{n})$ implies that even asymptotically, the population objective function is nearly flat in $\eta_1$. Such asymptotic behavior of the objective function will lead to inconsistent estimation of $\eta^0_1$ in the rotated parameter space and for the structural parameter $\theta_1$ in the original parameter space $\Theta_1$.

\begin{remark}\label{remark_SW}\normalfont{	
It is worth noting that this
definition of weak identification is a generalization of \cite{stock2000gmm} since it is considered at the true value $\theta _{2}^{0}$\ of the
parameters of the reduced form regression equation. This must be seen as the
relevant extension of the concept of weak instruments for the context of
control variables. As explained in Section 2.3, the relevant explanatory
variables for the structural equation are ${\phi _{i}\left( \eta, \theta _{2}^{0}\right) ( z_{i}^{\prime }\xi
^{0},y_{2i},x_{i}^{\prime })^{\prime } }$.
In particular, it is the impact $z_{i}^{\prime }\xi ^{0}$, at
the true value $\xi ^{0}$, that
matters for identification and the
pruning effect of $\phi _{i}\left( \eta, \theta _{2}^{0}\right) $, also
at the true value $\theta _{2}^{0}=\left( \pi ^{0\prime },\xi ^{0\prime
}\right) ^{\prime }$. This extension is made possible by the reinforced
identification condition in \textbf{Assumption 2} (identification of $\theta^0 _{2}$ by the second
set of moment conditions in isolation) and the choice of block-diagonal
weighting matrix.}
\end{remark}
\subsection{A Distorted J-test (DJ test) for the Null of Weak Identification}\label{sec:3.3}

The decomposition in equation (\ref{SW}), along with \textbf{Assumption 3}, clarifies and confines the weak
identification issue, under the parametrization $\zeta=(\eta',\theta_2')'$, to the $\eta_1$ direction. Therefore, to construct a distorted testing approach for weak identification along the lines proposed in Section \ref{sec:int}, it is precisely this direction, and only this direction, that should be distorted. 

To this end, let $\mathcal{Z}$ denote the parameter space of $\zeta$ and define the infeasible CUE
$$
\hat\zeta_n=\argmin_{\zeta\in\mathcal{Z}}\bar{g}_n(\zeta)'S_n^{-1}(\zeta)\bar{g}_n(\zeta),
$$ and consider distorting the first component of $\hat\zeta_n$
as
\[
\hat{\zeta}_{n}^{\delta }=\hat{\zeta}_{n}+
\begin{bmatrix}
\delta _{n} &
0 &\dots&0
\end{bmatrix}'=[\hat{\tilde{\rho}}_n,\hat{\tilde\rho}_n+\hat\alpha_n,\hat\beta'_n-\hat{\tilde\rho}_n\pi^{0'},\hat\theta_{2n}^{\prime}]'+\begin{bmatrix}
\delta _{n}, &
0, &\dots&0
\end{bmatrix}'.
\]
Under the change of basis, 
 this is equivalent to distorting the CUE $\hat\theta_n$ as
\[
\left[
\begin{array}{c}
\hat{\theta}_{1n} \\
\hat{\theta}_{2n}%
\end{array}%
\right] +\left[
\begin{array}{c}
\Delta^0 _{1n} \\
\textbf{0}%
\end{array}%
\right] ,
\text{ where }
\Delta^0 _{1n}=
\begin{bmatrix}
\delta _{n}, &
-\delta _{n}, &
\delta _{n}\pi^{0}%
\end{bmatrix}' ,
\]which distorts the entire vector of structural parameters $\theta_1$. However, the above perturbation of $\hat\theta_n$ is infeasible as it depends on the unknown $\pi^0$. A feasible perturbation can be produced by replacing $\pi^0$ with its estimated value $\hat\pi_n$, which yields
\begin{equation}\label{eq:distort}
\hat\theta_n^\delta:=\left[
\begin{array}{c}
\hat{\theta}_{1n} \\
\hat{\theta}_{2n}%
\end{array}%
\right] +\left[
\begin{array}{c}
\Delta _{1n} \\
\textbf{0}%
\end{array}%
\right] ,
\text{ where }
\Delta _{1n}=
\begin{bmatrix}
\delta _{n},&
-\delta _{n},&
\delta _{n}\hat\pi_{n}%
\end{bmatrix}.
\end{equation}

As explained in Section \ref{sec:int}, under weak identification, if we distort
the CUE $\hat{\theta}_{n}$ by some small value in the
directions of weak identification, i.e., $\eta_1$, the value of the GMM criterion at $\hat{%
	\theta}_{n}^{\delta }$\ should not differ significantly from the criterion evaluated at $\hat{\theta}_{n}$. More precisely, recalling the definitions of $\bar{g}_n(\theta)$ and $S_n(\theta)$ given in Section \ref{sec:int},
\[
J_{n}( \theta,\tilde\theta ) =n\bar{g}_{n}( \theta ) ^{\prime
}S_{n}^{-1}( \tilde\theta ) \bar{g}_{n}( \theta ), \quad
J_{n}( \hat{\theta}_{n},\hat\theta_n) =\min_{\theta\in\Theta }J_{n}( \theta
,\theta),
\]we introduce the distorted J-test statistic:%
\[
J_{n}^{\delta }=n\bar{g}_{n}( \hat\theta^\delta_n ) ^{\prime
}S_{n}^{-1}( \hat\theta_n ) \bar{g}_{n}( \hat\theta^\delta_n ).
\]

To deduce the behavior of $J^\delta_n$ under the null of weak identification, we must maintain a regularity condition on the Jacobian of the moments. However, given that our null of weak identification is local about $\eta_1$, at the fixed value of $\theta_2^0$, we are only required to maintain the following assumption.\footnote{We note that Assumption 4 is guaranteed under Assumption 1 and a functional central limit theorem. See the proof of Lemma \ref{lem_uniform_deriv} in the Appendix for details. We state this result as an assumption to ease the comparison with standard results.}

\medskip

\noindent\textbf{Assumption 4:} Uniformly over $\Upsilon \left( \theta_{2}^{0}\right) $, $\sqrt{n}\left\{\partial \bar{g}_n(\eta,\theta^0_2)/\partial\eta_1-\E_n\left[\partial \bar{g}_n(\eta,\theta^0_2)/\partial\eta_1\right]\right\}\Rightarrow\Psi(\eta,\theta_2^0)$, for $\Psi(\eta,\theta_2^0)$ a mean-zero Gaussian process, and where $\Rightarrow$ denotes weak convergence in the sup-norm.

\begin{proposition}[Lack of Consistency]\label{prop1}
If \textbf{Assumptions 1-4} are satisfied, and if ${\E[\|\tilde{a}(y_{2i},x_i,z_i)z_i'\|^2]<\infty}$, then under the null of weak identification, for any $\delta _{n}=o(1)$,
\[
\plim_{n\rightarrow \infty }\sqrt{n}\left[ \bar{g}_{n}( \hat{\theta}%
_{n}^{\delta }) -\bar{g}_{n}( \hat{\theta}_{n}) \right] =0 .
\]In addition, if $\sup_{\theta \in \Theta }\left\Vert S_{n}^{-1}(
\theta ) \right\Vert =O_{p}(1)$, then
\[
\plim_{n\rightarrow \infty }\left[ J_{n}^{\delta }-J_{n}( \hat{\theta}%
_{n},\hat\theta_n) \right] =0.
\]
\end{proposition}
\textbf{Proposition \ref{prop1}} demonstrates that under the null of weak identification, the curvature of the objective function is insensitive to a small departure from the CUE, indicating the lack of consistency of $\hat{\theta}_{n}$. By adapting the general testing approach of \citet{antoine2017testing}, Proposition \ref{prop1} paves the way for a testing
strategy for weak instruments in discrete choice models. Recall that the number of model parameters is $p=2+2k_x+k_z$, and $H$ denotes the number of moments.

\begin{theorem}[Distorted J-test: Under the Null]\label{thm0}
Under \textbf{Assumptions 1-4} and the null of weak identification, for
any deterministic sequence $\delta _{n}=o(1)$, define the
distorted J-test by the rejection region:%
\[
W_{n}^{\delta }=\left\{ J_{n}^{\delta }>\chi _{1-\alpha }^{2}\left(
H+1-p\right) \right\},
\]
where $\chi _{1-\alpha }^{2}\left( H+1-p\right) $ is the $%
\left( 1-\alpha \right) $\ quantile of the Chi-square distribution with $(H+1-p)$ degrees of freedom. Under the null hypothesis of weak identification, $W_{n}^{\delta }$ {has asymptotic size of at most} $\alpha $.	
\end{theorem}

\medskip

{As discussed in Section \ref{sec:int}, the CUGMM framework allows us to control the size of our test by ensuring that we can obtain a convenient upper bound for $J_n^\delta$ under the null of weak
	identification. Since there is only a single direction of weakness in the rotated parameter space, this bound can be based on the $\chi^2(H+1-p)$ distribution; please see the proof of Theorem \ref{thm0} for details. While the test statistic $J_n^\delta$ coincides with the one given in Section \ref{sec:int}, we have improved the asymptotic power of the test $W_n^\delta$ by using a
critical value calculated from $\chi ^{2}\left( H+1-p\right) $ instead
of $\chi ^{2}\left( H\right) $. This power gain is obviously important since
we may be afraid that our test would be overly conservative.}

\subsection{Estimation and Testing Under the Alternative}\label{sec:alt}
In this section, we prove that $W_n^\delta$, the distorted J-test based on $J_n^\delta$, is consistent under the alternative. Before presenting this result, we first discuss the asymptotic behavior of the CUE under the alternative.

\subsubsection{Estimation Under the Alternative}
We first deduce the properties of the infeasible CUE for the rotated parameter $\zeta=(\eta',\theta_2')'$. The vector $\zeta$ represents the following change of basis in the parameter space: 
\begin{equation}\label{eq:basis}
\theta=R\zeta=R\begin{pmatrix}
\eta\\\theta_2
\end{pmatrix},\text{ where } R=\begin{pmatrix}
R_1&\textbf{O}\\\textbf{O}&\textbf{I}_{k_x+k_z}
\end{pmatrix},\;R_{1}=\begin{pmatrix}
1&0&0\\-1&1&0\\\pi^0&\mathbf{0}&\textbf{I}_{k_x}
\end{pmatrix}.\end{equation}
For $\mathcal{Z}$ denoting the parameter space of $\zeta$, the CUE of $\zeta^0$ is given by
\begin{equation*}
\hat\zeta_n =\argmin_{\zeta\in \mathcal{Z} }\bar{g}_{n}(R\zeta)^{\prime }S_{n}^{-1}(R\zeta)\bar{g}_{n}(R\zeta).\label{eq:rotparm}
\end{equation*}Once the asymptotic properties of $\hat\zeta_n$ have been deduced, the asymptotic behavior of $\hat\theta_n$ can be ascertained by applying the change of basis $\hat\theta_n=R\hat\zeta_n$ in equation \eqref{eq:basis}. 

To deduce the properties of $\hat\zeta_n$ under the alternative, we first recall that the null of weak identification, defined by \eqref{nullweak}, implies that
$$
\sup_{\eta\in\Upsilon(\theta^0_2)}\left\|\frac{1}{n}\E_{n}\left\{ \sum_{i=1}^{n}\frac{\partial g_i(\eta,\theta^0_2)}{\partial\eta_1}\right\}\right\|=\sup_{\eta\in\Upsilon(\theta^0_2)}\left\|\frac{1}{n}\E_{n}\left\{ \sum_{i=1}^{n}\tilde{a}
\left( y_{2i},x_{i},z_{i}\right) \phi _{i}\left( \eta ^{},\theta
_{2}^{0}\right) z_{i}^{\prime }\xi ^{0}\right\}\right\|=O(1/\sqrt{n}).
$$The alternative hypothesis to this null implies the existence of a deterministic sequence $\varsigma_{n}=o(\sqrt{n})$ such that
$$
\limsup_{n\rightarrow\infty}\sup_{\eta\in\Upsilon(\theta^0_2)}\left\|\frac{1}{n}\E_{n}\left\{ \sum_{i=1}^{n}\tilde{a}%
\left( y_{2i},x_{i},z_{i}\right) \phi _{i}\left( \eta ^{},\theta
_{2}^{0}\right) z_{i}^{\prime }\xi ^{0}\right\}\varsigma_{n}\right\|>0.
$$
To deduce the behavior of the CUE $\hat\zeta_n$ under the alternative, we slightly reinforce this condition as follows. \medskip

\noindent\textbf{Assumption 5:} Under the alternative hypothesis, there exists a deterministic sequence $\varsigma_{n}=o(\sqrt{n})$ and
a continuous, and deterministic vector function $V^0(\eta)$ such that, ${\inf_{\eta\in\Upsilon(\theta^0_2)}\|V^0(\eta)\|>0}$, and\footnote{We note here that $V^0(\eta)$ technically depends on the drifting pseudo-true value $\theta^0_2$, but subsume this dependence in the definition to simply notations.}
\[
\lim_{n\rightarrow\infty}\sup_{\eta\in\Upsilon(\theta^0_2)}\left\|\frac{1}{n}\E_{n}\left\{ \sum_{i=1}^{n}\tilde{a}%
\left( y_{2i},x_{i},z_{i}\right) \phi _{i}\left( \eta ^{},\theta
_{2}^{0}\right) z_{i}^{\prime }\xi ^{0}\right\} \varsigma _{n}-V^0(\eta)\right\|=0.
\]

\begin{remark}
	\normalfont
Even though \textbf{Assumption 5} arguably limits the scope of the
alternative hypothesis, it is more general than if we were to follow the approach of \cite{staiger1997instrumental} and characterize identification strength only through the reduced form regression equation. In the latter case, one would consider that the reduced form regression evolves according to the drifting DGP
\[
\E_{n}[y_{2i}\left\vert x_{i},z_{i}\right] =x_{i}^{\prime }\pi^0
+z_{i}^{\prime }\xi^0 _{n}.
\]Under the null of weak identification, we have that
$\xi^0 _{n}=O(1/\sqrt{n}).
$ In contrast, \textbf{Assumption 5} would require that, for some $\gamma^0\in\mathbb{R}^{k_z} $ with $\|\gamma^0\|>0$ and some $\varsigma_{n}=o(\sqrt{n})$,
$$\lim_{n\rightarrow \infty }\varsigma_{n}\xi ^0_{n}=\gamma^0,\;\text{ and }
V^0(\eta)=\E_n\left[\tilde{a}(y_{2i},x_i,z_i)\phi_i(\eta,\theta^0_2)z_i'\right]\gamma^0\neq0.
$$However, as explained in Section \ref{sec:2.3}, this approach to characterize
identification strength is not sufficient in our opinion, since it only accounts for the instrument strength in the reduced form regression, $\xi^0_n$, and does not account for the interactions between the instrumental function $\tilde{a}(y_{2i},x_i,z_i)$ and $\phi
_{i}\left( \eta,\theta _{2n}^{0}\right) z_{i}^{\prime }\xi _{n}^{0}$, which may result in the pruning of large realizations of the instruments via the behavior of $\phi_{i}\left( \eta,\theta _{2n}^{0}\right) $.
\end{remark}

By defining the alternative hypothesis using \textbf{Assumption 5}, we clearly partition the two possible cases for estimation of $\zeta^0$: (i) if identification is weak, $\hat{\zeta}_{n}$\ is not consistent (as implied by Proposition \ref{prop1}), nor are other commonly applied estimators such as 2SCML or Quasi-LIML estimators; (ii) when identification is not weak, $\hat{\zeta}_{n}$ is consistent.

\begin{proposition}[Consistency]\label{lem:cons}If \textbf{Assumptions 1-5} are satisfied, and if ${\sup_{\zeta\in \mathcal{Z}}\|S_n^{-1}(\zeta)\|=O_p(1)}$, then $\|\hat\zeta_n-\zeta^0\|=o_p(1)$.
\end{proposition}

The asymptotic distribution of $\hat\zeta_n$ depends on the behavior of the Jacobian for the moments. Under \textbf{Assumption 3 and 5}, the scaled Jacobian of the moment functions, as defined below in Lemma \ref{local-lem}, is full rank under the following mild assumption, which, if we take $\tilde{b}(x_i,z_i)=\left(x_i'\;:\;z_i'\right)'$ is nothing but the standard rank condition on the reduced form regression. \medskip

\noindent\textbf{Assumption 6:} For all $n\ge1$,  $\E_n[\tilde{b}(x_i,z_i)(x_i'\;:\;z_i')]$ has column rank $(k_x+k_z)=\text{dim}(\theta_2)$. \medskip

\begin{lemma}\label{local-lem}
 Under \textbf{Assumption 1-6}, for a given sequence $\varsigma_{n}=o(\sqrt{n})$, the matrix $$M=\plim_{n\rightarrow\infty}\left\{\frac{\partial\bar{g}_n(\zeta^0)}{\partial\zeta'}\right\}\Lambda_n,\text{ where }\Lambda_n=\begin{bmatrix}\varsigma_{n}&\textbf{O}_{p-1}\\\textbf{O}_{p-1}&\textbf{I}_{p-1}
\end{bmatrix},
$$exists and is full column rank.
\end{lemma}

Given the full-rank nature of the scaled Jacobian, we would expect the CUE to be asymptotically normal. In particular, under the alternative (as defined by \textbf{Assumptions 3 and 5}), we can then deduce the following result.
\begin{theorem}[Asymptotic Normality]\label{theoremN} If \textbf{Assumptions 1-6} are satisfied  then
	$$\sqrt{n}\Lambda_n^{-1}(\hat{\zeta}_n-\zeta^0)\overset{d}{\rightarrow}\mathcal{N}\left(0,[M'S^{-1}M]^{-1}\right),\text{ where }S:=\plim_{n\rightarrow \infty }S_n(\zeta^0).$$
\end{theorem}

 As expected, all entries of $\zeta$, save for $\eta _{1}$, are $\sqrt{n}$-consistent and
asymptotically normal GMM estimators. In contrast, the direction $\eta _{1}$ converges at the $\{\sqrt{n}/\varsigma_{n}\}$-rate, which is possibly slower than $\sqrt{n}$. Of course, our goal is not to conduct inference on $\zeta^0$, but on $\theta^0$. By the change of basis in \eqref{eq:basis}, $\theta=R\zeta$, and Theorem \ref{theoremN} implies that the feasible CUGMM
estimator $\hat{\theta}_{n}$ satisfies
\begin{equation}
\sqrt{n}\Lambda _{n}^{-1}R^{-1}(\hat{\theta}_{n}-\theta _{}^{0})
\overset{d}{\rightarrow}\mathcal{N} \left( 0,[M'S^{-1}M]^{-1}\right) .
\label{contamin}
\end{equation} Importantly, since the matrix $R$ is not diagonal, the slower rate of $\{\sqrt{n}/\varsigma_{n}\}$ pollutes the entire vector of structural parameters $\theta_{1}=(\tilde{\rho},\alpha,\beta')'$, which follows from the change of basis $\theta=R\zeta$. Therefore, all structural parameter estimates in the probit model converge at the slower $\{\sqrt{n}/\varsigma_n\}$-rate.

Equation \eqref{contamin} itself does not directly provide a feasible
inference strategy since the matrix $R$ depends on the unknown $\pi^0$. Of course the matrix $R$
may be consistently estimated. However, as explained by \cite{antoine2012efficient} (see the discussion of their Theorem 4.5), a sufficient condition to ensure that the estimation of $R$ does not pollute the asymptotic distribution in \eqref{contamin} is that the matrix $R$ is estimable at a rate faster than $n^{1/4}$. In the case of the probit model, the matrix $R$ only depends on the unknown true reduced form parameter $\pi^0 $, which is strongly identified and consistently estimable at the $\sqrt{n}$-rate. Therefore, if $\hat{R}_{n}$ denotes the matrix $R$ where
$\pi^0 $ is replaced by $\hat{\pi}_{n}$, we can conclude that, following Theorem 4.5 in
\cite{antoine2012efficient},
\begin{equation}
\sqrt{n}\Lambda _{n}^{-1}\hat{R}_{n}^{-1}(\hat{\theta}_{n}-\theta
_{}^{0}) \overset{d}{\rightarrow}\mathcal{N} \left( 0,[M'S^{-1}M]^{-1}\right)   \label{inference}.
\end{equation}

\begin{remark}\label{rem:eight}{\normalfont
 The result in equation (\ref{inference})  implies that $\sqrt{n}\Lambda
_{n}^{-1}\hat{R}_{n}^{-1}(\hat{\theta}_{n}-\theta _{}^{0})$ behaves like a mean-zero Gaussian random variable, whose variance can be consistently estimated by $$[\Lambda_n\hat R_n^\prime\{\partial\bar{g}_n(\hat\theta_n)/\partial\theta'\}'S^{-1}_n(\hat\theta_n)\{\partial\bar{g}_n(\hat\theta_n)/\partial\theta'\}\hat R_n\Lambda_n]^{-1}
.$$ {However, Theorem \ref{theoremN} \textit{does not say} that the common estimator of the variance-matrix of $\sqrt{n}(\hat\theta_n-\theta^0)$, obtained using the standard formula  $$[\{\partial\bar{g}_n(\hat\theta_n)/\partial\theta'\}'S^{-1}_n(\hat\theta_n)\{\partial\bar{g}_n(\hat\theta_n)/\partial\theta'\}]^{-1},$$ is well-behaved, which follows by noting that the matrix $\frac{\partial\bar{g}_n(\theta^0)}{\partial\theta}S^{-1}_n(\theta^0)\frac{\partial\bar{g}_n(\theta^0)}{\partial\theta'}$
is asymptotically singular unless $\varsigma_n=O(1)$.}
Fortunately, Theorem 5.1 in \cite{antoine2012efficient} allows us to
conclude that standard formulas for Wald inference based on the GMM estimator $%
\hat{\theta}_{n}$ are asymptotically valid. The main intuition is that the
Studentization implied by Wald inference cancels out the required
rescaling terms. This is all the more important given that the rescaling factor $%
\varsigma _{n}$ is unknown in practice.

We stress that this result is in contrast to the general nonlinear case where the asymptotic normality requires faster than $n^{1/4}$ convergence rate, and it is only due to the specificities of the probit model that we are able to conduct valid Wald inference as soon as
identification is not genuinely weak. That is, any near weakness, even as severe as $%
\varsigma _{n}$ being arbitrarily close to $\sqrt{n}$, will still allow us to compute
a consistent GMM estimator and apply standard formulas for Wald
inference based on this estimator.
}
\end{remark}

\subsubsection{The Power of the Distorted J-Test}
The key to ensuring that the size of $W_{n}^{\delta}$ is asymptotically controlled is the equivalence between $J_n$, the usual $J$-statistic, and $J_n^\delta$, the distorted $J$-statistic, that obtains under the null of weak identification. However, as demonstrated by Proposition \ref{lem:cons} and Theorem \ref{theoremN}, under the alternative hypothesis the CUE is consistent and asymptotically normal. 
Therefore, there is no reason to suspect that $J_n$ and $J_{n}^{\delta}$ will be asymptotically equivalent under the alternative, at least under reasonable choices for the tuning parameter $\delta_n$.

The following result demonstrates that under the alternative, the distorted J-test, $W_n^\delta$, is a consistent test for the null of weak instruments across a wide range of choices for the perturbation sequence $\delta_n$.

\begin{theorem}[Distorted J-test: Under the Alternative]\label{theorem3}
	If \textbf{Assumptions 1-6}  are satisfied, then $W_n^\delta$ is consistent under the alternative so long as $\{\sqrt{n}/\varsigma_{n}\}\delta_n\rightarrow\infty$ as $n\rightarrow\infty$.
\end{theorem}

\begin{remark}\normalfont
	Theorem \ref{theorem3} implies that our choice of $\delta_{n}$ has important consequences for the power of the distorted $J$-test. All else equal, the test is more powerful the slower $\delta_n$ goes to zero. However, it is also helpful to understand how fast $\delta_n$ can converge to zero before the result of Theorem \ref{theorem3} is invalidated. To this end, consider the rate requirement on $\delta_n$ that results from parametrizing $\varsigma_{n}$ as  $\varsigma_{n}=n^{\lambda}$ for some $0<\lambda<1/2$. Using this parametrization, we see that the distorted $J$-test is consistent so long as $\delta_nn^{1/2-\lambda}\rightarrow\infty$, and clarifies that if $\delta_n$ goes to zero too fast, i.e., if $\delta_n\ll n^{\lambda-1/2}$, the test can not be consistent.
\end{remark}

\begin{remark}
	\normalfont
	It is worth keeping in mind that \textbf{Assumption 2} maintains that both the structural and reduced form moments are correctly specified. Thus, when the observed data lead to a rejection of $W_n^\delta$, we immediately conclude that it
	is \textit{not due to misspecification} of the moment conditions but \textit{due to their
		identification power}. However, if the model is misspecified, but we reject the null of weak identification, then we can actually consistently test for model misspecification. Indeed, under the alternative, the standard overidentification test $$\{J_n(\hat\theta_n,\hat\theta_n)>\chi^2_{1-\alpha}(H-p)\},$$ remains a consistent test for model misspecification. As such, if we reject the null of weak identification, we can compare the value of $J_n(\hat\theta_n,\hat\theta_n)$ against $\chi^2_{1-\alpha}(H-p)$ to deduce a consistent test for model misspecification.
\end{remark}

\subsection{Testing Procedure}\label{section3.3}
We now explain one approach to implement our distorted J-test in practice. The key step in the testing procedure is to choose the perturbation (tuning parameter) $\delta_{n}$. To this end, we take $\delta_n=\delta/r_n$, and fix $r_n=\log \{\log(n)\}$. It is then possible to choose $\delta$ using a data-driven approach.

To present our approach to choosing $\delta$, first recall that the perturbation $\delta_n=\delta/\log\{\log(n)\}$ can be thought of as only being applied to the single direction of weakness in the rotated parameter space; namely, the parameter $\eta_1$, which by equation \eqref{eq:basis} is nothing but $\tilde\rho$. Therefore, it is with respect to the magnitude of ${\hat{\tilde{\rho}}}_n$ that the perturbation $\delta_n$ should be chosen.

To ensure the value of $\delta_n$ is sufficiently close to the magnitude of $\hat{\tilde{\rho}}_n$, we design a grid of $m$ candidate points for $\delta$ by dissecting the \textit{standard confidence interval} of $\hat{\tilde{\rho}}_n$ into $m$ equal regions. For the $i$-th region, we set $\delta_i$, $i=1,\dots,m$, to be to the midpoint of the $i$-th region. This produces a grid of $m$ perturbations with $i$-th value, $i=1,\dots,m$ given by $\delta_{i,n}=\delta_i/r_n$.

Whilst it is possible to use any given $\delta_{n,i}$ to conduct the test, we suggest carrying out the test across the entire grid of $\delta_{n,i}$ values and then appropriately modify the critical value via a Bonferroni correction. In particular, let $J_{n,i}^{\delta}$ denote the test statistic $J_{n}^\delta$ calculated under the perturbation $\delta_{n,i}$. This approach would lead us to reject the null of weak identification if
$$
\max_{i\in\{1,\dots,m\}}J_{n,i}^\delta>\chi^2_{1-\alpha/m}(H+1-p).
$$

Using the above decision rule, our approach can be implemented using the following four steps.
\begin{enumerate}
	\item[(1)] Compute $\hat{\theta}_{n}=\argmin_{\theta\in\Theta}J_n(\theta,\theta)$;
	\item[(2)] For a given choice of $m$, choose the sequence of tuning parameter $\delta_n=\delta/r_n$, as described above;
	\item[(3)] For each $i=1,\dots,m$, compute the test statistic $J_{n,i}^{\delta}$, as defined in Section \ref{sec:3.3};
	\item[(4)] Rejection rule: reject if $
	\max_{i\in\{1,\dots,m\}}J_{n,i}^\delta>\chi^2_{1-\alpha/m}(H+1-p).
	$
	\end{enumerate}

Under the null hypothesis, the testing procedure is size controlled for any choice of $\delta_{n,i}=o(1)$, while under the alternative the choice of $\delta_{n,i}$ only has implications for the power of the test. Moreover, since the values of $\delta_i$ are chosen from some compact set, dividing by $\log\{\log(n)\}$ ensures that $\delta_{n,i}=o(1)$ under both the null and alternative.


\subsection{Generalizing the Rule-of-Thumb to Probit Models}\label{sectionGROT}

{We begin our discussion on the so-called ``rule-of-thumb",
initially inspired by the work of \cite{staiger1997instrumental}, in the infeasible situation where the latent endogenous
variable $y_{1i}^{\ast }$\ is observable,} meaning that we would consider
a bivariate linear model. {For sake of expositional simplicity, let us
consider a simplification of this model whereby} the vector $x_{i}$
only contains a constant, so that the model becomes
\begin{eqnarray}
y_{1i}^{\ast } &=&\alpha y_{2i}+\beta +u_{i}  \label{Probit} \\
y_{2i} &=&\pi +z_{i}^{\prime }\xi +v_{i}  \nonumber.
\end{eqnarray}

The rule-of-thumb starts from the reduced form regression and its OLS
estimator for $\xi$,
\[
\hat{\xi}_{n}=(\tilde{Z}^{\prime }\tilde{Z})^{-1}\tilde{Z}^{\prime }\tilde{Y}_{2} ,
\]
where for $\mathbf{1}_n$ a $(n\times 1)$-vector of ones
\begin{align*}
Y_{2}&=( y_{21},\dots,y_{2n})',\;\tilde{Y}_2=Y_2-\bar{y}_{2n}\mathbf{1}_n,\;\\
Z&=( z'_{1},\dots,z'_{n})',\;\tilde{Z}=Z-\bar{Z}_{n} ,
\end{align*}
and where $\bar{y}_{2n}=\frac{1}{n}\sum_{i=1}^{n}y_{2i} $ and $\bar{Z}_{n}$\ denotes the $\left( n\times k_{z}\right) $ matrix
whose $j^{th}$-column has all its entries equal to
\[
\bar{z}_{j,n}=\frac{1}{n}\sum_{i=1}^{n}z_{ij} .
\]
Let $F_{n}$\ denote the F-test statistic for testing the null hypothesis
that the vector $\xi $ of coefficients {for the variables} $z_{i}$ in the reduced form regression {are zero}. Under the assumption of conditional
homoskedasticity {for the error term $v_{i}$,}
the F-test statistic can be written as
\begin{equation*}
F_{n}=\frac{n-k_{z}}{nk_{z}}\frac{1}{\hat{\sigma}_{v,n}^{2}}\left[ \hat{\xi}%
_{n}^{\prime }\left( \tilde{Z}^{\prime }\tilde{Z}\right) \hat{\xi}_{n}\right],
\end{equation*}
with $\hat{\sigma}_{v,n}^{2}$ a consistent estimator of variance of $v_i$, $\sigma^2_v$.
The rule-of-thumb amounts to conclude that instruments are {strong (i.e., consistent estimation is feasible)} if $F_{n}$ {exceeds a pre-specified threshold value, which differs from the standard critical value used to test the null hypothesis $\text{H}_{0}:\xi =0,$ and which has} been extensively
documented by \cite{stock2005testing}. The rationale for this rule can be
understood from the drifting DGP considered in Remark 7. Under the
alternative hypothesis to the null of weak identification, for $n$ large,
\begin{equation}
\xi _{n}^{0}\sim \frac{\gamma ^{0}}{\varsigma _{n}}\Longrightarrow
k_{z}F_{n}\sim \frac{n}{\varsigma _{n}^{2}}\frac{1}{\sigma _{v}^{2}}\gamma
^{0\prime }\text{Var}\left( z_{i}\right)\gamma ^{0}.
\label{equiv}
\end{equation}

{Therefore, under the null of weak identification ($\varsigma _{n}=%
\sqrt{n}$), $F_{n}$ in equation (\ref{equiv}) has a finite limit,
whilst under} the alternative ($\varsigma
_{n}=o\left( \sqrt{n}\right) $) {the statistic $F_n$ diverges} to infinity
with a slope defined by the squared norm of $\gamma ^{0}$ and a weighting
matrix that is proportional to $\text{Var}(z_{i})/\text{Var}(v_{i})$. This
sounds like a natural criterion to measure {instrument strength in the infeasible model \eqref{Probit},} since the reduced form regression will lead
to the control variable $v_{i}=y_{2i}-\pi -z_{i}^{\prime }\xi $\ and
endogeneity in the structural equation will be controlled thanks to the
two-stage residual inclusion (2SRI):%
\begin{equation}
y_{1i}^{\ast }=\alpha y_{2i}+\beta +\tilde{\rho}\left[ y_{2i}-\pi
-z_{i}^{\prime }\xi \right] +\varepsilon _{i}  \label{2SRI}.
\end{equation}

{Since identification of $\eta _{1}=\tilde{\rho}$ in equation (\ref%
{2SRI}) depends on the variation of}
\[
z_{i}^{\prime }\xi _{n}^{0}\sim \frac{z_{i}^{\prime }\gamma ^{0}}{\varsigma
	_{n}} ,
\]
it may sound natural to assess the magnitude of $\gamma ^{0}$\ after
normalization by the variance of $z_{i}$. As noted by \cite{16324}, ``IVs can be weak and the F statistic small, either because $\gamma$ is close to zero or because the variability of $z_i$ is low relative to the variability of $v_i$.'' However, the F-test statistic
follows a Fisher distribution (and asymptotically a distribution $\chi
^{2}\left( k_{z}\right) /k_{z}$) under the null $\text{H}_{0}:\xi =0$\ only when
{the reduced form error term $v_{i}$} is
conditionally homoskedastic. When one is {concerned with the} presence of conditional
heteroskedasticity in this equation (i.e., non-constant $\text{Var}[v_{i}\mid z_{i}] $), one may consider the heteroskedasticty corrected Fisher test
statistic
\[
F_{n}^{\ast }=\frac{n-k_{z}}{k_{z}}\left[ \hat{\xi}_{n}^{\prime }\hat{\Sigma}%
_{n}^{-1}\hat{\xi}_{n}\right],
\]
where $\hat{\Sigma}_{n}$ is a consistent estimator of the asymptotic
variance of $\sqrt{n}( \hat{\xi}_{n}-\xi _{n}^{0}) $. While \cite{stock2005testing} propose to extend the use of the rule-of-thumb by using
instead $F_{n}^{\ast }$\ in case of conditional heteroskedasticity, several
authors, including \cite{andrews2018valid} and \cite{olea2013robust}, have
documented the disappointing performance of the heteroskedasticity corrected
rule-of-thumb. One may help to clarify this issue by noting that, denoting $\tilde{z}_{i}$ to be the $i$-th column vector of the matrix $\tilde{Z}'$, {for $n$ large and for $\sigma _{v}^{2}(z_{i})=\text{Var}[v_{i}\mid z_{i}]$},
\begin{equation}
\xi _{n}^{0}\sim \frac{\gamma ^{0}}{\varsigma _{n}}\Longrightarrow
k_{z}F_{n}^{\ast }\sim \frac{n}{\varsigma _{n}^{2}}\gamma ^{0\prime
}\text{Var}\left( z_{i}\right) \left[ \E\left( \tilde{z}_{i}\tilde{z}_{i}^{\prime
}\sigma _{v}^{2}(z_{i})\right) \right] ^{-1}\text{Var}\left( z_{i}\right) \gamma
^{0}  \label{RobF}.
\end{equation}

Equation (\ref{RobF}) is a straightforward extension of a result {provided by
\cite{antoine2017testing}, and makes explicit {how robustifying the}
test statistic for heteroskedasticity modifies the
rule-of-thumb.} This modification is arguably puzzling since what really
matters{ for identification power, namely the residual inclusion of $v_{i}$
in the structural equation (\ref{2SRI}), is not fully captured by $\sigma _{v}^{2}(z_{i})$.} More precisely, the conditional heteroskedasticity
that intuitively matters in the structural equation {is instead}
\[
\sigma _{u}^{2}(z_{i})=\text{Var}[u_{i}\left\vert z_{i}\right] =\tilde{\rho}%
^{2}\text{Var}[v_{i}\left\vert z_{i}\right] +\text{Var}[\varepsilon _{i}\left\vert z_{i}%
\right] .
\]

This intuition is confirmed by \cite{antoine2017testing} who show that,
when nesting the IV estimation procedure in a GMM framework, the distorted
J-test leads to a decision rule {based on the following weighted norm of $\gamma ^{0}$}:
\[
\frac{n}{\varsigma _{n}^{2}}\gamma ^{0\prime }\text{Var}\left( z_{i}\right) \left[
\E\left( \tilde{z}_{i}\tilde{z}_{i}^{\prime }\sigma _{u}^{2}(z_{i})\right) %
\right] ^{-1}\text{Var}\left( z_{i}\right) \gamma ^{0} .
\]
{In the context of the probit model,} where only the sign $y_{1i}$ of $y_{1i}^{\ast }$\ is
observed, the 2SRI equation becomes
\[
y_{1i}=\Phi \left[ \alpha y_{2i}+\beta +\tilde{\rho}\left( y_{2i}-\pi
-z_{i}^{\prime }\xi \right) \right] +\varepsilon _{i},
\]
for some error term $\varepsilon_i$, {and the conditional heteroskedasticity in the structural equation takes the form}
\begin{equation*}
	\text{Var}[\varepsilon _{i}\left\vert y_{2i},z_{i}\right] =\Phi _{i}\left( \theta
	^{0}\right) \left[ 1-\Phi _{i}\left( \theta ^{0}\right) \right],\;	\text{ where }\Phi _{i}\left( \theta \right) =\Phi \left[ \alpha y_{2i}+\beta +\tilde{%
		\rho}\left( y_{2i}-\pi -z_{i}^{\prime }\xi \right) \right].
\end{equation*}

{One may then expect that any generalized rule-of-thumb for probit models must account not only for} this conditional heteroskedasticity but also the
impact of the non-linearity in the structural equation. In the simple context of
Remark 7, we may then expect that the {key element to obtain a decision rule about
weak instruments in the probit model is the magnitude of the vector}
\[
V^{0}\left( \eta \right) =\E_n\left[ \tilde{a}\left( y_{2i},z_{i}\right) \phi
_{i}\left( \eta,\theta _{2}^{0}\right) z_{i}^{\prime }\right] \gamma ^{0}, \text{ where }\|\gamma^0\|>0.
\]

More generally, since the alternative to weak identification, defined by
\textbf{Assumption 5}, is tantamount to the non-nullity of the vector $V^{0}\left(
\eta \right) $, the generalized rule-of-thumb should be based on a norm of
$V^{0}\left( \eta \right) $. We argue that we do have a
well-suited generalization for the standard rule-of-thumb when applying a
decision rule that rejects the null of weak identification if the
norm $\left\Vert U\right\Vert $, of a certain vector $U$, exceeds a specified
threshold with the following definition for $U$.
\begin{enumerate}
	\item[(i)] $U=\sqrt{n}\text{Var}\left( z_{i}\right)^{1/2} /\sigma_v\xi^{0}$ for a
	linear model with conditional homoskedasticity (i.e. the standard
	rule-of-thumb);
	\item[(ii)] $U=\sqrt{n}\left[ \E\left( \tilde{z}_{i}\tilde{z}_{i}^{\prime }\sigma
	_{u}^{2}(z_{i})\right) \right] ^{-1/2}\text{Var}\left( z_{i}\right) \xi^{0}$
	for a linear model with conditional heteroskedasticity (i.e. the
	generalization of the standard rule-of-thumb proposed by \citealp{antoine2017testing});
	\item[(iii)]$U=\sqrt{n}S_{11,n}^{-1/2}\left( \theta ^{0}\right) \E_n\left[ \tilde{a}%
	\left( y_{2i},z_{i}\right) \phi _{i}\left( \eta,\theta _{2}^{0}\right)
	z_{i}^{\prime }\right] \xi^{0}\delta_n$ for the probit model (\ref{Probit}) (in
	the context of Remark 7) and more generally $U=\sqrt{n}S_{11,n}^{-1/2}\left( \theta
	^{0}\right) V^{0}\left( \eta \right)\delta_n/\varsigma_n$, where the perturbation term $\delta_n$ is introduced by the design of the distorted J-test.
\end{enumerate}

It is worth realizing that this generalized rule-of-thumb is, for $n$ large, precisely what is performed by our test for the null of weak
identification based on the distorted J-test statistic.\footnote{ It is worth realizing that our comparison between the
	so-called ``rules-of-thumb" is based only on the definition of the test
	statistic. We do not enter into debates regarding alternative definitions of
	the null of weak identification based either on Assumption 3, or the 2SLS
	relative bias, Wald test size distortion, Nagar bias, etc..} To see this, we
extend the argument of \cite{antoine2017testing} by noting that under the alternative, our
distorted J-test statistic sets the focus on the norm of
\[
U=S_{n}^{-1/2}\left( \theta ^{0}\right) \sqrt{n}\bar{g}_{n}( \hat{\theta%
}_{n}^{\delta }) ,
\]
where
\[
\bar{g}_{n}( \hat{\theta}_{n}^{\delta }) =\bar{g}_{n}( \hat{%
	\theta}_{n}) +\left[
\begin{array}{c}
\bar{g}_{1n}( \hat{\theta}_{n}^{\delta }) -\bar{g}_{1n}(
\hat{\theta}_{n})  \\
\mathbf{0}%
\end{array}%
\right].
\]
{Noting that,}
\[
\sqrt{n}\left[ \bar{g}_{1n}( \hat{\theta}_{n}^{\delta }) -\bar{g}%
_{1n}( \hat{\theta}_{n}) \right] =\sqrt{n}\frac{\partial \bar{g}%
	_{1n}}{\partial \eta _{1}}( \eta _{1n}^{\ast },\hat{\eta}_{2n},\hat{\eta}%
_{3n},\hat{\theta}_{2n}) \delta _{n},
\]
where $\eta _{1n}^{\ast }$ denotes a component-by-component intermediate
value between the first coefficient of $\hat{\theta}_{n}$\ and $\hat{\theta}%
_{n}^{\delta }$, under the alternative hypothesis to the null of weak
identification
\begin{eqnarray*}
	\frac{\partial \bar{g}_{1n}}{\partial \eta _{1}}\left( \eta _{1n}^{\ast },%
	\hat{\eta}_{2n},\hat{\eta}_{3n},\hat{\theta}_{2n}\right)  &=&\E_{n}\left[
	\frac{\partial \bar{g}_{1n}}{\partial \eta _{1}}\left( \eta ^{0},\theta
	_{2}^{0}\right) \right] +O_{p}\left( \frac{1}{\sqrt{n}}\right)  \\
	&=&\frac{1}{n}\E_{n}\left\{ \sum_{i=1}^{n}\tilde{a}\left( y_{2i},z_{i}\right)
	\phi _{i}\left( \eta^0,\theta _{2}^{0}\right) z_{i}^{\prime }\xi ^{0}\right\}
	+O_{p}\left( \frac{1}{\sqrt{n}}\right) ,
\end{eqnarray*}
and where\footnote{The $O_p(1/\sqrt{n})$ term in the expansion of ${\partial \bar{g}_{1n}( \eta _{1n}^{\ast },\hat{\eta}_{2n},\hat{\eta}_{3n},\hat{\theta}_{2n}) }/{\partial \eta _{1}}$ can be deduced via a Taylor series expansion, re-arranging terms, and noting that the derivative of the Jacobian, in the $\eta_1$ direction, is also degenerate at the $\varsigma_{n}$-rate.}
\[
\frac{1}{n}\E_{n}\left\{ \sum_{i=1}^{n}\tilde{a}\left( y_{2i},z_{i}\right)
\phi _{i}\left( \eta ,\theta _{2}^{0}\right) z_{i}^{\prime }\xi ^{0}\right\}
\sim \frac{V^{0}\left( \eta \right) }{\varsigma _{n}}
\]
is the dominant term since $\varsigma _{n}=o\left( \sqrt{n}\right) .$ To
summarize, under the alternative hypothesis to the null of weak
identification, and for a $\delta_n$ such that $\{\sqrt{n}/\varsigma_{n}\}\delta_n\rightarrow \infty$,
\begin{eqnarray*}
\left\Vert U\right\Vert =\left\Vert S_{n}^{-1/2}\left(
	\theta ^{0}\right) \sqrt{n}\bar{g}_{n}\left( \hat{\theta}_{n}^{\delta
	}\right) \right\Vert \sim \left\Vert S_{11,n}^{-1/2}\left( \theta ^{0}\right)
	V^{0}\left( \eta^0 \right) \right\Vert \frac{\sqrt{n}}{\varsigma _{n}}\delta
	_{n},
\end{eqnarray*}
{which diverges as $n\rightarrow\infty$ and yields a natural generalization of the rule-of-thumb to probit models.}

\section{Monte Carlo: Conventional Weak IV Tests v.s. Distorted $J$-test}\label{sectionMC}
In this section, we verify the properties of the distorted $J$-test (hereafter, DJ test) and compare this test against three commonly used weak IV tests, which, even though they are not designed for discrete choice models, have been widely applied in the literature on discrete choice modelling: (i) the \citet{staiger1997instrumental} standard rule-of-thumb (SS); (ii) \citet{stock2005testing} (SY); and (iii) the robust weak IV test of \citet{olea2013robust} (Robust).

We generate observed data according to
\begin{align}\label{16}
y_{1i}=1[\beta+\alpha y_{2i}+u_i>0],~~y_{2i}=\pi+\xi z_i+v_i,~~i=1,2,...,n,
\end{align}
where \(z_i\sim \mathcal{N}(0,\sigma_z^2)\) is i.i.d. univariate, \((u_i,v_i)'\) is i.i.d. homoskedastic and normally distributed, and \((u_i,v_i)'\) is independent of \(z_i\). We set \(\beta=0.5\), \(\alpha=1\) and \(\pi=0.3\). In addition, we take \(\rho=corr(u_i,v_i)\in\{0.5,~0.95\}\), and \(\sigma_u=1/\sqrt{1-\rho^2}\) (to ensure the normalization of \(\text{Var}[u_i|y_{2i},z_i]=1\)). To characterize the potential instrument weakness, we adjust the value of $\xi$ to restrict the correlation between the endogenous regressor $y_{2i}$ and the instrument $z_i$ to be $corr(y_{2i},z_i)=\gamma/n^\lambda$, with \(\gamma=1.5\), and we consider a grid of values for $\lambda\in\{0.5,~0.4,~0.3,~0.2,~0.1\}$.

Since the performance of the DJ test and the standard weak IV tests may depend on $\sigma_{z}$ and $\sigma_{v}$, we simulate data using the following grids: $\sigma_z\in\{0.2,~0.5,~1,~5,~10\}$ and $\sigma_v\in\{0.2,~0.5,~1,~5,~10\}$. For each Monte Carlo trial, we take the sample size to be one of \(n=500,~5000,~10000\) and consider $N=1000$ Monte Carlo replications.

Across each Monte Carlo design, \(\theta=(\widetilde{\rho},\alpha,\beta,\pi,\xi)'\) is estimated by CUGMM with a single degree of over-identification. We choose the instrument functions $a_i=a(y_{2i},z_i)=(1,y_{2i},z_i,z_i^2,0,0)'$ and $b_i=b(z_i)=(0,0,0,0,1,z_i)'$. The DJ test is implemented following the procedure presented in Section \ref{section3.3}.\footnote{For computational simplicity, in the Monte Carlo simulations, we adopt the perturbation $\delta_n=\hat{\tilde{\rho}}/\log(\log(n))$, where $\hat{\tilde{\rho}}$ is the CUGMM estimate of $\tilde{\rho}$ in each Monte Carlo replication. This procedures is a simplified version of the data-driven approach developed in Section \ref{section3.3}.} Using a 5$\%$ significant level, we reject the null hypothesis of weak instruments in accordance to Theorem \ref{theorem3}; i.e., we reject the null if $J^\delta_n>\chi^2_{0.95}(H+1-p)$, where in this case $H=6$, $p=5$ and $\chi^2_{0.95}(H+1-p)=5.99$. Theoretically, the hypotheses of the DJ test correspond to $H_0:\lambda=0.5$, and the alternative  to $\lambda<0.5$.\footnote{We note that the null hypothesis of each test are slightly different: DJ- $\text{H}_{0}:\lambda=0.5$; SS- $F_n<10$ as an informal null hypothesis; SY- the triple $\{\xi,\sigma_v^2,\sigma_z^2\}$ is such that 2SLS relative bias or Wald test size distortion is larger than a given tolerance using the Cragg-Donald statistic; the Robust test regards that the Nagar bias exceeds a fraction of the benchmark as null. Although the definitions of the weak instrument are different for each test, their null hypothesis are consistent in the sense to capture situations under which the instrument is weak.} However, we note that in finite samples, it is hardly the case that $\lambda$ alone determines the behavior of the CUEs.

Given this, to compare the behavior of the DJ test with the conventional linear tests, we introduce two sets of criteria to assess the potential impact of instrument weakness in finite samples: the behavior of the CUE and the size distortions of the
associated Wald statistic. Specifically, we compute the bias, standard deviation (s.d.) and relative root mean square error (rrmse) as below (taking $\alpha$ as an example) to measure the estimation performance under different designs:
\begin{align}\label{criterion1}
\text{bias}=\bar{\widehat{\alpha}}-\alpha^0,~~\text{s.d.}=\sqrt{\frac{1}{N}\sum_{l=1}^{N}(\widehat{\alpha}_{l}-\bar{\widehat{\alpha}})^{2}},~~\text{rrmse}=\sqrt{\frac{1}{N}\sum_{l=1}^{N}\left(\frac{\widehat{\alpha}_{l}-\alpha^0}{\alpha^0}\right)^{2}}
\end{align}
where $\bar{\widehat{\alpha}}=1/N\sum_{l=1}^N\widehat{\alpha}_{l}$, $\widehat{\alpha}_{l}$ stands for the $l$-th Monte Carlo CUGMM estimate and $\alpha^0$ is the true value. As proven in Sections \ref{sec:3.3} and \ref{sec:alt}, under the null the CUE is consistent, while under the alternative, the estimator will be consistent and asymptotically normal, albeit with non-standard rates. Unlike \citet{stock2005testing}, who choose the relative bias of 2SLS to OLS as one criterion to detect weak instruments, here we consider the bias, the s.d. and the rrmse defined in \eqref{criterion1} instead, for the following reasons. For the IV probit model \eqref{16}, the CUE (and other commonly adopted estimation methods) does not have a closed-from expression. Therefore, the usual notion of `bias towards OLS' under potential IV weakness in linear models is not valid in this nonlinear context, with the potential impact of the IV weakness now being complicated by the nonlinear features of the model. In this case, there is no guarantee that the positive and negative biases will not offset each other and lead to a spuriously small overall bias. Therefore, to capture the instrument strength and the resulting performance of the CUGMM estimation procedure, we rely on the bias, standard deviation and rrmse of the estimator.

In addition, to better understand weakness in this discrete choice model, we conduct a Wald test of $\text{H}_0:\alpha=\alpha^0$ and compute its size distortion, relative to the $5\%$ significant level, across all the Monte Carlo designs. We carry out this Wald test for two different estimation methods: the CUE considered in this paper and the 2SCML estimator proposed by \citet{rivers1988limited}. The size distortion of the Wald test is widely used to capture instrument weakness; see e.g. \citet{staiger1997instrumental} and \citet{stock2005testing}. This measure reflects not only the performance of the hypothesis test, but also the coverage rate of confidence intervals associated with the two estimation methods. 

Under the null hypothesis of $\lambda=0.5$, the performance of the CUE and the rejection probabilities for the different testing procedures are collected in Table \ref{tab:tab1} ($\rho=0.5$) and Table \ref{tab:tab2} ($\rho=0.95$). For brevity, we only report the estimation results for the structural parameter of interest, $\alpha$, and Wald test size distortions under five designs: $(\sigma_z,\sigma_v)\in\{(1,0.2),(1,10),(1,1),(0.2,1),(10,1)\}$. Additional results for all designs can be obtained from the authors.

Simulation results in Tables \ref{tab:tab1} and \ref{tab:tab2} confirm our asymptotic results. When $\lambda=0.5$, CUGMM estimation of $\alpha^0$ is inconsistent and behaves poorly in general. More specifically, the biases are unstable, and the s.d. and rrmse do not decrease (in any noticeable way) as the sample size increases, especially when the endogeneity degree is high ($\rho=0.95$). However, under the alternative, $\lambda<0.5$, the s.d. and rrmse drop dramatically as $n$ increases. In addition, the asymptotic normality of the CUE under $\lambda<0.5$ is verified by viewing the standardized sampling distributions of the estimators across the Monte Carlo replications, which is given in Figures \ref{fig:emp0.5} and \ref{fig:emp0.95}. The sampling distributions exhibit easily detectable bi-modality when $\lambda$ is 0.5, or close to 0.5, especially when $\sigma_v$ is small and $\rho$ is large, indicating that a standard inference approach, relying on the normal approximation, is likely to perform poorly in those cases.

The results in Tables \ref{tab:tab1} and \ref{tab:tab2} also show that the behavior of the Wald test varies across the different designs even when $\lambda=0.5$. For a moderate level of endogeneity ($\rho=0.5$), we see relative small size distortions, less than $5\%$, in most cases for the Wald tests based on both 2SCML and CUEs. However, for a high degree of endogeneity ($\rho=0.95$), the Wald tests are significantly over-sized, with the size distortions for the Wald test based on 2SCML being much larger than those based on CUGMM. One exception, however, is the case of $(\sigma_z,\sigma_v)=(1,10)$ and $\rho=0.95$, where the Wald size distortions based on both estimation methods are less than $5\%$. {For $(\sigma_z,\sigma_v)=(1,10)$ and $\rho=0.95$ case, the size distortion based on the CUGMM is 0.008 when $n=10000$, indicating that the $95\%$ confidence interval coverage rate is quite accurate even though $\lambda=0.5$ ($corr(y_{2i},z_i)=0.015$). As such, this design constitutes additional evidence that the value of $\lambda$ is not the only key in determining inference performance in weakly identified discrete choice models.}

The false rejection rates of SS, SY, Robust and DJ under $\lambda=0.5$ are displayed in Tables \ref{tab:tab1} and \ref{tab:tab2}. Firstly, as expected, the DJ test is asymptotically conservative, i.e., the size is less than the significance level of $5\%$.\footnote{Results not reported here due to space limitations also confirm that DJ test is conservative.} The size of the DJ test varies between 1.0\% and 1.9\% under $\rho=0.5$, and between 1.3\% and 3.1\% when $\rho=0.95$. However, we note that the DJ test is much less conservative than the Robust approach of \cite{olea2013robust}, which is extremely conservative, and gives virtually zero rejections across all designs where identification is weak. Therefore, while the DJ test is conservative, we can conclude that it is much less conservative than the Robust approach, and can be relatively close to the nominal level (5\%) when the degree of endogeneity is large.

In addition, we see that blindly applying conventional weak instrument tests can lead to poor outcomes. For example, for the design with $(\sigma_z,\sigma_v)=(1,0.2)$ and a high level of endogeneity ($\rho=0.95$ in Table \ref{tab:tab2}), the rejection rates of SS and SY (10\%)\footnote{The rejection rate of SY (10\%) is computed based on the critical value of a maximal 10$\%$ size distortion of a 5$\%$ Wald test, provided by \citet{stock2005testing}.} are all larger than 10$\%$ across different sample sizes, and are 13.8$\%$ and 18.5$\%$ respectively when $n=10000$. However, the rrmse in this case does not decrease as $n$ increases, and the rrmese for the estimated $\alpha$ is between 910\% and 1060$\%$ of the true value. Moreover, both of the Wald size distortions exceed their nominal size by at least 10$\%$. In particular, the 2SCML size distortion is between 17$\%$ and 27$\%$, while the CUGMM size distortion is between 11$\%$ to 17$\%$. Therefore, the identification is weak, but the SS and SY approaches can suggest the opposite, and hence fail to control size. In addition, false rejection rates for other designs, not reported here for brevity, demonstrate a similar pattern of over-rejection for SS and SY tests. Hence, in line with the analysis in Section \ref{sectionGROT}, when assessing identification strength in discrete choice models, the conventional weak IV tests of SS and SY may fail to provide reliable conclusions regarding identification strength, especially if the degree of endogeneity is high.

Figure \ref{fig:power0.5} ($\rho=0.5$) and Figure \ref{fig:power0.95} ($\rho=0.95$) display the power of the four tests. Due to the conservativeness of DJ test, size adjusted power of DJ and of the conventional tests are also computed and compared in Figures \ref{fig:adj_power0.5} and \ref{fig:adj_power0.95}.\footnote{Size adjusted power is computed as follows: obtain the 95$\%$ quantile of the test statistic from the 1000 Monte Carlo replications when $\lambda=0.5$ and use it as the critical value for cases when $\lambda<0.5$.} The resulting power curves show that the DJ test is consistent as the sample size diverges, and as identification strength increases. Moreover, in cases with high endogeneity (Figure \ref{fig:power0.95}), the unadjusted power of the DJ is higher than that of the Robust test across most designs. Furthermore, Figures \ref{fig:adj_power0.5} and \ref{fig:adj_power0.95} demonstrate that the DJ-test displays non-negligible power even when identification is close to being weak, i.e, when $\lambda=0.4$ or $\lambda=0.3$, which gives convincing numerical evidence of the results in Theorem \ref{theorem3}.

\section{Empirical Illustrations}
In this section, we apply our distorted $J$-test in two well-known empirical examples to test for the presence of weak instruments. We then contrast the results of our tests with those obtained from conventional weak IV tests for linear models, namely the SS, the SY, and the Robust tests. 
\subsection{Labor Force Participation of Married Women}\label{sec:lfp}
We first study the impact of education on married women's labor force participation (hereafter LFP), when education, measured as the women's years of schooling, is treated as an endogenous treatment. We use data from the University of Michigan Panel Study of Income Dynamics (PSID) for the year 1975\footnote{The data is publicly available at \citet{wooldridge2010econometric} Supplemental Content.}, which have been used in several studies. \citet{mroz1987sensitivity} provides an extensive analysis of the women's hours of labor supply, and considers a range of specifications including potential endogeneity of several regressors, the use of different instrumental variables and controls for self-selection into labor force participation. As a text book example, \citet{wooldridge2010econometric} used the same dataset to study women's LFP decisions, and the potential endogeneity of education is tested after estimating an IV probit model using \citet{rivers1988limited} two-step conditional maximum likelihood estimator (2SCML). In what follows we use similar specification as in \citet{wooldridge2010econometric}.

The PSID consists of data on 753 married, Caucasian women who are between 30 and 60 years of age at the time the sampling. The dependent variable LFP is a binary response that equals unity if the respondent worked at some time during the year, and zero otherwise. Exogenous regressors include spousal income, the individual's work experience and its square, age, the number of children less than six years old, and the number of children older than six years old. The individual's education, measured as years of schooling, is considered to be endogenous. Following the strategy in \citet{wooldridge2010econometric}, the individual's family education, which are recorded as the years of schooling for both the individual's father and mother, are used as instruments for education.

Estimated coefficients and the average partial effects on the probability of LFP for all regressors are presented in Table \ref{tab:tab7} using two estimation methods: 2SCML as used in \citet{wooldridge2010econometric}\footnote{For the LFP example, the 2SCML estimation allows for heteroskedastic standard errors.} and CUGMM. More specifically, for the 2SCML, the first step is to regress the endogenous regressor on the instruments and all other exogenous regressors to obtain the reduced form residual. The second step is to run a probit maximum likelihood estimation of the binary response on the endogenous and the exogenous regressors, and the reduced form residual. The CUGMM estimation with over-identification degree one is conducted using $a_i=(1,y_{2i},x_i',z_{i}',\textbf{0}_{k+2}')'$ and $b_i=(\textbf{0}_{k+3}',1,x_i',z_i')'$, where $y_{2i}$, $x_i$ and $z_{i}$ denote the standardized variables corresponding to the women's education, exogenous regressors and two instruments, and $k$ is the number of exogenous regressors and the intercept. The first step estimation of the 2SCML and the reduced form of the CUGMM are listed in the first and fourth columns of Table \ref{tab:tab7} respectively. Both the two IVs are highly significant based on both estimation methods. The CUGMM estimation results are reported in columns four through six. Broadly speaking, the CUGMM and 2SCML results are similar, with both methods providing evidence that education has a significant positive effect: one extra year of education increasing the probability of LFP by 5.87 percentage points for both the 2SCML and the CUGMM. Hansen's $J$-statistic is 0.122 which is less than $\chi^2_{0.95}(1)=3.84$, therefore we fail to reject the null that all the moments are valid.

The weak IV test results are collected in Table \ref{tab:tab6} for all four tests, SS, SY, Robust and DJ. The Kleibergen-Paap $F$-statistic \citep{kleibergen2006generalized} is 81.89, based on which the SS rule-of-thumb and the SY test both reject the null that IVs are weak.\footnote{The Kleibergen-Paap $F$-statistic is utilized when allowing for heteroskedastic standard error. The reduced form regression $F$-statistic and the Cragg-Donald statistic are 95.70, when assuming homoskedastic standard error. SY rejects its null according to the critical values of the maximal desired size distortion $5\%$ and 10\% of a 5\% Wald test.} For the Robust test, the effective $F$-statistic is 91.44, the critical values for the tolerance thresholds $\{5\%,10\%\}$ are 11.59 and 8.58, respectively.\footnote{The estimated effective degrees of freedom of the Robust test for the tolerance thresholds $\{5\%,10\%\}$ are both about 1.8. See \cite{olea2013robust} for the definitions of the effective $F$-statistic, the tolerance threshold and the effective degrees of freedom. The Robust test statistic and the critical values are obtained using the Stata command "weakivtest" (\citet{pflueger2014robust}) under heteroskedastic-robust estimation.}  Comparing the effective $F$-statistic 91.44 to the critical values, the Robust test also rejects the null of weak IV.

Finally, for the DJ test, the perturbation $\delta_n$ is computed as in Section \ref{section3.3}, using $m=20$ candidate grid points. This choice of $m$ leads us to use the critical value $\chi^2_{1-0.05/20}(H+1-p)=11.98$, where we note that we have used $H=19$ moments and estimated $p=18$ parameters. Of the candidate grid points, three lead to a value of the DJ statistic larger than 11.98, leading us to soundly reject the null of weak identification. The rejection conclusion of the DJ test is quite straightforward: when perturbing the CUE $\hat{\theta}_n$ by $\delta_n$, the value of the $J$-statistic increases dramatically from 0.122 to a maximum of 17.44, implying that the CUGMM criterion is sensitive to even small departures. Overall, results reported in Table \ref{tab:tab6} suggest that the DJ test and the three conventional tests for linear models all agree in this example.

\subsection{US Food Aid and Civil Conflicts}
In the second example we examine the impact of US food aid on the incidence of civil conflicts in recipient countries. The research in \citet{nunn2014us} was motivated by concerns that humanitarian food aid may be ineffective and may even promote civil conflicts. The main challenge of this study is the potential endogeneity of US food aid due to reverse causality and joint determination. Their identification strategy relies on using the product of the lagged US wheat production and the average probability of receiving any US food aid for each country as the instrumental variable for wheat aid. \citet{nunn2014us} estimate many variations of binary and duration models and consider different kinds of wars, different controls and alternative specifications.

Herein, we focus on the cases of onset and offset of civil wars as 
considered in \citet{nunn2014us}. More specifically, we estimate the impact of US wheat aid on the probability of civil war \emph{onset} after a period of peace, or on the probability of civil war \emph{offset} after a period of war (columns (3)-(9), Table 7, \citep{nunn2014us}), using precisely the same datasets and model controls as in \citet{nunn2014us}.\footnote{Data sets used to construct the incidence of conflict, US food aid, US wheat production and other variables include the UCDP/PRIO Armed Conflict Dataset Version 4-2010, the Food and Agriculture Organization's FAOSTAT database and the data from the United States Department of Agriculture. See \citet{nunn2014us} for more detailed information.} We examine the IV strength by applying our DJ test to the model, as well as the three conventional weak IV tests for linear models. The dataset in this analysis involves observations on 78 non-OECD countries from 1971 to 2006.

For the onset analysis, the data used are those country-year observations that have no intra-state civil conflict in the previous period (columns (3)-(6), Table 7, \citep{nunn2014us}). The event indicator for civil war onset {is defined as one} if it is the first period of a intra-state conflict episode, and zero otherwise. \citet{nunn2014us} estimate a logistic discrete time hazard model for the onset of war, controlling for the previous duration of peace {using a third degree polynomial}. The US wheat aid in year $t$ is instrumented by the product of US wheat production in year $t-1$ and the probability of receiving any US food aid between 1971 and 2006 for each country. For the purpose of the paper, we estimate a binary probit model for the onset of war. The summary statistics for the data used in the onset analysis are given in part (a) of Table \ref{tab:tab8}.


Using the specification of controls in columns (3) of Table 7 in \citet{nunn2014us}, in Part (a) of Table \ref{tab:tab10}, we present the estimated coefficients and average partial effects from both 2SCML probit\footnote{The 2SCML in this example allows intragroup correlation for standard errors, clustered by countries.} and CUGMM with the degree of overidentification equal to unity. For comparison purposes, column (1) of Table \ref{tab:tab10}-(a) gives the estimated average partial effect of US wheat aid on the onset of war as reported by \citet{nunn2014us} using a 2SCML logit approach. For CUGMM, we use $a_i=(x_i',z_i,z_i^2,z_i^3,z_ix_{1i},\textbf{0}_{k_{x}+1}')'$ and $b_i=(\textbf{0}_{k_{x}+3}',1,z_i,x_i')'$ to construct moments. The variables $x_i$ and $z_i$ denote the standardized variables of exogenous regressors and the instrument, $x_{1i}$ is the non-standardized onset duration, and $k_{x}$ is the number of exogenous regressors {(including an intercept)}. Columns (2) and (5) of Table \ref{tab:tab10}-(a) demonstrate that the IV is significantly related to the endogenous regressor of wheat aid at the 1\% significant level by both estimation methods. However, the estimates of interest, the treatment effects of the US wheat aid on onset are statistically insignificant from both estimation methods, same as the result from \citet{nunn2014us} in Column (1).
\footnote{The insignificance of US food aid on onset of civil conflict is also pointed out by \citet{nunn2014us}. However, without reliable evidence on instrument strength, we should be cautious when drawing any conclusions based on standard inference procedures.} Estimates for other coefficients are quite stable and similar across the three sets of results. Finally, Hansen's $J$-statistic is 0.553, less than the critical value $\chi^2_{0.95}(1)=3.84$, thus we cannot reject the null that moments are all valid. This evidence leads to the suspicion that the potential weakness of the IV could be one of the possible reasons for the unstable estimates of the US wheat aid coefficient.

This suspicion is verified by the DJ test. The perturbation for the onset analysis is chosen as described in Section \ref{section3.3}, again using $m=20$ candidate grid points. Panel (a) of Table \ref{tab:tab9} demonstrates that the DJ test cannot reject the null of weak identification. In contrast to the earlier example in Section \ref{sec:lfp}, in this example perturbing the estimators by $\delta_n$ does not lead to a significant change in the corresponding J-statistic, which indicates a lack of curvature and thus identification weakness. Across the entire grid of candidate $\delta_n$ values, the maximum of the DJ statistics is 7.5, which is less than the corresponding 5\% critical value given by $\chi^2_{1-0.05/20}(H+1-p)=11.98$, and which is based on using $H=12$ moments to estimate $p=11$ parameters. However, when we apply the conventional SS, SY and Robust tests to the onset of the civil conflict case, the SS and SY tests all return a rejection of the weak IV hypothesis and the Robust test also rejects the null if the tolerance threshold is greater than $10\%$. As shown in Table 
\ref{tab:tab9}-(a), the reduced form regression Kleibergen-Paap $F$-statistic for SS and SY is 26.07, much larger than 10 and the critical values 16.38 and 8.96 of SY.\footnote{To be consistent with \citet{nunn2014us}, standard errors (s.e.) are computed using clustered s.e. by countries. Kleibergen-Paap $F$-statistic \citep{kleibergen2006generalized} is utilized when allowing for intragroup correlation s.e. The critical values 16.38 and 8.96 of SY test are based on the desired maximal size distortion 5\% and 10$\%$ of a 5$\%$ Wald test, respectively.} The Robust test effective $F$-statistic 26.39 is also larger than its 10$\%$ tolerance critical value 23.11.\footnote{The effective $F$-statistic and critical values are computed using the Stata command "weakivtest" (\citet{pflueger2014robust}). The critical value 23.11 is for the case of effective degrees of freedom one and the tolerance threshold $10\%$. Robust test fails to reject the weak instrument based on the critical value of 5\% tolerance.} In summary, for this onset example, the conventional weak IV tests reject the weak IV hypothesis, while our DJ test suggests the opposite. This serves as a reminder that applying conventional weak IV tests for linear models to binary outcome models should be cautioned. 

Subsequently, we have repeated this analysis for the other 5 specifications considered in \citet{nunn2014us} (columns (4)-(8) of Table 7, \citealt{nunn2014us}), which include different controls and the study for conflict \emph{offset} after a period of war. Our DJ test fails to reject the null of weak instrument for all cases except for column (7), whilst the SS and SY tests all result in a rejection of the weak instrument hypothesis.\footnote{Not all results are reported due to space limitation. SS and SY tests are based on Kleibergen-Paap $F$-statistic}. The DJ test is implemented using the same $a_i$ and $b_i$ as those used in Table \ref{tab:tab10}. The perturbation is again chosen as in Section \ref{section3.3} with $m=20$.} The Robust test also rejects the null in some cases.\footnote{Based on the critical value 23.11 ($\tau=10\%$), the Robust test rejects weak IV of the analysis in columns (4) and (8), but fails to reject in columns (5), (6) and (7). Results are obtained by using the Stata command "weakivtest" \citealt{pflueger2014robust} and clustered s.e.} In part (b) of Table \ref{tab:tab9} and Table \ref{tab:tab10}, we report the estimation results for the probability of \emph{offset} of civil war after a period of war for the specification in column (6) of Table 7, \citealt{nunn2014us}, as well as the test results for weak instrument.

One important result to note is that \citet{nunn2014us} estimate a significant and negative effect for offset of war, indicating that aid prolongs civil wars with 1,000 MT extra of US wheat aid reducing the probability of civil war offset by 0.04 percentage point. The causal effects estimated by the Probit model in Panel (b) of Table \ref{tab:tab10} are also both negative, with statistical significance for the 2SCML results but not for the CU-GMM results. However, as shown in part (b) of Table \ref{tab:tab9}, if one applies the DJ test using the same methodology as above, the DJ statistic varies between 1.50 and 9.46, which is again less than the corresponding critical value of 11.98, indicating that identification may be weak in this example. If identification is indeed weak, as the DJ test suggests, conducting standard inference on the estimated treatment effect is no longer valid. Therefore, the conclusion that US food aid prolongs civil conflict should be viewed with caution.

\section{Conclusion}
Estimating the causal effects of policy relevant treatment variables is the key goal of many empirical studies in economics and other diverse fields. Instrumental variables play a crucial role in the identification and estimation of treatment effects when the treatment is endogenous, but weak instruments have been identified as a potentially serious problem, with consequences including inconsistent estimation {and, consequently, invalid statistical inference.} Consequences and detection of weak identification due to instrument weakness have been extensively studied for linear models, but similar issues have not been thoroughly studied for discrete outcome models. In search for a suitable weak identification test, empirical researchers have often resorted to the inappropriate use of linear model weak IV tests for discrete outcome models, or the use of a linear probability model with a 2SLS estimator treating the discrete outcomes as continuous. The suitability of these linear tests in this nonlinear setting is not usually questioned in many empirical studies (see \citealp{dufour2013weak} and \citealp{li2019bivariate} for additional analysis on the performance of the \citealp{stock2005testing} testing approach in binary models).

This paper proposes a much needed weak identification test in endogenous discrete choice models. The proposed test has desirable asymptotic properties including size control under the null of weak identification, and consistency under the alternative. Moreover, we demonstrate that once the null of weak identification is rejected, standard Wald-based inference can be applied as usual. Our Monte Carlo results demonstrate that, whilst the conventional \cite{stock2005testing} and \cite{staiger1997instrumental} tests are often over-sized, and thus fail to reliably detect weakness, our test always controls size and has reasonable power. We apply this testing approach to two empirical examples in the literature, and demonstrate that there are importance instances where our approach produces contradictory conclusions to the commonly applied linear testing approaches. Analyzing the causal impact of U.S. food aid on civil conflict, our approach fails to reject the null of weak identification, however, several commonly applied linear testing approaches all conclude that identification is not weak.

Another key contribution of the paper is the construction of comprehensive concept of weak identification in discrete choice models, based not only
on the convergence rate of drifting moments, but also on the respective
weight of the key parameters, including variances of error terms and the level of
simultaneity. This allows us to provide a unified GMM estimation framework for examining both linear and nonlinear models, and for comparing the asymptotic properties of GMM estimators against other conventional two-step estimators for endogenous discrete choice models. While building on the general testing strategy of \citet{antoine2017testing}, the test proposed in this paper is based on a null hypothesis of genuine identification weakness, and not the nearly-strong identification null hypothesis analyzed in, e.g., \cite{andrews2012estimation} and \cite{antoine2017testing}.


The conclusion {our research gives} to empirical researchers wishing to evaluate identification weakness in discrete choice models {is clear: the} canonical tests developed for linear models are not suitable for nonlinear models, are likely to be overly optimistic, and can fail to detect genuinely weak identification. Our recommendation is a two-step approach. Conduct our testing approach in a first-step, then, if the null is rejected, one can be very confident that identification is not weak, and conventional inference can proceed as usual. If the null of weak identification cannot be rejected, identification robust inference methods (as proposed in \citealp{stock2000gmm} or \citealp{magnusson2010inference}) would be more suitable to assert the significance of any estimated causal effects.

Furthermore, our asymptotic theory is conformable with the point of view on weak identification defended by \cite{16324}: ``weak instruments should not be thought of as merely a small-sample problem, and the difficulties associated with weak instruments can arise even if the sample size is very large.'' We do see weak identification as a population problem (i.e. independent of the sample size): either the GMM estimator is not consistent (under the null of weak identification) or it is consistent (under the alternative). In this respect, the device of using a drifting DGP, as contemplated in the weak identification literature, can be seen as a way to disentangle point identification (a maintained hypothesis in the framework of weak identification) and existence of a consistent estimator. This point of view may look at odds with the one put forward by \cite{lewbel2019identification} where it is stated that: ``a parameter that is weakly identified (meaning that standard asymptotics provide a poor finite sample approximation to the actual distribution of the estimator) when $n$ = 100 may be strongly identified when $n$ = 1000.'' However, for all practical purpose, the methodological recommendation may not be so different: in our case, it is only for a large enough sample size that our test may allow us to reject the null of weak identification. In these circumstances, the researcher can trust the consistency of the estimator and confidently use {Wald-based} inference.

\subsection*{Acknowledgements}Frazier was supported by the Australian Research Council's Discovery Early Career Researcher Award funding scheme (DE200101070), and the Australian Centre for Excellence in Mathematics and Statistics. Zhang acknowledges travel support provided by the International Association for Applied Econometrics (IAAE). The authors thank the seminar participants at University of Cambridge, CEMFI, University of Oxford, University College of London, University of Surrey, University of Warwick, Boston University, University of Sydney, the 10th French Econometrics Conference PSE, the 2019 IAAE conference, and the 2019 North  American Summer Meeting of the Econometric society for helpful comments. 

\bibliography{reference}

\begin{thebibliography}{50}
\providecommand{\natexlab}[1]{#1}
\providecommand{\url}[1]{\texttt{#1}}
\expandafter\ifx\csname urlstyle\endcsname\relax
  \providecommand{\doi}[1]{doi: #1}\else
  \providecommand{\doi}{doi: \begingroup \urlstyle{rm}\Url}\fi

\bibitem[Ai and Chen(2003)]{ai2003efficient}
C.~Ai and X.~Chen.
\newblock Efficient estimation of models with conditional moment restrictions
  containing unknown functions.
\newblock \emph{Econometrica}, 71\penalty0 (6):\penalty0 1795--1843, 2003.

\bibitem[Andrews and Cheng(2012)]{andrews2012estimation}
D.~W. Andrews and X.~Cheng.
\newblock Estimation and inference with weak, semi-strong, and strong
  identification.
\newblock \emph{Econometrica}, 80\penalty0 (5):\penalty0 2153--2211, 2012.

\bibitem[Andrews and Cheng(2014)]{andrews2014gmm}
D.~W. Andrews and X.~Cheng.
\newblock {GMM} estimation and uniform subvector inference with possible
  identification failure.
\newblock \emph{Econometric Theory}, 30\penalty0 (2):\penalty0 287--333, 2014.

\bibitem[Andrews(2018)]{andrews2018valid}
I.~Andrews.
\newblock Valid two-step identification-robust confidence sets for {GMM}.
\newblock \emph{Review of Economics and Statistics}, 100\penalty0 (2):\penalty0
  337--348, 2018.

\bibitem[Antoine and Renault(2009)]{antoine2009efficient}
B.~Antoine and E.~Renault.
\newblock Efficient {GMM} with nearly-weak instruments.
\newblock \emph{The Econometrics Journal}, 12\penalty0 (s1), 2009.

\bibitem[Antoine and Renault(2012)]{antoine2012efficient}
B.~Antoine and E.~Renault.
\newblock Efficient minimum distance estimation with multiple rates of
  convergence.
\newblock \emph{Journal of Econometrics}, 170\penalty0 (2):\penalty0 350--367,
  2012.

\bibitem[Antoine and Renault(2020)]{antoine2017testing}
B.~Antoine and E.~Renault.
\newblock Testing identification strength.
\newblock \emph{Journal of Econometrics}, 2020.

\bibitem[Arendt(2005)]{arendt2005does}
J.~N. Arendt.
\newblock Does education cause better health? {A} panel data analysis using
  school reforms for identification.
\newblock \emph{Economics of Education Review}, 24\penalty0 (2):\penalty0
  149--160, 2005.

\bibitem[Block et~al.(2013)Block, Hoogerheide, and Thurik]{block2013education}
J.~H. Block, L.~Hoogerheide, and R.~Thurik.
\newblock Education and entrepreneurial choice: {A}n instrumental variables
  analysis.
\newblock \emph{International Small Business Journal}, 31\penalty0
  (1):\penalty0 23--33, 2013.

\bibitem[Blundell and Powell(2004)]{blundell2004endogeneity}
R.~W. Blundell and J.~L. Powell.
\newblock Endogeneity in semiparametric binary response models.
\newblock \emph{The Review of Economic Studies}, 71\penalty0 (3):\penalty0
  655--679, 2004.

\bibitem[Caner(2009)]{caner2009testing}
M.~Caner.
\newblock Testing, estimation in {GMM} and {CUE} with nearly-weak
  identification.
\newblock \emph{Econometric Reviews}, 29\penalty0 (3):\penalty0 330--363, 2009.

\bibitem[Cawley and Meyerhoefer(2012)]{cawley2012medical}
J.~Cawley and C.~Meyerhoefer.
\newblock The medical care costs of obesity: {A}n instrumental variables
  approach.
\newblock \emph{Journal of Health Economics}, 31\penalty0 (1):\penalty0
  219--230, 2012.

\bibitem[Chaudhuri and Renault(2020)]{chaudhuri2017score}
S.~Chaudhuri and E.~Renault.
\newblock Score tests in {GMM}: Why use implied probabilities?
\newblock \emph{Journal of Econometrics}, 2020.

\bibitem[Dufour and Wilde(2018)]{dufour2013weak}
J.-M. Dufour and J.~Wilde.
\newblock Weak identification in probit models with endogenous covariates.
\newblock \emph{AStA Advances in Statistical Analysis}, 102\penalty0
  (4):\penalty0 611--631, 2018.

\bibitem[Finlay and Magnusson(2009)]{finlay2009implementing}
K.~Finlay and L.~M. Magnusson.
\newblock Implementing weak-instrument robust tests for a general class of
  instrumental-variables models.
\newblock \emph{The Stata Journal}, 9\penalty0 (3):\penalty0 398--421, 2009.

\bibitem[Goto and Iizuka(2016)]{goto2016cartel}
U.~Goto and T.~Iizuka.
\newblock Cartel sustainability in retail markets: {E}vidence from a health
  service sector.
\newblock \emph{International Journal of Industrial Organization}, 49:\penalty0
  36--58, 2016.

\bibitem[Hahn and Kuersteiner(2002)]{hahn2002discontinuities}
J.~Hahn and G.~Kuersteiner.
\newblock Discontinuities of weak instrument limiting distributions.
\newblock \emph{Economics Letters}, 75\penalty0 (3):\penalty0 325--331, 2002.

\bibitem[Hansen(1982)]{hansen1982large}
L.~P. Hansen.
\newblock Large sample properties of generalized method of moments estimators.
\newblock \emph{Econometrica: Journal of the Econometric Society}, pages
  1029--1054, 1982.

\bibitem[Hausman(1978)]{hausman1978specification}
J.~A. Hausman.
\newblock Specification tests in econometrics.
\newblock \emph{Econometrica: Journal of the Econometric Society}, pages
  1251--1271, 1978.

\bibitem[Imbens and Newey(2009)]{imbens2009identification}
G.~W. Imbens and W.~K. Newey.
\newblock Identification and estimation of triangular simultaneous equations
  models without additivity.
\newblock \emph{Econometrica}, 77\penalty0 (5):\penalty0 1481--1512, 2009.

\bibitem[Kawaguchi et~al.(2017)Kawaguchi, Matsushita, and
  Naito]{kawaguchi2017moment}
D.~Kawaguchi, Y.~Matsushita, and H.~Naito.
\newblock Moment estimation of the probit model with an endogenous continuous
  regressor.
\newblock \emph{The Japanese Economic Review}, 68\penalty0 (1):\penalty0
  48--62, 2017.

\bibitem[Kinda(2010)]{kinda2010investment}
T.~Kinda.
\newblock Investment climate and {FDI} in developing countries: firm-level
  evidence.
\newblock \emph{World Development}, 38\penalty0 (4):\penalty0 498--513, 2010.

\bibitem[Kleibergen(2005)]{kleibergen2005testing}
F.~Kleibergen.
\newblock Testing parameters in {GMM} without assuming that they are
  identified.
\newblock \emph{Econometrica}, 73\penalty0 (4):\penalty0 1103--1123, 2005.

\bibitem[Kleibergen and Paap(2006)]{kleibergen2006generalized}
F.~Kleibergen and R.~Paap.
\newblock Generalized reduced rank tests using the singular value
  decomposition.
\newblock \emph{Journal of Econometrics}, 133\penalty0 (1):\penalty0 97--126,
  2006.

\bibitem[Lewbel(2019)]{lewbel2019identification}
A.~Lewbel.
\newblock The identification zoo: Meanings of identification in econometrics.
\newblock \emph{Journal of Economic Literature}, 57\penalty0 (4):\penalty0
  835--903, 2019.

\bibitem[Li et~al.(2019)Li, Poskitt, and Zhao]{li2019bivariate}
C.~Li, D.~S. Poskitt, and X.~Zhao.
\newblock The bivariate probit model, maximum likelihood estimation, pseudo
  true parameters and partial identification.
\newblock \emph{Journal of Econometrics}, 209\penalty0 (1):\penalty0 94--113,
  2019.

\bibitem[Lochner and Moretti(2004)]{lochner2004effect}
L.~Lochner and E.~Moretti.
\newblock The effect of education on crime: {E}vidence from prison inmates,
  arrests, and self-reports.
\newblock \emph{American Economic Review}, 94\penalty0 (1):\penalty0 155--189,
  2004.

\bibitem[Magnusson(2010)]{magnusson2010inference}
L.~M. Magnusson.
\newblock Inference in limited dependent variable models robust to weak
  identification.
\newblock \emph{The Econometrics Journal}, 13\penalty0 (3), 2010.

\bibitem[McKenzie and Rapoport(2011)]{mckenzie2011can}
D.~McKenzie and H.~Rapoport.
\newblock Can migration reduce educational attainment? {E}vidence from
  {M}exico.
\newblock \emph{Journal of Population Economics}, 24\penalty0 (4):\penalty0
  1331--1358, 2011.

\bibitem[Miguel et~al.(2004)Miguel, Satyanath, and
  Sergenti]{miguel2004economic}
E.~Miguel, S.~Satyanath, and E.~Sergenti.
\newblock Economic shocks and civil conflict: {A}n instrumental variables
  approach.
\newblock \emph{Journal of Political Economy}, 112\penalty0 (4):\penalty0
  725--753, 2004.

\bibitem[Montiel~Olea and Pflueger(2013)]{olea2013robust}
J.~L. Montiel~Olea and C.~Pflueger.
\newblock A robust test for weak instruments.
\newblock \emph{Journal of Business \& Economic Statistics}, 31\penalty0
  (3):\penalty0 358--369, 2013.

\bibitem[Mroz(1987)]{mroz1987sensitivity}
T.~A. Mroz.
\newblock The sensitivity of an empirical model of married women's hours of
  work to economic and statistical assumptions.
\newblock \emph{Econometrica: Journal of the Econometric Society}, pages
  765--799, 1987.

\bibitem[Newey and McFadden(1994)]{newey1994large}
W.~K. Newey and D.~McFadden.
\newblock Large sample estimation and hypothesis testing.
\newblock \emph{Handbook of Econometrics}, 4:\penalty0 2111--2245, 1994.

\bibitem[Newey and Windmeijer(2009)]{newey2009generalized}
W.~K. Newey and F.~Windmeijer.
\newblock Generalized method of moments with many weak moment conditions.
\newblock \emph{Econometrica}, 77\penalty0 (3):\penalty0 687--719, 2009.

\bibitem[Newey et~al.(1999)Newey, Powell, and Vella]{newey1999nonparametric}
W.~K. Newey, J.~L. Powell, and F.~Vella.
\newblock Nonparametric estimation of triangular simultaneous equations models.
\newblock \emph{Econometrica}, 67\penalty0 (3):\penalty0 565--603, 1999.

\bibitem[Nunn and Qian(2014)]{nunn2014us}
N.~Nunn and N.~Qian.
\newblock {US} food aid and civil conflict.
\newblock \emph{American Economic Review}, 104\penalty0 (6):\penalty0 1630--66,
  2014.

\bibitem[Pflueger and Wang(2015)]{pflueger2014robust}
C.~E. Pflueger and S.~Wang.
\newblock A robust test for weak instruments in {S}tata.
\newblock \emph{The Stata Journal}, 15\penalty0 (1):\penalty0 216--225, 2015.

\bibitem[Powell et~al.(2005)Powell, Tauras, and Ross]{powell2005importance}
L.~M. Powell, J.~A. Tauras, and H.~Ross.
\newblock The importance of peer effects, cigarette prices and tobacco control
  policies for youth smoking behavior.
\newblock \emph{Journal of Health Economics}, 24\penalty0 (5):\penalty0
  950--968, 2005.

\bibitem[Rivers and Vuong(1988)]{rivers1988limited}
D.~Rivers and Q.~H. Vuong.
\newblock Limited information estimators and exogeneity tests for simultaneous
  probit models.
\newblock \emph{Journal of Econometrics}, 39\penalty0 (3):\penalty0 347--366,
  1988.

\bibitem[Ruseski et~al.(2014)Ruseski, Humphreys, Hallman, Wicker, and
  Breuer]{ruseski2014sport}
J.~E. Ruseski, B.~R. Humphreys, K.~Hallman, P.~Wicker, and C.~Breuer.
\newblock Sport participation and subjective well-being: {I}nstrumental
  variable results from german survey data.
\newblock \emph{Journal of Physical Activity and Health}, 11\penalty0
  (2):\penalty0 396--403, 2014.

\bibitem[Staiger and Stock(1997)]{staiger1997instrumental}
D.~Staiger and J.~H. Stock.
\newblock Instrumental variables regression with weak instruments.
\newblock \emph{Econometrica}, 65\penalty0 (3):\penalty0 557--586, 1997.

\bibitem[Stock and Andrews(2005)]{16324}
J.~Stock and D.~Andrews.
\newblock Inference with weak instruments.
\newblock \emph{Advances in Economics and Econometrics, Theory and
  Applications: Ninth World Congress of the Econometric Society, Vol III,
  Cambridge University Press, url
  http://www.economics.harvard.edu/faculty/stock/files/worldcongresspaper9.pdf},
  2005.

\bibitem[Stock and Wright(2000)]{stock2000gmm}
J.~H. Stock and J.~H. Wright.
\newblock {GMM} with weak identification.
\newblock \emph{Econometrica}, 68\penalty0 (5):\penalty0 1055--1096, 2000.

\bibitem[Stock and Yogo(2005)]{stock2005testing}
J.~H. Stock and M.~Yogo.
\newblock Testing for weak instruments in linear {IV} regression. {C}hapter 5
  in {I}dentification and {I}nference in {E}conometric {M}odels: {E}ssays in
  {H}onor of {T}homas {J}. {R}othenberg, edited by {DWK} {A}ndrews and {JH}
  {S}tock, 2005.

\bibitem[Terza et~al.(2008)Terza, Basu, and Rathouz]{terza2008two}
J.~V. Terza, A.~Basu, and P.~J. Rathouz.
\newblock Two-stage residual inclusion estimation: Addressing endogeneity in
  health econometric modeling.
\newblock \emph{Journal of Health Economics}, 27\penalty0 (3):\penalty0
  531--543, 2008.

\bibitem[van~der Vaart and Wellner(1996)]{van1996weak}
A.~W. van~der Vaart and J.~A. Wellner.
\newblock Weak convergence.
\newblock In \emph{Weak Convergence and Empirical Processes}, pages 16--28.
  Springer, 1996.

\bibitem[Windmeijer(2019)]{windmeijer2019two}
F.~Windmeijer.
\newblock Two-stage least squares as minimum distance.
\newblock \emph{The Econometrics Journal}, 22\penalty0 (1):\penalty0 1--9,
  2019.

\bibitem[Wooldridge(2010)]{wooldridge2010econometric}
J.~M. Wooldridge.
\newblock \emph{Econometric analysis of cross section and panel data}.
\newblock MIT Press, 2010.

\bibitem[Wooldridge(2014)]{wooldridge2014quasi}
J.~M. Wooldridge.
\newblock Quasi-maximum likelihood estimation and testing for nonlinear models
  with endogenous explanatory variables.
\newblock \emph{Journal of Econometrics}, 182\penalty0 (1):\penalty0 226--234,
  2014.

\bibitem[Wooldridge(2015)]{wooldridge2015control}
J.~M. Wooldridge.
\newblock Control function methods in applied econometrics.
\newblock \emph{Journal of Human Resources}, 50\penalty0 (2):\penalty0
  420--445, 2015.

\end{thebibliography}

\appendix
\numberwithin{equation}{section}
\section{Appendix}The appendix contains proofs for the main results in the paper.
\subsection{Lemmas} We first give several lemmas that are used to prove the main results.
\begin{lemma}\label{lem_uniform}
	Under \textbf{Assumption 1}, for $\nu_n(\theta):=\frac{1}{\sqrt{n}}\sum_{i=1}^{n}\left(g_i(\theta)-\mathbb{E}[g_i(\theta)]\right)$, $$\nu_n(\theta)\Rightarrow \nu(\theta),$$ for $\nu(\theta)$ a mean-zero Gaussian process with (uniformly) bounded covariance kernel $S(\theta,\tilde{\theta})$.
\end{lemma}
\begin{proof}[Proof of Lemma \ref{lem_uniform}]
	First, recall that  for $g_i(\theta)=[a_i,b_i]r_i(\theta)$, with $r_i(\theta):=[r_{1i}(\theta),r_{2i}(\theta_2)]'$ so that $$\|g_i(\theta)\|=\left\|[a_i,b_i] r_i(\theta)\right\|\leq\|[a_i,b_i]\|\|r_i(\theta)\|.$$ Under Assumption (A1), $[a_i,b_i]$ is i.i.d. and $\mathbb{E}[\|[a_i,b_i]\|^2]<\infty$. The result then follows if we can demonstrate that $r_i(\theta)$ is Donsker.
		
Consider the re-parameterization $\vartheta=(\vartheta_1',\vartheta_2')'$, where $\vartheta_1:=(\alpha+\tilde{\rho},\beta'-\tilde{\rho}\pi',\tilde{\rho}\xi')'$, and $\vartheta_2:=(\pi',\xi')'$. By compactness of $\Theta$, the new parameter space $V:=\{\vartheta=(\vartheta_1',\vartheta_2')':~\theta\in\Theta\}$ is also compact. Denote $w_{1i}=(y_{2i},x_i',-z_i')'$ and $w_{2i}=(x_i',z_i')'$. Rewrite $\Phi[y_{2i}(\alpha+\tilde{\rho})+x_i'(\beta'-\tilde{\rho}\pi')- z_{i}'\xi\tilde{\rho}]=\Phi(w_{1i}'\vartheta_{1})$. By abuse of notation, define $r_{1i}(\vartheta_1)=y_{1i}-\Phi(w_{1i}'\vartheta_1)$, $r_{2i}(\vartheta_2):=y_{2i}-w_{2i}'\vartheta_2$, and define the class of functions $$\mathcal{F}:=\big{\{}r_{i}(\vartheta)=(r'_{1i}(\vartheta_1),r'_{2i}(\vartheta_2))':~\vartheta\in V\big{\}},$$ from the compactness of $V$, $(\mathcal{F},\|\cdot\|)$ is totally bounded with $\|\cdot\|$ the Euclidean norm.
	
First, focus on $r_{1i}(\vartheta_1)$. For every $w_{1i}$ and for $\vartheta_1,\bar{\vartheta}_1\in V_1$, with $V_1$ a subspace of $V$ associated with $\vartheta_1$, without loss of generality, suppose $w_{1i}'\vartheta\geq w_{1i}'\bar{\vartheta}$. Then,
	\begin{flalign*}
	\|r_{1i}(\vartheta_1)-r_{1i}(\bar{\vartheta}_1)\| &=|\Phi(w_{1i}'\vartheta_1)-\Phi(w_{1i}'\bar{\vartheta}_1)|\\
	&=\left|\int_{w_{1i}'\bar{\vartheta}_1}^{w_{1i}'\vartheta_1}\phi(t)dt\right|=\phi(c)|w_{1i}'(\vartheta_1-\bar{\vartheta}_1)|\leq C\|w_{1i}\|\|\vartheta_1-\bar{\vartheta}_1\|,
	\end{flalign*}
	for $c\in(w_{1i}'\bar{\vartheta}_1, w_{1i}'\vartheta_1)$ and some constant $C>0$. For $P$ the law of $(w_{1i}',w_{2i}')$, by \textbf{Assumption (A.1)}, we know that
	$$\mathbb{E}_{P}[\|w_{1i}\|^2]<\infty.$$ Now, consider $r_{2i}(\vartheta_2)$ and note that, for $\vartheta_2,\bar{\vartheta}_2\in V_2$, with $V_2$ a subspace of $V$ associated with $\vartheta_2$,
	\begin{flalign*}
	\|r_{2i}(\vartheta_2)-r_{2i}(\bar{\vartheta}_2)\|\leq \|w_{2i}\|\|\vartheta_2-\bar{\vartheta}_2\|.
	\end{flalign*}It then follows from \textbf{Assumption (A.1)} that $$\mathbb{E}_{P}[\|w_{2i}\|^2]<\infty.$$
	
	Defining $L=\max\{\|w_{1i}\|,\|w_{2i}\|\}$, $\vartheta=(\vartheta'_1,\vartheta'_2)'$ and $\bar{\vartheta}=(\bar{\vartheta}'_1,\bar{\vartheta}'_2)'$, we have that $\mathbb{E}[L]<\infty$ and  $$\|r_i(\vartheta)-r_i(\bar{\vartheta})\|\leq L\|\vartheta-\bar{\vartheta}\|.$$ This Lipschitz property, together with the compactness of $V$ implies that, by Theorem 2.7.11 of \citet{van1996weak}, $\mathcal{F}$ is $P$-Donsker. For $g_i(\theta)=[a_i,b_i]r_i(\theta)$, we then have that  $$\nu_n(\theta):=\sqrt{n}\left(\bar{g}_n(\theta)-\mathbb{E}[g_{i}(\theta)]\right)\Rightarrow \nu(\theta),$$
	for $\theta\in\Theta$ where $\nu(\theta)$ denotes a Gaussian process with zero mean and variance kernel $$S(\theta,\tilde\theta):=\mathbb{E}\left\{(g_i(\theta)-\mathbb{E}[g_i(\theta)])(g_i(\tilde\theta)-\mathbb{E}[g_i(\tilde\theta)])'\right\}.$$ By the continuity of $S(\theta,\theta)$ in $\theta$, \textbf{Assumption (A.1)}, and the compactness of $\Theta$, we have $$0<\sup_{\theta,\tilde\theta\in\Theta}\|S(\theta,\tilde\theta)\|<\infty.$$
\end{proof}

The following results demonstrates that \textbf{Assumption 4} in the main text is satisfied under \textbf{Assumption 1}
\begin{lemma}\label{lem_uniform_deriv}
	Under \textbf{Assumption 1}, if $\tilde{a}_i:=\tilde{a}(y_{2i},z_i,x_i)$ satisfies $\E_n\left[\|\tilde{a}_iz_i'\xi^0(y_{2i},z_i',x_i')'\|^2\right]<\infty$, for $\Psi_n(\eta,\theta^0_2):=\frac{1}{\sqrt{n}}\sum_{i=1}^{n}\left\{\tilde{a}_i\phi_i(\eta,\theta^0_2)z_i'\xi^0-\mathbb{E}_n[\tilde{a}_i\phi_i(\eta,\theta^0_2)z_i'\xi^0]\right\}$, $$\Psi_n(\eta,\theta^0_2)\Rightarrow \Psi(\eta,\theta^0_2),$$ for $\Psi(\eta,\theta^0_2)$ a mean-zero Gaussian process over $\Upsilon(\theta^0_2)$.
\end{lemma}
\begin{proof}[Proof of Lemma \ref{lem_uniform_deriv}]
Similar to the proof of Lemma \ref{lem_uniform}, it suffices to show that the class of functions
$$\mathcal{F}:=\big{\{}r_{3i}(\eta)=\tilde{a}_{i}\phi(\eta,\theta^0_2)z_i'\xi^0:~\eta\in \Upsilon(\theta_2^0)\big{\}},$$ is Donsker, where $\eta:=(\tilde\rho,\tilde\rho+\alpha,\beta'-\tilde\rho\pi^{0'})'$. Hence, we only sketch the details. 	

Let $w_i:=(-z_i'\xi^0,y_{2i},x_i')'$. For every $w_{i}$ and for $\eta,\bar{\eta}\in\Upsilon(\theta^0_2)$, without loss of generality, suppose $w_{i}'\eta\geq w_{i}'\bar{\eta}$. Let $\phi'(x)$ denote the derivative of the density function $\phi(x)$. Then, for $c\in(w_i'\bar{\eta}, w_i'\eta)$,
\begin{flalign*}
\|\tilde{a}_i\phi(\eta,\theta^0_2)z_i'\xi^0-\tilde{a}_i\phi(\bar{\eta},\theta^0_2)z_i'\xi^0\| &=\phi'(c)\|\tilde{a}_iz_i'\xi^0w_{i}'(\eta-\bar{\eta})\|\leq C\|\tilde{a}_iz_i'\xi^0w_{i}\|\|\eta-\bar{\eta}\|,
\end{flalign*}for some constant $C>0$, and where the equality follows by the intermediate value theorem, and the inequality from Cauchy-Schwartz. For $P$ the joint law of $w_{i}$, by \textbf{Assumption (A.1)}, the compactness of $\Theta_2$ (\textbf{Assumption (A.4)}), and the moment hypothesis for $\tilde{a}_i$.,
$$\mathbb{E}_{P}[\|\tilde{a}_iz_i'\xi^0w_{i}'\|^2]<\infty.$$ The remainder of the proof follows that of Lemma \ref{lem_uniform} and is omitted for brevity.
\end{proof}

For $A_n=R\Lambda_n$, the following result demonstrates that, regardless of the interpretation for instrument weakness, for any consistent estimator the sample estimator $\partial\bar{g}_n(\theta_n)/\partial\theta'A_n$ is a consistent estimator of $M$ in \textbf{Assumption 5}.

\begin{lemma}\label{lems_old}
	If $\{\theta_n\}$ is such that $\|\theta_n-\theta^0\|=o_p(1)$, then under \textbf{Assumptions 1-6}: $$M=\underset{{n\rightarrow\infty}}{\text{plim}}\frac{\partial \bar{g}_{n}(\theta_n)}{\partial\theta'}A_n,\text{ where }A_n=R\Lambda_n.$$
\end{lemma}
\begin{proof}[Proof of Lemma \ref{lems_old}] Let $\bar{g}_{n}(\theta)=(\bar{g}_{1n}(\theta),\bar{g}_{2n}(\theta),...,\bar{g}_{H,n}(\theta))'$. The mean value expansion of $\frac{\partial \bar{g}_{l,n}(\theta_n)}{\partial\theta'}$ at $\theta^0$ yields
	$$\frac{\partial \bar{g}_{l,n}(\theta_n)}{\partial\theta'}=\frac{\partial \bar{g}_{l,n}(\theta^0)}{\partial\theta'}+(\theta_n-\theta^0)'\frac{\partial^2 \bar{g}_{l,n}(\widetilde{\theta}_n)}{\partial\theta'\partial\theta},~~~~l=1,2,...,H
	$$
	where $\widetilde{\theta}_n$ is component-by-component between $\theta^0$ and $\theta_n$. By the structure of the moment $\bar{g}_n(\theta)$, the smoothness conditions on $\Phi(\cdot)$ and its derivatives, $a_i$ and $b_i$ are all measurable, it is not hard to prove that $\|\theta_n-\theta^0\|=o_p(1)$ implies the Hessian multiplied by $A_n$, $\frac{\partial^2 \bar{g}_{l,n}(\widetilde{\theta}_n)}{\partial\theta'\partial\theta}A_n=O_p(1)$ for $l=1,2,...,H$. Therefore, $\|\theta_n-\theta^0\|=o_p(1)$ and Lemma \ref{local-lem} implies the result is satisfied.
\end{proof}

\begin{lemma}\label{lems:order}
	Under \textbf{Assumptions 1-6}, and for $\Lambda_n$ as in Lemma \ref{local-lem}, $\sqrt{n}\Lambda_{n}^{-1}(\hat\zeta_n-\zeta^0)=O_p(1)$.
\end{lemma}
\begin{proof}[Proof of Lemma \ref{lems:order}] The result is a consequence of Proposition \ref{lem:cons} and Lemma \ref{lems_old}, and the following inequality:
\begin{flalign*}
J_n(\zeta^0,\zeta^0)\geq J_n(\hat\zeta_n,\hat\zeta_n)= J_n(\hat\zeta_n,\zeta^0)\{1+o_p(1)\},
\end{flalign*}	which follows from the definition of $\hat\zeta_n$ and the consistency of $\hat\zeta_n$ in Proposition \ref{lem:cons}. For some component-by-component intermediate value $\zeta^*_n$,
$$
\sqrt{n}\bar{g}_n(\hat\zeta_n)=\sqrt{n}\bar{g}_n(\zeta^0)-\sqrt{n}\frac{\partial \bar{g}_n(\zeta^*_n)}{\partial\zeta'}(\zeta^0-\hat\zeta_n),
$$ and we can apply the inequality $\|a-b\|\ge -\|a\|+\|b\|$ to obtain
\begin{flalign*}
J^{1/2}_n(\hat\zeta,\zeta^0)&\ge -\|\sqrt{n}\bar{g}_n(\zeta^0)\|_{\Omega_n}+\|\sqrt{n}{\partial \bar{g}_n(\zeta^*_n)}/{\partial\zeta'}(\zeta^0-\hat\zeta_n)\|_{\Omega_n},
\end{flalign*}
where $\Omega_n=S_n^{-1}(\zeta^0)$, $\|x\|_{\Omega_n}:=(x'\Omega_nx)^{1/2}$ and where we have used the fact that (with probability converging to unity) $\lambda_{\text{min}}(\Omega_n)>0$. By the consistency of $\hat\zeta_n$ proved in Proposition \ref{lem:cons} and Lemma \ref{lems_old}, and for $M$ as defined in Lemma \ref{local-lem}, we have
\begin{flalign*}
\|\sqrt{n}{\partial \bar{g}_n(\zeta^*_n)}/{\partial\zeta'}(\zeta^0-\hat\zeta_n)\|_{\Omega_n}&=\|{\partial \bar{g}_n(\zeta^*_n)}/{\partial\zeta'}\Lambda_n\sqrt{n}\Lambda_n^{-1}(\zeta^0-\hat\zeta_n)\|_{\Omega_n}\\&=\|M\sqrt{n}\Lambda^{-1}_n(\hat\zeta_n-\zeta^0)+o_p\left(\sqrt{n}\Lambda^{-1}_n(\hat\zeta_n-\zeta^0)\right)\|_{\Omega_n}\\&\geq C\|\sqrt{n}\Lambda^{-1}_n(\hat\zeta_n-\zeta^0)\{1+o_p(1)\}\|
\end{flalign*}
for some constant $C>0$, where the last inequality follows from the fact that $M$ is full column rank and the fact that $\lambda_{\text{min}}(\Omega_n)>0$ (with probability converging to unity). Applying the above inequality into the first inequality, and using the fact that $J_n(\zeta^0,\zeta^0)=O_p(1)$, we obtain
$$
O_p(1)\geq C\|\sqrt{n}\Lambda^{-1}_n(\hat\zeta_n-\zeta^0)\{1+o_p(1)\}\|.
$$
\end{proof}

\subsection{Proofs of Main Results}\label{A2}

\begin{proof}[Proof of Proposition \ref{prop1}] First, note that
\[
\sqrt{n}\left[ \bar{g}_{n}\left( \hat{\theta}_{n}^{\delta }\right) -\bar{g}%
_{n}\left( \hat{\theta}_{n}\right) \right] =\sqrt{n}\frac{\partial \bar{g}%
	_{n}}{\partial \eta _{1}}\left( \eta _{1n}^{\ast },\hat{\eta}_{2n},\hat{%
	\eta}_{3n},\hat{\theta}_{2n}\right) \delta _{n},\]
where $\eta _{1n}^{\ast }$ denotes a component-by-component intermediate
value between the first coefficients of $\hat{\theta}_{n}$ and $\hat{\theta}%
_{n}^{\delta }$. Recall $\delta_n\rightarrow0$ as $n\rightarrow\infty$. Thus, we only have to prove that
\[
\sqrt{n}\frac{\partial \bar{g}_{n}}{\partial \eta _{1}}\left( \eta
_{1n}^{\ast },\hat{\eta}_{2n},\hat{\eta}_{3n},\hat{\theta}_{2n}\right)
=O_{p}\left(1\right) .
\]
 For this purpose, we write the Taylor expansion
\begin{equation}
\sqrt{n}\frac{\partial \bar{g}_{n}}{\partial \eta _{1}}\left( \eta
_{1n}^{\ast },\hat{\eta}_{2n},\hat{\eta}_{3n},\hat{\theta}_{2n}\right) =%
\sqrt{n}\frac{\partial \bar{g}_{n}}{\partial \eta _{1}}\left( \eta
_{1n}^{\ast },\hat{\eta}_{2n},\hat{\eta}_{3n},\theta _{2}^{0}\right) +%
\frac{\partial ^{2}\bar{g}_{n}}{\partial \eta _{1}\partial\theta_2'}\left( \eta
_{1n}^{\ast },\hat{\eta}_{2n},\hat{\eta}_{3n},\theta _{2n}^{\ast
}\right) \sqrt{n}\left( \hat{\theta}_{2n}-\theta _{2}^{0}\right),
\label{Taylweak}
\end{equation}
for some intermediate value $\theta _{2n}^{\ast
}$. By construction, the separation of estimators\ of $\theta _{1}$\
(or $\eta _{1}$) and $\theta _{2}$\ (see Remark 3 in Section 2.2) implies
that $\sqrt{n}( \hat{\theta}_{2n}-\theta _{2}^{0}) =O_{p}(1).$
It is also worth noting that application of Lemma A1 of \citet{stock2000gmm}
would allow us to prove this result in an even more general context.

To see that the second part of the RHS of (\ref{Taylweak}) is $O_{p}(1)$, note the following: (i), ${\partial ^{2}\bar{g}_{n}}/{\partial \eta _{1}\partial\theta_2'}$ is continuous in $\eta$ and $\theta_2$; (ii), $\Upsilon(\theta_2^0)\times\Theta_2$ is compact; (iii), verify that $\|{\partial ^{2}\bar{g}_{n}}/{\partial \eta _{1}\partial\theta_2'}\|\le 2\|\tilde{a}(y_{2i},z_i,x_i)z_i'\|$, where $\E[\|\tilde{a}(y_{2i},z_i,x_i)z_i'\|]<\infty$ by hypothesis. From the i.i.d. nature of the data, the uniform law of large number (ULLN) then implies that the second derivative in question converges uniformly, and together with the fact that $\sqrt{n}( \hat{\theta}_{2n}-\theta _{2}^{0}) =O_{p}(1)$ implies that the second term on the RHS of (\ref{Taylweak}) is $O_p(1)$.

Finally, it is straightforward to deduce that
\begin{flalign*}
\sup_{\eta \in \Upsilon \left( \theta _{2}^{0}\right) }\left\Vert \sqrt{n}%
\frac{\partial \bar{g}_{n}}{\partial \eta _{1}}\left( \eta,\theta
_{2}^{0}\right) \right\Vert &\leq \sup_{\eta \in \Upsilon \left( \theta
	_{2}^{0}\right) }\left\Vert \E_{n}\left\{ \sqrt{n}\frac{\partial \bar{g}_{n}}{%
	\partial \eta _{1}}\left( \eta,\theta _{2}^{0}\right) \right\} \right\Vert
\\&+\sup_{\eta \in \Upsilon \left( \theta _{2}^{0}\right) }\left\Vert \sqrt{n}%
\frac{\partial \bar{g}_{n}}{\partial \eta _{1}}\left( \eta,\theta
_{2}^{0}\right) -\E_{n}\left\{ \sqrt{n}\frac{\partial \bar{g}_{n}}{\partial
	\eta _{1}}\left( \eta,\theta _{2}^{0}\right) \right\} \right\Vert.
\end{flalign*}The first term is $O(1)$ under the null, while the second term is $O_p(1)$ under \textbf{Assumption 4} (or \textbf{Assumption 1} and Lemma \ref{lem_uniform_deriv}).
\end{proof}

\begin{proof}[Proof of Theorem \ref{thm0}]The result follows direction from \textbf{Proposition \ref{prop1}}. To see this, note that, by definition,
	\begin{equation}
	J_{n}\left( \hat{\theta}_{n},\hat\theta_n\right) \leq J_{n}\left[\left( \eta _{1}^{0},\tilde{%
		\eta}_{2n},\tilde{\eta}_{3n},\tilde{\theta}_{2n}\right),\theta^0\right]  \label{size},
	\end{equation}
	where $( \tilde{\eta}_{2n},\tilde{\eta}_{3n},\tilde{\theta}%
	_{2n}) $ denotes the infeasible CUGMM\ estimator of $\left(
	\eta _{2},\eta _{3},\theta _{2}\right) $\ that would result if we knew $%
	\eta _{1}^{0}$; i.e.,
	\[
	\left( \tilde{\eta}_{2n},\tilde{\eta}_{3n},\tilde{\theta}_{2n}\right)
	=\argmin_{\left( \eta _{2},\eta _{3},\theta _{2}\right) }J_{n}\left[\left( \eta
	_{1}^{0},\eta _{2},\eta _{3},\theta _{2}\right),\left( \eta
	_{1}^{0},\eta _{2},\eta _{3},\theta _{2}\right) \right].
	\]
	However, under \textbf{Assumptions 1-3}, the standard theory of the J-test for over-identification test for estimation of $(\eta_2,\eta_3,\theta_2)$ yields
	\[
	J_{n}\left[\left( \eta _{1}^{0},\tilde{\eta}_{2n},\tilde{\eta}_{3n},\tilde{\theta%
	}_{2n}\right),\theta^0\right] \overset{d}{\rightarrow}\chi ^{2}\left( H+1-p\right),
	\]
	where $\overset{d}{\rightarrow}$ denotes convergence in
	distribution. Hence, the result in Proposition \ref{prop1} implies that $J_n^\delta$ is asymptotically bounded above by a $\chi^2(H+1-p)$ random variable, which yields the necessary size control for the test $W_n^\delta$.
\end{proof}

\begin{proof}[Proof of Proposition \ref{lem:cons}]We work in the rotated parameter space, collected as $\zeta:=(\eta',\theta_2')'$, and note here that the result can be moved to the original parameters through the change of basis  $\theta=R^{}\zeta$ in \eqref{eq:basis}.
	
	Firstly, we demonstrate that there exist a deterministic diagonal matrix $\widetilde{\Lambda}_n$, a vector function $\gamma(\zeta)$, continuous in $\zeta$, and a vector function $q_2(\eta_2,\eta_3)$, continuous in $(\eta_2,\eta_3)$, such that under our drifting DGP,\footnote{Technically, the functions $\gamma(\cdot)$ and $q_2(\cdot,\cdot)$ will be $n$-dependent, since we are in the context of a drifting DGP. However, to lessen the notational burden we suppress the dependence of these functions on $n$.}	 $$\mathbb{E}_n\left[\bar{g}_n(\zeta)\right]=\frac{\widetilde{\Lambda}_n}{\sqrt{n}}\gamma(\zeta)+q_{2}(\eta_2,\eta_3),$$
and
$$\gamma(\zeta)=0\text{ and }q_{2}(\eta_2,\eta_3)=0\quad\iff\quad\zeta=0,$$
where $\widetilde{\Lambda}_{n}$  has minimal and maximal eigenvalues, denoted by ${\lambda}_{\text{min}}[\widetilde{\Lambda}_n]$ and ${\lambda}_{\text{max}}[\widetilde{\Lambda}_n]$, respectively, that satisfy:
	$$\lim_{n\rightarrow\infty}{\lambda}_{\text{min}}[\widetilde{\Lambda}_n]=\infty\text{ and }\lim_{n\rightarrow\infty}{\lambda}_{\text{max}}[\widetilde{\Lambda}_n]/{\sqrt{n}}<\infty. $$
	After this, we can apply a similar strategy to Theorem 2.1 of \citet{antoine2012efficient} to establish estimation consistency for the parameters $\zeta^0:=(\eta^{0'},\theta_2^{0'})'$.
	To simplify the calculations, we establish this result in the case where $x_i=1$, for all $i$, and scalar $z_i$, which yields the moment functions: $g_i(\theta)=(g_{1i}(\eta,\theta_2)',g_{2i}(\theta)')'$, where $$g_{1i}(\eta,\theta_2)=a _ { i } \left( y _ { 1i } - \Phi \left[ -\eta_1z_i\xi+\eta_2 y _ { 2i }+\eta_3 \right]\right) ,\;g_{2i}(\theta)=\begin{pmatrix}y _ { 2i } - \pi - \xi z _ { i }\\ z _ { i } \left( y _ { 2i } - \pi - \xi z _ { i}\right)
	\end{pmatrix}.$$
 From the identification condition in \textbf{Assumption 2},  $\theta_2^0=(\pi^0,\xi^0)'$ can be directly identified from $\mathbb{E}_n[g_{2i}(\theta)]=0$, which would yield least square estimators $$\hat{\theta}_2:=\begin{pmatrix}
	\hat{\pi}_n\\\hat{\xi}_n
	\end{pmatrix}=\begin{pmatrix}\bar{y}_{2n}-\hat{\xi}_n\bar{z}_n\\\sum_{i=1}^{n}(z_i-\bar{z}_n)(y_{2i}-\bar{y}_{2n})/\sum_{i=1}^{n}(z_i-\bar{z}_n)^2
	\end{pmatrix},
	$$
for $\bar{z}_n=\sum_{i=1}^nz_i/n$ and $\bar{y}_{2n}=\sum_{i=1}^ny_{2i}/n$,	which are clearly $\sqrt{n}$-consistent and asymptotically normal under \textbf{Assumptions 1 and 2}.
	
	Now, define the stochastic process $\nu_n(\eta,\theta_2)=(\nu_{1n}(\eta,\theta_2)',\nu_{2n}(\theta_2))'$ to be conformable to $g_i(\eta,\theta_2)=(g_{1i}(\eta,\theta_2)',g_{2i}(\theta_2)')'$, where by abuse of notation, we write $g_{2i}(\theta)$ as $g_{2i}(\theta_2)$. From the $\sqrt{n}$-consistency of $(\hat{\pi}_n,\hat{\xi}_n)'$ and stochastic equicontinuity of $\nu_{1n}(\eta,\theta_2)$, we can restrict our analysis on the uniform behavior of $\nu_{1n}(\eta,\theta_2)$ to the set $\Upsilon_n:=\{(\eta,\theta_2):\eta\in\Upsilon(\theta_2),\theta_2\in\Theta_{2,n^{}}\}$, for $\Upsilon(\theta_2)$ as defined above equation \eqref{nullweak}, and where for some $\delta>0$ and $\delta=o(1)$, $$\Theta_{2,n}:=\left\{\theta_{2n}:\|\theta_{2n}-\theta_2^0\|\leq \delta/\sqrt{n}\right\}.$$ In the remainder, we take $\theta_{2n}$ to be an arbitrary sequence in $\Theta_{2,n^{}}$.
	
	For $\theta_{2n}$ as above, recall that, using the decomposition in equation \eqref{decomp}, for some $\bar\eta_1$ such that $\eta_1^0\le \bar\eta_1\le \eta_1$,
	\begin{flalign}
	m_{1n}(\eta,\theta_{2n})
&=m_{1n}(\eta,\theta^{0}_2)+m_{1n}(\eta,\theta_{2n})-m_{1n}(\eta,\theta^{0}_2)\nonumber\\
&=q_{11,n}(\eta)/\varsigma_{n}+q_{12,n}(\eta_2,\eta_3)+o_p(n^{-1/2})\nonumber\\
&=(\eta_1-\eta^0_1)\E_{n}\left[\frac{1}{n}\sum_{i=1}^{n}\tilde{a}_i\phi_{i}(\bar\eta_1,\eta_2,\eta_3;\theta^0_{2})z_i\xi^{0}\right]+q_{12,n}(\eta_2,\eta_3)+o_p(n^{-1/2}).\label{eq:d0}
	\end{flalign} Moreover, by \textbf{Assumption 5}, uniformly over $\bar{\eta}=(\bar{\eta}_1,\eta_2',\eta_3')'\in\Upsilon(\theta^0_2)$,
	$$
	\left\|\E_{n}\left[\frac{1}{n}\sum_{i=1}^{n}\tilde{a}_i\phi_{i}(\bar\eta,\theta^0_{2})z_i\xi^{0}\right]\varsigma_{n}-V^0(\bar\eta)\right\|=o(1)
	$$so that
	\begin{flalign}
	m_{1n}(\eta,\theta_{2n})&=\varsigma^{-1}_{n}(\eta_1-\eta^0_1)V^0(\bar\eta)+q_{12,n}(\eta_2,\eta_3)+o_p(n^{-1/2})\label{eq:d1}.
	\end{flalign}

	Now, decompose $\sqrt{n}\bar{g}_{1n}(\eta,\theta_{2n})$ as
	\begin{flalign*}
	\sqrt{n}\bar{g}_{1n}(\eta,\theta_{2n})&=\sqrt{n}\left\{\bar{g}_{1n}(\eta,\theta_{2n})-m_{1n}(\eta,\theta_{2n}) \right\}+\sqrt{n}m_{1n}(\eta,\theta_{2n}),
	\end{flalign*}and apply equation \eqref{eq:d1} to obtain
	\begin{flalign*}
	\sqrt{n}\bar{g}_{1n}(\eta,\theta_{2n})&=\nu_{1n}(\eta,\theta^0_2)+\sqrt{n}m_{1n}(\eta,\theta^0_2)+o_p(1)\\&=\nu_{1n}(\eta,\theta^0_2)+\frac{\sqrt{n}}{\varsigma_{n}}V^0(\bar\eta)(\eta_1-\eta^0_1)\{1+o_p(1)\}+\sqrt{n}q_{12,n}(\eta_2,\eta_3).
	\end{flalign*}Recall that by Lemma \ref{lem_uniform}, $\nu_n(\eta,\theta^0_2)\Rightarrow\nu(\eta,\theta^0_2)$, and hence is $O_p(1)$ uniformly for $\eta\in\Upsilon(\theta^0_2)$.

	Define $\bar{\lambda}_n:=\sqrt{n}/\varsigma_{n}$, which satisfies $\bar{\lambda}_n\rightarrow\infty$, as $n\rightarrow\infty$, where $\bar{\lambda}_n=o(\sqrt{n})$ by the definition of $\varsigma_n$ in \textbf{Assumption 5}. Now, define the matrix  $$\widetilde{\Lambda}_n:=\begin{bmatrix}\bar{\lambda}_n\mathbf{I}_{\text{dim}(g_1)}&\mathbf{O}\\\mathbf{O}&n^{1/2}\mathbf{I}_{\text{dim}(g_2)}&
	\end{bmatrix}$$
	and the vectors $$\gamma(\zeta)=\begin{pmatrix}
	V^0(\eta)(\eta_1-\eta_1^0)\\\E_n\left[\bar g_{2n}(\theta_2)\right]
	\end{pmatrix}, \quad q_{2}(\eta_2,\eta_3)=\begin{pmatrix}
	q_{12,n}(\eta_2,\eta_3)\\\mathbf{0}
	\end{pmatrix}. $$  Then, up to $o_p(1)$ terms,
	\begin{flalign*}
	\sqrt{n}\bar{g}_n(\eta,\theta_{2})&=\sqrt{n}\left\{\bar{g}_n(\eta,\theta_{2})-\E_n[\bar g_n(\eta,\theta_2)]\right\}+\sqrt{n}\E_n[\bar g_n(\eta,\theta_2)]\\&=\nu_n(\eta,\theta_2)+{\widetilde{\Lambda}_n}\gamma(\zeta)+\sqrt{n}q_{2}(\eta_2,\eta_3).
	\end{flalign*}

	The remainder of the result follows a similar strategy to Theorem 2.1 in \citet{antoine2012efficient}. Let $W$ be a positive-definite $H\times H$ matrix, and define $\|x\|^2_{W}:=x'Wx$. For $\nu_n(\zeta)$, $\widetilde{\Lambda}_n$ and $\gamma(\zeta)$ as above, we can rewrite the CUGMM objective function in the rotated parameter space as $$J_{n}[\zeta,\zeta]/n=\left\|\frac{\nu_{n}(\zeta)}{\sqrt{n}}+\frac{\widetilde{\Lambda}_n}{\sqrt{n}}\gamma(\zeta)+q_{2}(\eta_2,\eta_3)\right\|_{\Omega_{n}(\zeta)}^2,\text{ for }\Omega_n(\zeta):=S_n^{-1}(\zeta).$$ By definition of $\hat{\zeta}_n$, $J_{n}[\zeta^0,\zeta^0]\geq J_{n}[\hat{\zeta}_n,\hat{\zeta}_n]$ which implies
	\begin{flalign}\label{eq_ineq}
	\left\|{\nu_{n}(\zeta^0)}/\sqrt{n}\right\|_{\Omega_{n}(\zeta^0)}^2\geq\left\|{\nu_{n}(\hat{\zeta}_{n})}/\sqrt{n}+{\widetilde{\Lambda}_n}\gamma(\hat{\zeta}_n)/\sqrt{n}+q_{2}(\hat\eta_{2n},\hat\eta_{3n})\right\|_{\Omega_{n}(\hat{\zeta}_n)}^2.
	\end{flalign}
	Define $\Omega_{n}^0:=\Omega_{n}(\zeta^0)$, $\hat{\Omega}_{n}:=\Omega_{n}(\hat{\zeta}_n)$, $x_n:=\nu_n(\hat{\zeta}_n)$, $y_n:=\widetilde{\Lambda}_n\gamma(\hat{\zeta}_n)+\sqrt{n}q_{2}(\hat\eta_{2n},\hat\eta_{3n})$ and $d_n:= \nu_{n}(\hat{\zeta}_{n})'\hat{\Omega}_n \nu_{n}(\hat{\zeta}_{n})-\nu_{n}({\zeta}^{0})'\Omega^0_n\nu_{n}({\zeta}^{0})$. Denote $\lambda_{\text{min}}[A]$ and $\lambda_{\text{max}}[A]$ as the smallest and the largest eigenvalue of a matrix $A$, respectively. Then, from \eqref{eq_ineq}, we obtain
	\begin{flalign}
	0&\geq J_n[\hat{\zeta}_n,\hat{\zeta}_n]-J_n[\zeta^0,\zeta^0]= d_n+\|y_n\|^2_{\hat{\Omega}_n}+2(\hat{\Omega}_n x_n)' y_n\nonumber\\
	&\geq d_n +\|y_n\|^2\lambda_{\text{min}}\left[\hat{\Omega}_n\right]-2\|y_n\|\|\hat{\Omega}_nx_n\|.\label{cons1}
	\end{flalign}
	Defining $z_n:=\|y_n\|$, and for $\lambda_{\text{min}}\left[\hat{\Omega}_n\right]>0$, we can re-arrange equation \eqref{cons1} as
	\begin{equation*}
	z_n^2-2z_n\frac{\|\hat{\Omega}_nx_n\|}{\lambda_{\text{min}}\left[\hat{\Omega}_n\right]}+\frac{d_n}{\lambda_{\text{min}}\left[\hat{\Omega}_n\right]}\leq0
	\end{equation*}
	Solving the above equation for $z_n$ yields:
	\begin{equation}
	B_n-\left[B_n^2-C_n\right]^{1/2}\leq z_n\leq B_n+\left[B_n^2-C_n\right]^{1/2},\;\;B_n:= \frac{\|\hat{\Omega}_nx_n\|}{\lambda_{\text{min}}\left[\hat{\Omega}_n\right]},\;\;C_n:=\frac{d_n}{\lambda_{\text{min}}\left[\hat{\Omega}_n\right]}\label{cons2},
	\end{equation}
	where by definition of $C_n$ and $B_n$ we know that $B_n^2-C_n\geq0$. From \eqref{cons2}, the result follows if  $$B_n=O_{p}(1),\text{ and }C_n=O_{p}(1).$$Consider first, $B_n$ and note that
	\begin{flalign*}
	B_n&\leq \|x_n\|\frac{\lambda_{\text{max}}\left[\hat{\Omega}_n\right]}{\lambda_{\text{min}}\left[\hat{\Omega}_n\right]}\leq \sup_{\zeta\in \mathcal{Z}}\|\nu_n(\zeta)\|\frac{\sup_{\zeta\in \mathcal{Z}}\lambda_{\text{max}}\left[{\Omega}_n(\zeta)\right]}{\inf_{\zeta\in \mathcal{Z}}\lambda_{\text{min}}\left[{\Omega}_n(\zeta)\right]}.
	\end{flalign*}
	By the result of Lemma \ref{lem_uniform}, $\sup_{\zeta\in \mathcal{Z}}\|\nu_n({\zeta})\|=O_{p}(1)$. It then follows that $B_n=O_{p}(1)$ so long as, for all $n$ large enough, with probability approaching one,  $$0<\inf_{\zeta\in \mathcal{Z}}\lambda_{\text{min}}\left[\Omega^{}_n(\zeta)\right]\leq \sup_{\zeta\in \mathcal{Z}}\lambda_{\text{max}}\left[\Omega^{}_n(\zeta)\right]< \infty,$$ which is guaranteed to be satisfied for $n$ large enough under the assumptions of the result. For $C_n$, recalling that $d_n=\|\nu_n(\hat{\zeta}_n)\|^2_{\hat{\Omega}_n}-  \|\nu_n({\zeta}^0)\|^2_{{\Omega}_n^0}$, we obtain
	\begin{flalign*}
	|C_n|&\leq 2\sup_{\zeta\in \mathcal{Z}}\left\|\nu_n(\zeta)\right\|^2\frac{\sup_{\zeta\in \mathcal{Z}}\lambda_{\text{max}}\left[{\Omega}_n(\zeta)\right]}{\inf_{\zeta\in \mathcal{Z}}\lambda_{\text{min}}\left[{\Omega}_n(\zeta)\right]}.
	\end{flalign*}Repeating the same argument for $C_n$ as for $B_n$ yields $C_n=O_{p}(1)$. Applying $B_n=O_{p}(1),\;C_n=O_{p}(1)$ to equation \eqref{cons2}, we have $$z_n=\|y_n\|=\|\widetilde{\Lambda}_n\gamma(\hat\zeta_{n})+\sqrt{n}q_{2}(\hat\eta_{2n},\hat\eta_{3n})\|=O_{p}(1)$$ It then follows that,  $$\|\gamma(\hat{\zeta}_n)+q_{2}(\hat\eta_{2n},\hat\eta_{3n})\|=O_{p}\left(1/\bar{\lambda}_n\right).$$
	
	Consistency of $\hat{\zeta}_n$ now follows by modifying the standard argument (see, e.g., \citet{newey1994large}, page 2132). By continuity of $\gamma(\zeta)+q_{2}(\eta_2,\eta_3)$, for any $\epsilon>0$, there exists some $\delta_\epsilon$ such that $$\text{Pr}\left[\|\hat{\zeta}_n-\zeta^0\|>\epsilon\right]\leq \text{Pr}\left[\left\|\left\{\gamma(\hat{\zeta}_n)+q_{2}(\hat\eta_{2n},\hat\eta_{3n})\right\}-\gamma(\zeta^0)-q_{2}(\eta_{2}^0,\eta_{3}^0)\right\|>\delta_\epsilon\right].$$However,
	by \textbf{Assumption 5}, $V^0(\eta)$ is non-zero uniformly for $\eta\in\Upsilon(\theta^0_2)$, so that under the identification condition in \textbf{Assumption 2} and the identification of $q_{12,n}(\eta_2,\eta_3)$ in \textbf{Assumption 3}, we can conclude:
	$$
	\|\gamma(\zeta)+q_{2}(\eta_{2},\eta_{3})\|\leq\sup_{\eta\in\Upsilon(\theta^0_2)}\|V^0(\eta)\|\|\eta_1-\eta^0_1\|+\|\E_n[\bar g_{2n}(\theta_2)]\|+\|q_{12,n}(\eta_{2},\eta_{3})\|=0\quad\iff\quad\zeta=\zeta^0.
	$$
	Therefore,
	\begin{flalign*}
	\text{Pr}\left[\|\hat{\zeta}_n-\zeta^0\|>\epsilon\right]\leq\text{Pr}\left[\delta_\epsilon<\left\|\gamma(\hat{\zeta}_n)+q_{2}(\hat\eta_{2n},\hat\eta_{3n})\right\|\right]=o(1),
	\end{flalign*}where the last equality follows from the fact that $\|\gamma(\hat{\zeta}_n)+q_{2}(\hat\eta_{2n},\hat\eta_{3n})\|=O_{p}\left(1/\bar{\lambda}_n\right)$, and $\bar{\lambda}_n\rightarrow\infty$ as $n\rightarrow\infty$.
\end{proof}

\begin{proof}[Proof of Lemma \ref{local-lem}]
	In the rotated parameter space, the rotated moment function is given by
	$$g_i(\zeta)=a_i r_{1i}(\zeta)+b_i r_{2i}(\theta_2)=\begin{pmatrix}\tilde{a}_i(y_{2i},x_i,z_i)r_{1i}(\zeta)\\\tilde{b}_i(x_i,z_i)r_{2i}(\theta_2)
	\end{pmatrix}=\begin{pmatrix}
	g_{1i}(\zeta)\\g_{2i}(\theta_2)
	\end{pmatrix}.
	$$The $(H\times p)$-dimensional Jacobian matrix $\partial g_i(\zeta)/\partial\zeta'$ is given by
	$$
	\partial g_i(\zeta)/\partial\zeta'=\begin{pmatrix}\partial g_{1i}(\zeta)/\partial\eta'&\partial g_{1i}(\zeta)/\partial\theta_2'\\
	\textbf{O}&\partial g_{2i}(\theta_2)/\partial\theta_2'
	\end{pmatrix}	.
	$$
	For $\Lambda_n$ as in the statement of the result,
	\begin{flalign}
	\frac{\partial \bar{g}_n(\zeta^0)}{\partial\zeta'}\Lambda^{}_n&=\left\{\frac{\partial \bar{g}_n(\zeta^0)}{\partial\zeta'}-\E_n\left[\frac{\partial \bar{g}_n(\zeta^0)}{\partial\zeta'}\right]\right\}\Lambda^{}_n+\left\{\E_n\left[\frac{\partial \bar{g}_n(\zeta^0)}{\partial\zeta'}\right]\right\}\Lambda^{}_n\nonumber\\&=O_p(\varsigma_{n}/\sqrt{n})+o_p(1)+\left\{\E_n\left[\frac{\partial \bar{g}_n(\zeta^0)}{\partial\zeta'}\right]\right\}\Lambda^{}_n\nonumber\\&=o_p(1)+\left\{\E_n\left[\frac{\partial \bar{g}_n(\zeta^0)}{\partial\zeta'}\right]\right\}\Lambda^{}_n.\label{deriv:decom}
	\end{flalign}The second equality follows from \textbf{Assumption 5}, and the uniform convergence of the remaining derivatives, which follows from \textbf{Assumptions 1, 2} and a ULLN for iid data. The third equation follows from the fact that $\varsigma_{n}/\sqrt{n}=o(1)$. For $\Lambda_{1n}$ denoting the diagonal matrix $$\Lambda_{1n}:=\begin{pmatrix}
	\varsigma_n&\textbf{O}\\\textbf{O}&\mathbf{I}_{k_x+1}
	\end{pmatrix}$$
	we decompose the $(p\times p)$-dimensional matrix $\Lambda_n$ as $$ \Lambda_n=\begin{pmatrix}
	\Lambda_{1n}&\textbf{O}\\\textbf{O}&\mathbf{I}_{k_x+k_z}
	\end{pmatrix} .$$
	From this definition, the	last term in equation \eqref{deriv:decom} can be stated as
	\begin{flalign}
	\E_n\left[\frac{\partial \bar{g}_n(\zeta^0)}{\partial\zeta'}\right]\Lambda^{}_n&=\E_n\left[\frac{\partial \bar{g}_n(\zeta^0)}{\partial\eta'}\vdots \frac{\partial \bar{g}_n(\zeta^0)}{\partial\theta_2'}\right]\Lambda^{}_n=\E_n\left[\frac{\partial \bar{g}_n(\zeta^0)}{\partial\eta'}\Lambda^{}_{1n}\;\vdots \;\frac{\partial \bar{g}_n(\zeta^0)}{\partial\theta_2'}\right]
	.\label{eq:derivs}
	\end{flalign}Recalling the functions $q_{11,n}(\eta)$ and $q_{12,n}(\eta_2,\eta_3)$ underlying \textbf{Assumption 3}, the first component in equation \eqref{eq:derivs} can be seen to be given by
	\begin{flalign*}
	\E_n\left[\frac{\partial \bar{g}_n(\zeta^0)}{\partial\eta'}\right]\begin{pmatrix}
	\varsigma_n&\textbf{O}_{}\\\textbf{O}_{}&\textbf{I}_{k_x+1}
	\end{pmatrix}&=\begin{pmatrix}\frac{\partial q_{11,n}(\eta^0)}{\partial\eta_1}\varsigma_n&\textbf{O}_{}\\\textbf{O}_{}&\frac{\partial q_{12,n}(\eta_2^0,\eta_3^0)}{\partial(\eta_2,\eta_3')'}	
	\end{pmatrix}=\begin{pmatrix}
	V^0(\eta^0)&\textbf{O}_{}\\\textbf{O}_{}& \frac{\partial q_{12,n}(\eta_2^0,\eta_3^0)}{\partial(\eta_2,\eta_3')'}		
	\end{pmatrix}\\&=M_1(\eta^0).
	\end{flalign*}	By \textbf{Assumption 3(ii)} the south-east block of $M_1(\eta^0)$ has column rank $1+k_x$, while by \textbf{Assumption 5} the north-east block of $M_1(\eta^0)$ is of column rank $1$. Therefore, since $M_1(\eta^0)$  is block diagonal, conclude that $$\lim_{n\rightarrow\infty}\text{col-rank}\left[M_1(\eta^0)\right]=2+k_x.$$
	For the second term in \eqref{eq:derivs}, recalling the Jacobian of $\partial g_i(\zeta)/\partial\zeta'$, we have that
	\begin{flalign*}
	\E_n\left[\frac{\partial \bar{g}_n(\zeta^0)}{\partial\theta_2'}\right]&=\E_n\left[\begin{pmatrix}
	\partial \bar{g}_{1n}(\eta^0,\theta^0_2)/\partial\theta_2'\\\partial \bar{g}_{2n}(\theta^0_2)/\partial\theta_2'
	\end{pmatrix}\right]=	\begin{pmatrix}\E_n\left[\left(\textbf{O}:\tilde{a}(y_{2i},x_i,z_i)\phi_i(\eta^0,\theta^0_2)\eta^0_1z_i'\right)\right]\\\E_n\left[\tilde{b}(x_i,z_i)\left(x_i'\;:\;z_i'\right)\right]
	\end{pmatrix}
	\end{flalign*}By \textbf{Assumption 5}, the matrix $\E_n\left[\tilde{b}(x_i,z_i)\left(x_i'\;:\;z_i'\right)\right]$ has column rank $(k_x+k_z)$.
	
	Combing the two Jacobian terms, the $H\times p$ dimensional Jacobian matrix in equation \eqref{eq:derivs} can be seen as
	$$
	\E_n\left[\frac{\partial \bar{g}_n(\zeta^0)}{\partial\zeta'}\right]\Lambda^{}_n=\begin{pmatrix}
	M_1(\eta^0)&\E_n\left[\left(\textbf{O}:\tilde{a}(y_{2i},x_i,z_i)\phi_i(\eta^0,\theta^0_2)\eta^0_1z_i'\right)\right]\\\textbf{O}&\E_n\left[\tilde{b}(x_i,z_i)\left(x_i'\;:\;z_i'\right)\right]
	\end{pmatrix}.
	$$
The matrix $$M=\plim_{n\rightarrow\infty}\left\{\frac{\partial \bar{g}_n(\zeta^0)}{\partial\zeta'}\Lambda^{}_n\right\},$$ then exists and satisfies
	\begin{flalign*}
	\text{col-rank}[M]&=\lim_{n\rightarrow\infty} \text{col-rank}\left[M_1(\eta^0)\right]+\lim_{n\rightarrow\infty}\text{col-rank}\left\{\E_n\left[\tilde{b}(x_i,z_i)\left(x_i'\;:\;z_i'\right)\right]\right\}\\
&=(2+k_x)+(k_x+k_z)=p.
	\end{flalign*}
\end{proof}

\begin{proof}[Proof of Theorem \ref{theoremN}]From the first order condition of the CUGMM objective function, $\hat{\zeta}_n$ satisfies
	\begin{flalign}\label{p0}
	 n\frac{\partial\bar{g}_n(\hat{\zeta}_n)'}{\partial\zeta}S_n(\hat{\zeta}_n)^{-1}\bar{g}_n(\hat{\zeta}_n)-W\cdot n\frac{\partial\bar{g}_n(\hat{\zeta}_n)'}{\partial\zeta}S_n(\hat{\zeta}_n)^{-1}\bar{g}_n(\hat{\zeta}_n)=0
	\end{flalign}
	for $W$ defined as
	\begin{align}\label{p1}
	 W\cdot\sqrt{n}\frac{\partial\bar{g}_n(\hat{\zeta}_n)'}{\partial\zeta}=\text{Cov}\left(\frac{\partial\bar{g}_n(\hat{\zeta}_n)'}{\partial\zeta},\bar{g}_n(\hat{\zeta}_n)\right)\left(\textbf{I}_H\otimes\left[S_n(\hat{\zeta}_n)^{-1}\sqrt{n}\bar{g}_n(\hat{\zeta}_n)\right]\right),
	\end{align}
	and where Cov$(\cdot)$
	\begin{align}\label{p2}
	 \text{Cov}\left(\frac{\partial\bar{g}_n(\hat{\zeta}_n)'}{\partial\zeta},\bar{g}_n(\hat{\zeta}_n)\right):=\Bigg[\text{Cov}\left(\frac{\partial\bar{g}_{1n}(\hat{\zeta}_{n})}{\partial\zeta},\bar{g}_n(\hat{\zeta}_n)\right),\cdot\cdot\cdot,\text{Cov}\left(\frac{\partial\bar{g}_{H,n}(\hat{\zeta}_n)}{\partial\zeta},\bar{g}_n(\hat{\zeta}_n)\right)\Bigg].
	\end{align}
	Substituting \eqref{p1} into \eqref{p0}, and multiplying both sides of the equation \eqref{p0} by $n^{-1/2}$, we obtain
	\begin{align}\label{p3}
	 &\sqrt{n}\frac{\partial\bar{g}_n(\hat{\zeta}_n)'}{\partial\zeta}S_n(\hat{\zeta}_n)^{-1}\bar{g}_n(\hat{\zeta}_n)-\text{Cov}\left(\frac{\partial\bar{g}_n(\hat{\zeta}_n)'}{\partial\zeta},\bar{g}_n(\hat{\zeta}_n)\right)\nonumber\\
	 &~~~~~~~~\times\left(\textbf{I}_H\otimes\left[S_n(\hat{\zeta}_n)^{-1}\sqrt{n}\bar{g}_n(\hat{\zeta}_n)\right]\right)S_n(\hat{\zeta}_n)^{-1}\bar{g}_n(\hat{\zeta}_n)=0.
	\end{align}
	Apply the mean value theorem to $\bar{g}_n(\hat{\zeta}_n)$,
	\begin{align*}
	 \bar{g}_n(\hat{\zeta}_n)&=\bar{g}_n(\zeta^0)+\frac{\partial\bar{g}_n(\zeta^*_n)}{\partial\zeta'}(\hat{\zeta}_n-\zeta^0)\nonumber\\
	 &=\bar{g}_n(\zeta^0)+n^{-1/2}\frac{\partial\bar{g}_n(\zeta^*_n)}{\partial\zeta'}\Lambda_nn^{1/2}\Lambda_n^{-1}(\hat{\zeta}_n-\zeta^0).
	\end{align*}
	By Proposition \ref{lem:cons}, $\hat{\zeta}_n$ is consistent and by Lemma \ref{lems:order}, $\sqrt{n}\Lambda_n^{-1}({\hat\zeta_n}-\zeta^0)=O_p(1)$. Then Lemma \ref{lems_old} and \textbf{Assumption 5} yield
	$$
	n^{-1/2}\frac{\partial\bar{g}_n(\zeta^*_n)}{\partial\zeta'}\Lambda_nn^{1/2}\Lambda_n^{-1}(\hat{\zeta}_n-\zeta^0)=n^{-1/2}MO_p(1) +o_p(n^{-1/2})=O_p(n^{-1/2}),
	$$ so that we can conclude
	\begin{align}\label{p4}
	\bar{g}_n(\hat{\zeta}_n)&=\bar{g}_n(\zeta^0)+O_p(n^{-1/2}).
	\end{align}
	
	From \eqref{p4}, the convergence rate of $\bar{g}_n(\hat{\zeta}_n)$ is determined by $\bar{g}_n(\zeta^0)$, and by  Lemma \ref{lem_uniform}, and the fact $\mathbb{E}_n[g_i(\zeta^0)]=0$ (under \textbf{Assumption 2}),
	$$\sqrt{n}\bar{g}_n(\zeta^0)\Rightarrow\nu(\zeta^0),$$
	where $\nu(\zeta^0)$ is a Gaussian process with mean-zero and variance matrix $S(\zeta^0)$. Therefore, $\bar{g}_n(\zeta^0)=O_p(n^{-1/2})$ and together with \eqref{p4}, we have that $\bar{g}_n(\hat{\zeta}_n)=O_p(n^{-1/2})$. Given Lemmas \ref{lem_uniform} and \ref{lem_uniform_deriv}, and the fact that $\sup_{\zeta\in \mathcal{Z}}\|S^{-1}_n(\zeta)\|<\infty$, the above result then yields:
	\begin{align}\label{p5}
	 \text{Cov}\left(\frac{\partial\bar{g}_n(\hat{\zeta}_n)'}{\partial\zeta},\bar{g}_n(\hat{\zeta}_n)\right)=O_p(1),~~\text{and}~~\textbf{I}_H\otimes\left[S_n(\hat{\zeta}_n)^{-1}\sqrt{n}\bar{g}_n(\hat{\zeta}_n)\right]=O_p(1).
	\end{align}
	From $\bar{g}_n(\zeta^0)=O_p(n^{-1/2})$ and the results in \eqref{p5}, the second term on the left hand side of \eqref{p3} is $O_p(n^{-1/2})$. Then, \eqref{p3} becomes
	\begin{align}\label{p6}
	 \sqrt{n}\frac{\partial\bar{g}_n(\hat{\zeta}_n)'}{\partial\zeta}S_n(\hat{\zeta}_n)^{-1}\bar{g}_n(\hat{\zeta}_n)=O_p(n^{-1/2}).
	\end{align}
	Plugging \eqref{p4} into \eqref{p6} and multiplying both sides by $\Lambda_n'$, we obtain
	\begin{align}\label{p7}
	 O_p(n^{-1/2})\Lambda_n'&=\sqrt{n}\Lambda_n'\frac{\partial\bar{g}_n(\hat{\zeta}_n)'}{\partial\zeta}S_n(\hat{\zeta}_n)^{-1}\bar{g}_n(\zeta^0)+\sqrt{n}\Lambda_n'\frac{\partial\bar{g}_n(\hat{\zeta}_n)'}{\partial\zeta}S_n(\hat{\zeta}_n)^{-1}\frac{\partial\bar{g}_n(\zeta^*_n)}{\partial\zeta'}\Lambda_n\Lambda_n^{-1}(\hat{\zeta}_n-\zeta^0).
	\end{align}
	In addition, from to the uniform convergence of $S_n(\zeta)$ to $S(\zeta)$ over $\zeta\in \mathcal{Z}$, which follows from compactness of $\mathcal{Z}$, continuity of $g_i(\zeta)$, \textbf{Assumption 1}, and the consistency of $\hat{\zeta}_n$,
	\begin{align}\label{p9}
	 \|S_n(\hat{\zeta}_n)-S(\zeta^0)\|&=\|S_n(\hat{\zeta}_n)-S(\hat{\zeta}_n)+S(\hat{\zeta}_n)-S(\zeta^0)\|\nonumber\\
	 &\leq\|S_n(\hat{\zeta}_n)-S(\hat{\zeta}_n)\|+\|S(\hat{\zeta}_n)-S(\zeta^0)\|\nonumber\\
	 &\leq\sup_{\zeta\in \mathcal{Z}}\|S_n(\zeta)-S(\zeta)\|+\|S(\hat{\zeta}_n)-S(\zeta^0)\|\nonumber\\
	&=o_p(1).
	\end{align}
	Moreover, by the consistency of $\hat{\zeta}_n$, Lemma \ref{lems_old} and equation \eqref{p9} imply that
	 $$\frac{\partial\bar{g}_n(\hat{\zeta}_n)}{\partial\zeta'}\Lambda_n\overset{p}{\rightarrow}M,~~\Lambda_n'\frac{\partial\bar{g}_n(\hat{\zeta}_n)'}{\partial\zeta}S_n(\hat{\zeta}_n)^{-1}\frac{\partial\bar{g}_n(\zeta^*_n)}{\partial\zeta'}\Lambda_n\overset{p}{\rightarrow}M'S^{-1}M.$$
	Because the $H\times p$ matrix $M$ is full column rank under \textbf{Assumption 5(i)}, then the non-singularity of $S$ and the rank condition of $M$ imply that $\Lambda_n'\frac{\partial\bar{g}_n(\hat{\zeta}_n)'}{\partial\zeta}S_n(\hat{\zeta}_n)^{-1}\frac{\partial\bar{g}_n(\zeta^*_n)}{\partial\zeta'}\Lambda_n$ is invertible for large enough $n$. Hence, from \eqref{p7} and $\Lambda_n'O_p(n^{-1/2})=O_p(\|\Lambda_n/\sqrt{n}\|)=o_p(1)$, we obtain
	\begin{align}\label{p8}
	&\sqrt{n}\Lambda_n^{-1}(\hat{\zeta}_n-\zeta^0)\nonumber\\
	 =&-\left[\Lambda_n'\frac{\partial\bar{g}_n(\hat{\zeta}_n)'}{\partial\zeta}S_n(\hat{\zeta}_n)^{-1}\frac{\partial\bar{g}_n(\zeta^*_n)}{\partial\zeta'}\Lambda_n\right]^{-1}\Lambda_n'\frac{\partial\bar{g}_n(\hat{\zeta}_n)'}{\partial\zeta}S_n(\hat{\zeta}_n)^{-1}\sqrt{n}\bar{g}_n(\zeta^0)+o_p(1).
	\end{align}
	Therefore, based on \eqref{p9}, \eqref{p8} and the asymptotic normality of $\sqrt{n}\bar{g}_n(\zeta^0)$ from Lemma \ref{lem_uniform}, the desired results follow.
\end{proof}

\begin{proof}[Proof of Theorem \ref{theorem3}]

Recalling the definition of $\hat\theta_n^\delta$ in equation \eqref{eq:distort}, a mean value expansion of $\bar{g}_n(\hat{\theta}^\delta_n)$ yields, for $A_n:=R\Lambda_n$ with $R$ defined in \eqref{eq:basis},
	\begin{flalign}
	 \sqrt{n}\bar{g}_{n}(\hat{\theta}^{\delta}_{n})&=\sqrt{n}\bar{g}_{n}(\theta^0)+\frac{1}{\sqrt{n}}\sum_{i=1}^{n}\frac{\partial {g}_{i}(\theta^*_n)}{\partial\theta'}(\hat{\theta}_n^\delta-\theta^0)\nonumber\\&=\sqrt{n}\bar{g}_{n}(\theta^0)+\frac{\partial \bar{g}_{n}(\theta^*_n)}{\partial\theta'}A_{n}\sqrt{n}A_{n}^{-1}(\hat{\theta}_{n}-\theta^0)+\sqrt{n}\frac{\partial \bar{g}_{n}(\theta^*_n)}{\partial\theta'}\begin{pmatrix}
	 \Delta_{1n}\\\mathbf{0}_{p-2-k_x}
	 \end{pmatrix}\label{a5},
	\end{flalign}
	where $\theta^*_n$ is component-by-component between $\hat{\theta}_{n}$ and $\theta^0$ and where $\Delta_{1n}=(\delta_n,-\delta_n,\delta_n\hat\pi_n')'$. We now analyze each of the terms in \eqref{a5}.
	
	For the first term in \eqref{a5}, by Lemma \ref{lem_uniform}, $\sqrt{n}\bar{g}_{n}(\theta^0)=O_p(1)$. For the second term, 	recall the rotated parameter $\zeta=(\eta',\theta_2')'$, where $\zeta:=R^{-1}\theta$, as defined in \eqref{eq:basis}. Under the alternative hypothesis, $\|\hat{\theta}_n-\theta^0\|=o_p(1)$ (by \textbf{Proposition \ref{lem:cons}}), which by  \eqref{eq:basis} also implies $\|\zeta^*_n-\zeta^0\|=o_p(1)$. Then, it follows that
	\begin{flalign*}
	\frac{\partial \bar{g}_{n}(\theta^*_n)}{\partial\theta'}A_{n}\sqrt{n}A_{n}^{-1}(\hat{\theta}_{n}-\theta^0)&=\frac{\partial \bar{g}_n(R\zeta^*_n)}{\partial\zeta}\Lambda_n\sqrt{n}\Lambda_{n}^{-1}\left(\hat\zeta_n-\zeta^0\right)
	=M\cdot O_{p}(1)+o_p(1)=O_{p}(1),
	\end{flalign*}where the second equality follows from Lemma \ref{lems_old} and \ref{lems:order}, in particular $\sqrt{n}\Lambda_n^{-1}(\hat\zeta_n-\zeta^0)=O_p(1)$, and the third
from the fact that $M$ is full column rank. Therefore, the second term in \eqref{a5} is $O_{p}(1)$.

Now, focus on the last term in \eqref{a5}. From $\Delta_{1n}=(\delta_n,-\delta_n,\delta_n\hat\pi_n')'$, and for $\hat{R}_n$ denoting the matrix $R$ with $\pi^0$ replaced by $\hat\pi_n$, we have that $$\begin{pmatrix}\Delta_{1n}\\\mathbf{0}_{p-2-k_x}\end{pmatrix}=\hat{R}_n\begin{pmatrix}\delta_n\\\mathbf{0}_{p-1}\end{pmatrix},$$ so that we can write
	$$
	\sqrt{n}\frac{\partial \bar{g}_{n}(\theta^*_n)}{\partial\theta'}\begin{pmatrix}\Delta_{1n}\\\mathbf{0}_{p-2-k_x}	
	\end{pmatrix}=\sqrt{n}\frac{\partial\bar{g}_{n}({\theta}^\star_n)}{\partial\theta'}\hat{R}\begin{pmatrix}\delta_n\\\textbf{0}_{p-1}\end{pmatrix}.
	$$
	Then, 
	\begin{flalign}
	\sqrt{n}\frac{\partial\bar{g}_{n}({\theta}^\star_n)}{\partial\theta'}\hat{R}\begin{pmatrix}\delta_n\\\textbf{0}_{p-1}\end{pmatrix}&=	\sqrt{n}\frac{\partial\bar{g}_{n}({\theta}^\star_n)}{\partial\theta'}{R}\begin{pmatrix}\delta_n\\\textbf{0}_{p-1}\end{pmatrix}+	\frac{\partial\bar{g}_{n}({\theta}^\star_n)}{\partial\theta'}\sqrt{n}\left(\hat{R}-R\right)\begin{pmatrix}\delta_n\\\textbf{0}_{p-1}\end{pmatrix}\nonumber\\&=\sqrt{n}\frac{\partial\bar{g}_{n}({\theta}^\star_n)}{\partial\theta'}{R}\begin{pmatrix}\delta_n\\\textbf{0}_{p-1}\end{pmatrix}+O_p(\delta_n)\nonumber\\&=\frac{\partial\bar{g}_{n}(R{\zeta}^\star_n)}{\partial\zeta'}\Lambda_n\sqrt{n}\Lambda_n^{-1}\begin{pmatrix}\delta_n\\\textbf{0}_{p-1}\end{pmatrix}+O_p(\delta_n)\nonumber\\&=\begin{pmatrix}V^0(\eta^0)\delta_n\left\{\sqrt{n}/\varsigma_{n}\right\}\\\mathbf{0}_{p-1}	
	\end{pmatrix}+o_{p}(1),\label{last}
	\end{flalign}
	where the second line follows by Lemma \ref{lems:order}, which implies $\sqrt{n}(\hat{R}-R)=O_p(1)$, and the convergence of the sample Jacobian in Lemma \ref{lems_old}, the third line from rewriting terms; the fourth from Lemma \ref{lems_old}; and the last line follows from Lemma \ref{lems_old} and the fact that $M$ is full rank (Lemma \ref{local-lem}). Applying these order results for the three terms in \eqref{a5}, we obtain  $$\sqrt{n}\bar{g}_{n}(\hat{\theta}^{\delta}_{n})=O_p(1)+\begin{pmatrix}V^0(\eta^0)\delta_n\left\{\sqrt{n}/\varsigma_{n}\right\}\\\mathbf{0}_{p-1}	
	\end{pmatrix}+o_{p}(1).$$ Since $\|V^0(\eta^0)\|>0$ by \textbf{Assumption 5}, conclude that $\sqrt{n}\bar{g}_{n}(\hat{\theta}^{\delta}_{n})$ diverges if $\{\sqrt{n}/\varsigma_{n}\}\delta_n\rightarrow\infty$.

	Using the above result, we can now show that  $J_n^\delta$ diverges under the alternative. From the proof of Lemma \ref{lem_uniform},
	\begin{align}\label{FCLT}
	 n^{1/2}\{\bar{g}_n(\theta)-\mathbb{E}_n[\bar{g}_n(\theta)]\}\Rightarrow\nu(\theta),
	\end{align}
	where $\nu(\theta)$ is a Gaussian stochastic process on $\Theta$ with mean-zero and bounded covariance kernel $S(\theta,\theta)$. Since $\hat{\theta}_n^\delta\overset{p}{\rightarrow}\theta^0$ under \textbf{Assumption 5}, the uniform convergence \eqref{FCLT} indicates that the sample covariance matrix satisfies $S_n(\hat{\theta}_n^\delta)\overset{p}{\rightarrow}S(\theta^0)$. Thus, for $n$ large enough, $S_n(\hat{\theta}_n^\delta)$ is positive-definite with bounded maximal eigenvalue. Therefore,
	\begin{align}\label{a6}
	J^\delta_n\geq \lambda_{\text{min}}\left[S^{-1}_n(\hat{\theta}_n^\delta)\right]\left\|\sqrt{n}\bar{g}_{n}(\hat{\theta}^{\delta}_{n})\right\|^2,
	\end{align}
	where $\lambda_{\text{min}}\left[S^{-1}_n(\hat{\theta}_n^\delta)\right]>0$ for large enough $n$. Thus, $\{\sqrt{n}/\varsigma_{n}\}\delta_n\rightarrow\infty$ implies $\underset{{n\rightarrow\infty}}{\text{plim}}\;J_n^\delta\rightarrow\infty$.
\end{proof}

\subsection{Table and Figures}
\begin{table}[H]
	\begin{center}
		\caption{Estimation and Rejection Rates under $\lambda=0.5$ (Significant Level $5\%$, $\rho=0.50$)}\label{tab:tab1}
		\begin{tabular}{clccccc}
			\hline
			\hline
			&& $\sigma_z=1$ & $\sigma_z=1$ & $\sigma_z=1$& $\sigma_z=0.2$ & $\sigma_z=10$ \\
			&& $\sigma_v=0.2$ & $\sigma_v=10$ & $\sigma_v=1$&$\sigma_v=1$ & $\sigma_v=1$ \\
			\hline
			\multirow{11}{*}{n=500}& bias& 0.690& -0.045  &-0.050& -0.058& -0.048 \\
			& s.d. & 4.982& 0.627&1.307& 1.455& 1.501 \\
			& rrmse  & 5.027 & 0.628&1.308& 1.456 & 1.501\\
			& Wald size distortion (2SCML)& -0.003& -0.004&-0.003& -0.004 & 0.000\\
			& Wald size distortion (CUGMM)& -0.026&-0.036&-0.037& -0.031 & -0.031\\
			& SS& 0.061 & 0.056&0.061& 0.060 & 0.063\\
			& SY ($5\%$) & 0.007& 0.005&0.009& 0.008  & 0.004\\
			& SY ($10\%$)& 0.091& 0.085&0.090& 0.076  & 0.080\\
			& Robust ($5\%$)& 0.000  & 0.000&0.000& 0.000& 0.000 \\
			& Robust ($10\%$)& 0.000& 0.000&0.000& 0.001 & 0.000\\
			& DJ & 0.018& 0.022&0.016 & 0.010& 0.017\\
			\hline
			\multirow{11}{*}{n=5000}& bias& 0.550& -0.128&0.002& -0.076 & 0.124\\
			& s.d.  & 4.526& 0.301&1.078& 1.091 & 1.260\\
			& rrmse & 4.557& 0.327&1.078 & 1.093& 1.266\\
			& Wald size distortion (2SCML)  & -0.005& -0.023&-0.009& -0.016& 0.017\\
			& Wald size distortion (CUGMM) & -0.030 & -0.047&-0.033 & -0.040& -0.023\\
			& SS & 0.099& 0.069&0.057& 0.070& 0.085\\
			& SY ($5\%$) & 0.015  & 0.008&0.009 & 0.010  & 0.003\\
			& SY ($10\%$) & 0.132 & 0.088&0.095& 0.091 & 0.119\\
			& Robust ($5\%$)& 0.000& 0.000&0.000 & 0.000 & 0.000\\
			& Robust ($10\%$)& 0.001& 0.002&0.001& 0.001 & 0.000\\
			& DJ  & 0.013& 0.025&0.012 & 0.013 & 0.022\\
			\hline
			\multirow{11}{*}{n=10000}& bias& 0.581& -0.103&0.046& -0.002 & 0.130\\
			& s.d. & 4.354& 0.266&1.050& 0.993& 1.191\\
			& rrmse& 4.391& 0.285&1.051& 0.992& 1.197\\
			& Wald size distortion (2SCML)& 0.007& -0.016&0.001& -0.004 & 0.022\\
			& Wald size distortion (CUGMM)&-0.026& -0.047&-0.030& -0.032 & -0.026\\
			& SS& 0.130& 0.072&0.091& 0.088& 0.103\\
			& SY ($5\%$) & 0.023& 0.012&0.016 & 0.006  & 0.008\\
			& SY ($10\%$)& 0.174& 0.098&0.129& 0.116& 0.151\\
			& Robust ($5\%$)& 0.000& 0.000&0.000& 0.000& 0.000\\
			& Robust ($10\%$)& 0.001& 0.000&0.001& 0.000 & 0.001\\
			& DJ& 0.013& 0.019&0.019 & 0.010 & 0.018 \cr
			\hline
			\hline
		\end{tabular}
	\end{center}
	\footnotesize
	Note: (a) SS rejects the null if $F_n>10$. SY ($5\%$) and SY ($10\%$) reject the null if the Cragg-Donald statistic is larger than the critical value of a maximal $5\%$ and $10\%$ size distortion of a 5$\%$ Wald test, respectively.\\
	(b) For the Robust ($5\%$) and Robust ($10\%$) tests, reject rates are computed based on critical values in Table 1 of \citet{olea2013robust}, corresponding to the effective degree of freedom one and tolerance thresholds $5\%$ and $10\%$, respectively, where the tolerance is the fraction that the Nagar bias relative to the benchmark.\\
	(c) The reject rates of DJ test are computed based on perturbation $\hat{\tilde{\rho}}/\log\{\log(n)\}$ and critical value $\chi^2_{0.95}(2)=5.99$.
\end{table}

\begin{table}[H]
	\begin{center}
		\caption{Estimation and Rejection Rates under $\lambda=0.5$ (Significant Level $5\%$, $\rho=0.95$)}\label{tab:tab2}
		\begin{tabular}{clccccc}
			\hline
			\hline
			&&$\sigma_z=1$ & $\sigma_z=1$&$\sigma_z=1$ & $\sigma_z=0.2$ & $\sigma_z=10$ \\
			&&$\sigma_v=0.2$ & $\sigma_v=10$ & $\sigma_v=1$&$\sigma_v=1$ & $\sigma_v=1$ \\
			\hline
			\multirow{11}{*}{n=500}& bias& 2.422 & -0.023&-0.117  & -0.045 & 0.008 \\
			& s.d.   & 10.316 & 0.758&2.866 & 3.145 & 2.883                    \\
			& rrmse  & 10.591 & 0.758&2.867& 3.144  & 2.881                    \\
			& Wald size distortion (2SCML)& 0.168  & 0.003 &0.110& 0.128 & 0.126 \\
			& Wald size distortion (CUGMM) & 0.110& -0.022&0.073& 0.074 & 0.090 \\
			& SS & 0.072& 0.049&0.053 & 0.062 & 0.061\\
			& SY ($5\%$)& 0.006& 0.004&0.004& 0.010  & 0.007   \\
			& SY ($10\%$) & 0.105 & 0.066&0.076& 0.088 & 0.088                    \\
			& Robust ($5\%$)& 0.000 & 0.000&0.000 & 0.000& 0.000                    \\
			& Robust ($10\%$)& 0.000 & 0.000&0.000 & 0.000  & 0.000 \\
			& DJ& 0.040 & 0.044 & 0.039&0.042 & 0.047                    \\
			\hline
			\multirow{11}{*}{n=5000}& bias& 3.480 & -0.072 &0.304& 0.170  & 0.405\\
			& s.d. & 8.506& 0.444 &2.232& 2.119 & 2.259                    \\
			& rrmse& 9.187 & 0.449&2.251& 2.124 & 2.294                    \\
			& Wald size distortion (2SCML)& 0.236 & 0.014&0.151 & 0.121 & 0.156 \\
			& Wald size distortion (CUGMM)& 0.158 & -0.012&0.091 & 0.076 & 0.102 \\
			& SS& 0.113 & 0.050&0.076 & 0.063 & 0.087 \\
			& SY ($5\%$)& 0.013& 0.005&0.012 & 0.009 & 0.007 \\
			& SY ($10\%$)& 0.158& 0.068&0.099 & 0.085 & 0.120    \\
			& Robust ($5\%$)& 0.000& 0.000&0.000& 0.000& 0.000                    \\
			& Robust ($10\%$) & 0.001 & 0.000&0.000 & 0.000 & 0.000 \\
			& DJ  & 0.034 & 0.017 &0.026 & 0.034  & 0.031  \\
			\hline
			\multirow{11}{*}{n=10000}& bias & 3.329& -0.073&0.533& 0.472 & 0.677 \\
			& s.d.& 8.826& 0.459 &2.167& 1.915  & 1.988                    \\
			& rrmse & 9.429& 0.465&2.230 & 1.971  & 2.099                    \\
			& Wald size distortion (2SCML)  & 0.271& 0.019&0.164& 0.140& 0.177 \\
			& Wald size distortion (CUGMM) & 0.171 & -0.008&0.112 & 0.094 & 0.122 \\
			& SS & 0.138 & 0.047&0.079 & 0.077& 0.090 \\
			& SY ($5\%$)& 0.016 & 0.004&0.006 & 0.008 & 0.008                    \\
			& SY ($10\%$)& 0.185 & 0.074&0.106& 0.102  & 0.116                    \\
			& Robust ($5\%$)& 0.000 & 0.000 & 0.000& 0.000  & 0.000       \\
			& Robust ($10\%$) & 0.000 & 0.000& 0.000& 0.000  & 0.001        \\
			& DJ & 0.027& 0.013&0.031 & 0.021  & 0.031\cr
			\hline
			\hline
		\end{tabular}
	\end{center}
	\footnotesize
	Note: (a) SS rejects the null if $F_n>10$. SY ($5\%$) and SY ($10\%$) reject the null if the Cragg-Donald statistic is larger than the critical value of a maximal $5\%$ and $10\%$ size distortion of a 5$\%$ Wald test, respectively.\\
	(b) For the Robust ($5\%$) and Robust ($10\%$) tests, reject rates are computed based on critical values in Table 1 of \citet{olea2013robust}, corresponding to the effective degree of freedom one and tolerance thresholds $5\%$ and $10\%$, respectively, where the tolerance is the fraction that the Nagar bias relative to the benchmark.\\
	(c) The reject rates of DJ test are computed based on perturbation $\hat{\tilde{\rho}}/\log\{\log(n)\}$ and critical value $\chi^2_{0.95}(2)=5.99$.
\end{table}

\begin{table}[H]
	\centering
	\begin{threeparttable}
		\caption{Data Summary of Married Women LFP (Obs. 753)}\label{tab:tab5}
		\begin{tabular}{lcccc}
			\hline
			\hline
			&Mean&Std. Dev.&Min&Max\cr
			\hline
			LFP&0.57& 0.50&	0	&1\cr
			Education&12.29&    2.28&5	&	17\cr
			Father educ.&8.81&3.57&	0&17\cr
			Mother educ.&9.25&3.37&0&17\cr		
			Experience& 10.63&8.07&0&45\cr
			Exper. square&178.04& 249.63&	0&2025\cr
			Nonwife income ($\$1000$)&20.13&11.64&-0.029&96\cr
			Age&42.54&8.07&	30&60\cr
			$\#$ Kids $<$ 6 years old&0.24&0.52&	0&3\cr
			$\#$ Kids $>$ 6 years old&1.35&1.32&	0&8\cr
			\hline
			\hline
		\end{tabular}
		\begin{tablenotes}
			\small
			\item Note: Education, father/mother education and experience are measured in years.
		\end{tablenotes}
	\end{threeparttable}
\end{table}

\begin{table}[H]
	\centering
	\begin{threeparttable}
		\caption{Tests of Weak Instruments (Significance level $5\%$)}\label{tab:tab6}
		\begin{tabular}{lcccccc}
			\hline
			\hline
			&SS&SY ($5\%$)&SY ($10\%$)&Robust ($5\%$)&Robust($10\%$)&DJ (min \& max)\cr
			\hline
			Statistic&81.89&81.89&81.89&91.44&91.44&0.14 \& 17.44\cr
			Critical value&10&19.93&11.59&8.58&6.17&11.98\cr
			Reject $H_0$& Reject&Reject&Reject&Reject&Reject&Reject\cr
			\hline
			\hline
		\end{tabular}
		\footnotesize
		Note: (a) SS and SY test statistics 81.89 are Kleibergen-Paap $F$-stat, which is heteroskedastic-robust. When assuming homoskedastic standard error, the reduced form $F$-statistic and the Cragg-Donald $F$-stat is 95.70. SS critical value 10 is the rule-of-thumb. SY (5$\%$) and SY (10$\%$) critical values 19.93 and 11.59 are for i,i.d. errors, the maximal desired size distortions 5$\%$ and 10$\%$ of a 5\% Wald test, respectively.\\
		(b) Robust test statistics and critical values are computed using Stata command "weakivtest" (\citet{pflueger2014robust}) based on heteroskedastic-robust s.e. Robust (5\%) and Robust (10\%) critical values 8.58 and 6.17 are for 2SLS with 5$\%$ and 10$\%$ tolerance of the Nagar bias over benchmark, respectively. The estimated effective degrees of freedom with the tolerance \{5\%,10\%\} are 1.82 and 1.84.\\
		(c) The perturbation of DJ test is chosen using the approach in Section \ref{sec:3.3}. The critical value is $\chi^2_{1-0.05/20}(H-p+1)=11.98$.
	\end{threeparttable}
\end{table}

\begin{table}[H]
	\centering
	\small
	\begin{threeparttable}
		\caption{Regression Results of Labor Force Participation (LFP)}\label{tab:tab7}
		\begin{tabular}{lcccccc}
			\hline
			\hline
			\multirow{3}{*}{}
			&\multicolumn{3}{c}{2SCML Probit}&\multicolumn{3}{c}{CUGMM}\cr
			\cmidrule(lr){2-4}\cmidrule(lr){5-7}
			&1st step&2nd step&margin&reduced form&structural eq.&margin\cr
			&(1)&(2)&(3)&(4)&(5)&(6)\cr
			\hline
			Dependent Var.&Education&LFP&&Education&LFP&\cr	
			\cr
			Education&& 0.1503***  & 0.0587***  && 0.1500***  & 0.0587***  \\
			&           & (0.0539)   & (0.0211)   &         & (0.0538)  & (0.0211)  \\
			Experience& 0.0930*** & 0.1213***  & 0.0474***& 0.0929***  & 0.1208***  & 0.0472***  \\
			& (0.0251)  & (0.0194)   & (0.0076)   & (0.0249)  & (0.0195)  & (0.0076)  \\
			Exper. square & -0.0016*& -0.0018*** & -0.0007*** & -0.0016* & -0.0018***&-0.0007*** \\
			& (0.0009)  & (0.0006)   & (0.0002)   & (0.0009)  & (0.0006)  & (0.0002)  \\
			Nonwife income (\$1000)& 0.0452*** & -0.0132**  & -0.0052**  & 0.0453***& -0.0139** & -0.0054** \\
			& (0.0071)  & (0.0061)   & (0.0024)   & (0.0070)  & (0.0061)  & (0.0024)  \\
			Age       & -0.0217** & -0.0518*** & -0.0202*** & -0.0218** & -0.0514*** &-0.0201*** \\
			& (0.0109)  & (0.0087)   & (0.0034)   & (0.0109)  & (0.0088)  & (0.0034)  \\
			\# Kids \textless 6 years old    & 0.2268    & -0.8733*** & -0.3411*** & 0.2268  & -0.8727*** & -0.3412*** \\
			& (0.1570)  & (0.1176)   & (0.0462)   & (0.1561)  & (0.1210)  & (0.0476)  \\
			\# Kids \textgreater 6 years old & -0.0934*  & 0.0395     & 0.0154     & -0.0933* & 0.0396  & 0.0155  \\
			& (0.0554)  & (0.0459)   & (0.0179)   & (0.0551)  & (0.0475)  & (0.0186)  \\
			Father educ. & 0.1552*** &            &            & 0.1551***  &         &         \\
			& (0.0237)  &            &            & (0.0236)  &         &         \\
			Mother educ. & 0.1721*** &            &            & 0.1724***  &         &         \\
			& (0.0252)  &            &            & (0.0250)  &         &   \cr     	
			\hline
			Correlation $\rho$&&-0.0453&&&-0.0453\cr	
			&&(0.1105)&&&(0.1102)\cr	
			$J$-statistic&--&--&--&--&0.122\cr
			Obs.&753&753&753&753&753&753\cr
			\hline
			\hline
		\end{tabular}
		\begin{tablenotes}
			\footnotesize
			\item[Note: (a)] Standard errors (s.e.) in parentheses. Significance *** p$<$0.01, ** p$<$0.05, * p$<$0.1. The s.e. in columns (1)-(3) are heteroskedastic-robust. The s.e. in columns (4)-(6) are computed based on Theorem \ref{theoremN}. According to \citet{antoine2017testing}, when DJ rejects the null, standard inference procedures still work for all practical purpose.
			\item[(b)] For CUGMM estimation, overidentification degree is one. Hansen's $J$-statistic 0.122 is less than $\chi^2_{0.95}(1)=3.84$. Overidentification test fails to reject the null hypothesis that moments are all valid.
			\item[(c)] Correlation $\rho$ is the correlation of errors $(u_i,v_i)$ in structural equation and reduced form.
			\item[(d)] Margins in columns (3) and (6) are computed using the sample average of explanatory variables and IVs.
		\end{tablenotes}
	\end{threeparttable}
\end{table}

\begin{table} [H]
	\centering
	\caption{Data Summary of US Food Aid and Civil Conflict}\label{tab:tab8}
	\begin{subtable}[c]{1.0\textwidth}
		\centering
		\caption{Civil Conflict Onset (obs. 1454) \label{tab:tab8a}}
		\begin{tabular}{lcccc}
			\hline
			\hline
			&Mean&Std. Dev.&Min&Max\cr
			\hline
			Onset of intra-state conflict&0.063&0.244&0&1\cr
			US wheat aid (1000 metric tons)&21.08&59.42&0&791.60\cr
			Lagged US wheat production (1000 metric tons)&59187&8754&36787&75813\cr
			Average US food aid probability 1971-2006&0.387&0.328&0&1\cr
			Peace duration (years)&11.59&9.48&1&46\cr
			Instrument&22936&19924&0&75813\cr
			\hline
			\hline
		\end{tabular}
		\vspace{3mm}
	\end{subtable}
	\quad%
	\begin{subtable}[c]{1.0\textwidth}
		\centering
		\caption{Civil Conflict Offset (obs. 709) \label{tab:tab8b}}
		\begin{tabular}{lcccc}
			\hline
			\hline
			&Mean&Std. Dev.&Min&Max\cr
			\hline
			Offset of intra-state conflict&0.185&0.388&0 &1\cr
			US wheat aid (1000 metric tons)&56.07&123.58& 0&854.7\cr
			Lagged US wheat production (1000 metric tons)&60374& 8626&36787&75813\cr
			Average US food aid probability 1971-2006& 0.503& 0.313& 0 &1\cr
			Conflict duration (years)&8.70&8.45& 1 & 42\cr
			Instrument& 30413&19676& 0&75813\cr
			\hline
			\hline
		\end{tabular}
	\end{subtable}
	\centering
	\begin{tablenotes}
		\footnotesize
		\item Note: An observation is a country and year. Instrument is lag of US wheat production times average probability of receiving any US food aid during 1971 to 2006.
	\end{tablenotes}
\end{table}

\begin{table}[H]
	\begin{center}
		\caption{Tests of Weak Instrument (Significance level $5\%$)}\vspace{-0.1cm}\label{tab:tab9}
		\begin{subtable}[c]{1.0\textwidth}
			\centering
			\caption{Civil Conflict Onset\label{tab:tab9a}}
			\begin{tabular}{lcccccc}
				\hline
				\hline
				&SS&SY ($5\%$)&SY ($10\%$)&Robust ($5\%$)&Robust ($10\%$)&DJ (min \& max)\cr
				\hline
				Statistic&26.07&26.07&26.07&26.39&26.39&0.57 \& 7.50\cr
				Critical value&10&16.38&8.96&37.42&23.11&11.98\cr
				Reject $H_0$& Reject&Reject&Reject&Not Reject&Reject&Not Reject\cr
				\hline
				\hline
			\end{tabular}
			\vspace{1mm}
		\end{subtable}
		\begin{subtable}[c]{1.0\textwidth}
			\centering
			\caption{Civil Conflict Offset\label{tab:tab9b}}
			\begin{tabular}{lcccccc}
				\hline
				\hline
				&SS&SY ($5\%$)&SY ($10\%$)&Robust ($5\%$)&Robust ($10\%$)&DJ (min \& max)\cr
				\hline
				Statistic&17.29&17.29&17.29&17.49&17.49&1.50 \& 9.46\cr
				Critical value&10&16.38&8.96&37.42&23.11&11.98\cr
				Reject $H_0$& Reject&Reject&Reject&Not Reject&Not Reject&Not Reject\cr
				\hline
				\hline
			\end{tabular}
			\vspace{-3mm}
		\end{subtable}
	\end{center}
	\footnotesize
	Note: (a) For both onset and offset data, SS and SY test statistics are Kleibergen-Paap $F$-stat (\citet{kleibergen2006generalized}) based on clustered s.e. by countries, to be consistent with \citet{nunn2014us}. SS critical value 10 is the rule-of-thumb. SY (5$\%$) and SY (10$\%$) critical values 16.38 and 8.96 are for i.i.d. errors, one endogenous regressor and one IV, desired maximal size distortion 5$\%$ and 10\% of a 5$\%$ Wald test.\\
	(b) Robust test statistics and critical values are computed using Stata command "weakivtest" (\citet{pflueger2014robust}) based on clustered s.e. by countries. For both onset and offset data, Robust (5$\%$) and Robust (10$\%$) critical values 37.42 and 23.11 are for 2SLS with 5\% and 10\% tolerance of the Nagar bias over benchmark, respectively. The estimated effective degrees of freedom with the tolerance $\{5\%,10\%\}$ are both 1.\\
	(c) For the offset data, the Robust test rejects weak IV when tolerance is larger than $20\%$.\\
	(d) The perturbation of DJ test is chosen based on the process in Section \ref{sec:3.3}. The critical value is $\chi^2_{1-0.05/20}(H-p+1)=11.98$.
\end{table}

\begin{table}[H]
	\footnotesize
	\centering
	\caption{Regression Results of US Food Aid and Civil Conflict}\label{tab:tab10}
	\begin{threeparttable}
		\subcaption{Civil Conflict Onset}
		\begin{tabular}{lccccccc}
			\hline
			\hline
			&Nunn $\&$ Qian&
			\multicolumn{3}{c}{2SCML Probit}&\multicolumn{3}{c}{CU-GMM}\cr
			&(2014)&&&&&\cr
			\cmidrule(lr){2-2}\cmidrule(lr){3-5}\cmidrule(lr){6-8}
			\multirow{2}{*}{}&margin&1st step&2nd step&margin&reduced form&structural eq.&margin\cr
			&(1)&(2)&(3)&(4)&(5)&(6)&(7)\cr
			\hline
			Dependent Var.&Onset&Wheat aid&Onset&Onset&Wheat aid&Onset&Onset\cr
			\cr
			Wheat aid&0.000064&&0.0011&0.000114&&-0.0013&-0.000123\cr
			&(0.00026)&&(0.0025)&(0.00027)&&(0.0028)&(0.00028)\cr
			Peace dur.&-0.018***&-1.66&-0.18***&-0.020***&-1.66&-0.1815***&-0.018***\cr
			&(0.0043)&(1.18)&(0.041)&(0.0046)&(1.21)&(0.045)&(0.0055)\cr
			Peace dur.$^{\wedge}2$&	 0.00087***&0.053&0.0087***&0.00093***&0.053&0.0085***&0.00082**\cr
			&(0.00028)&(0.066)&(0.0026)&(0.00029)&(0.072)&(0.0031)&(0.00032)\cr
			Peace dur.$^{\wedge}3$&-0.00001**&-0.00042&-0.00012***&-0.00001**&-0.00042&-0.00011&-0.00001*\cr
			&(0.00000)&(0.0011)&(0.00005)&(0.00001)&(0.0012)&(0.00014)&(0.00001)\cr
			Instrument&&0.0012***&&&0.0012***	\cr	
			&&(0.0002)&&&(0.0002)\cr	
			\hline
			Correlation $\rho$&&&-0.0837&&&0.3109**\cr
			&&&(0.1318)&&&(0.1408)\cr
			$J$-statistic&--&--&--&--&--&0.553&--\cr
			Obs.&1454&1454&1454&1454&1454&1454&1454\cr
			\hline
			\hline
		\end{tabular}
		\subcaption{Civil Conflict Offset}
		\centering
		\begin{tabular}{lccccccc}
			\hline
			\hline
			&Nunn $\&$ Qian&
			\multicolumn{3}{c}{2SCML Probit}&\multicolumn{3}{c}{CU-GMM}\cr
			&(2014)&&&&&\cr
			\cmidrule(lr){2-2}\cmidrule(lr){3-5}\cmidrule(lr){6-8}
			\multirow{2}{*}{}&margin&1st step&2nd step&margin&reduced form&structural eq.&margin\cr
			&(1)&(2)&(3)&(4)&(5)&(6)&(7)\cr
			\hline
			Dependent Var.&Offset&Wheat aid&Offset&Offset&Wheat aid&Offset&Offset\cr
			\cr
			Wheat aid&-0.000428*&&-0.0019*&-0.000446*&&-0.0013&-0.000302\cr
			&(0.00025)&&(0.0011)&(0.00026)&&(0.0021)&(0.00029)\cr
			Conflict dur.&-0.0619***&4.97&-0.2794***&-0.0653***&4.97&-0.2998***&-0.0687***\cr
			&(0.0117)&(4.65)&(0.0525)&(0.0125)&(4.34)&(0.0690)&(0.0132)\cr
			Conflict dur.$^{\wedge}2$&0.0037***&-0.406&0.0164***&0.0038***&-0.406&0.0184**&0.0042***\cr
			&(0.0010)&(0.288)&(0.0046)&(0.0011)&(0.371)&(0.0084)&(0.0011)\cr
			Conflict dur.$^{\wedge}3$&-0.0001***&0.007&-0.0003***&-0.0001***&0.007&-0.0003&-0.0001***\cr
			&(0.0000)&(0.005)&(0.0001)&(0.0000)&(0.009)&(0.0003)&(0.0000)\cr
			Instrument&&0.003***&&&0.003***	\cr	
			&&(0.0007)&&&(0.0006)	\cr	
			\hline
			Correlation $\rho$&&&0.1277&&&0.1768\cr
			&&&(0.1238)&&&(0.1585)\cr
			$J$-statistic&--&--&--&--&--&1.500&--\cr
			Obs.&709&709&709&709&709&709&709\cr
			\hline
			\hline
		\end{tabular}
		\begin{tablenotes}
			\footnotesize
			\item[Note: (a)] Standard errors (s.e.) in parentheses. Significance *** p$<$0.01, ** p$<$0.05, * p$<$0.1. For both panels (a) (b), the s.e. in column (1) is from \citet{nunn2014us}. The s.e. in columns (2)-(4) are clustered s.e. by countries, based on the 2SCML probit estimation. The s.e. in columns (5)-(7) are calculated by bootstrap with 1000 replications. Since DJ test fails to reject its null, implying standard inference procedures may no longer hold, we should be cautious of drawing any inference conclusions based on those s.e reported in the above tables.
			\item[(b)] For CU-GMM estimation, overidentification degree is one. Hansen's $J$-statistics are less than $\chi^2_{0.95}(1)=3.84$. Overidentification test fails to reject the null hypothesis that moments are all valid in both onset and offset cases.
			\item[(c)] Correlation $\rho$ is the correlation of errors $(u_i,v_i)$ in structural equation and reduced form.
			\item[(d)] Margins in columns (4) and (7) are computed based on sample average of explanatory variables and IVs.
		\end{tablenotes}
	\end{threeparttable}
\end{table}

\begin{sidewaysfigure}[p]
	\begin{center}
		\caption{Rejection Rates under $\lambda<0.5$ ($\rho=0.50$)}\label{fig:power0.5}
		\includegraphics[width=1.02\textwidth]{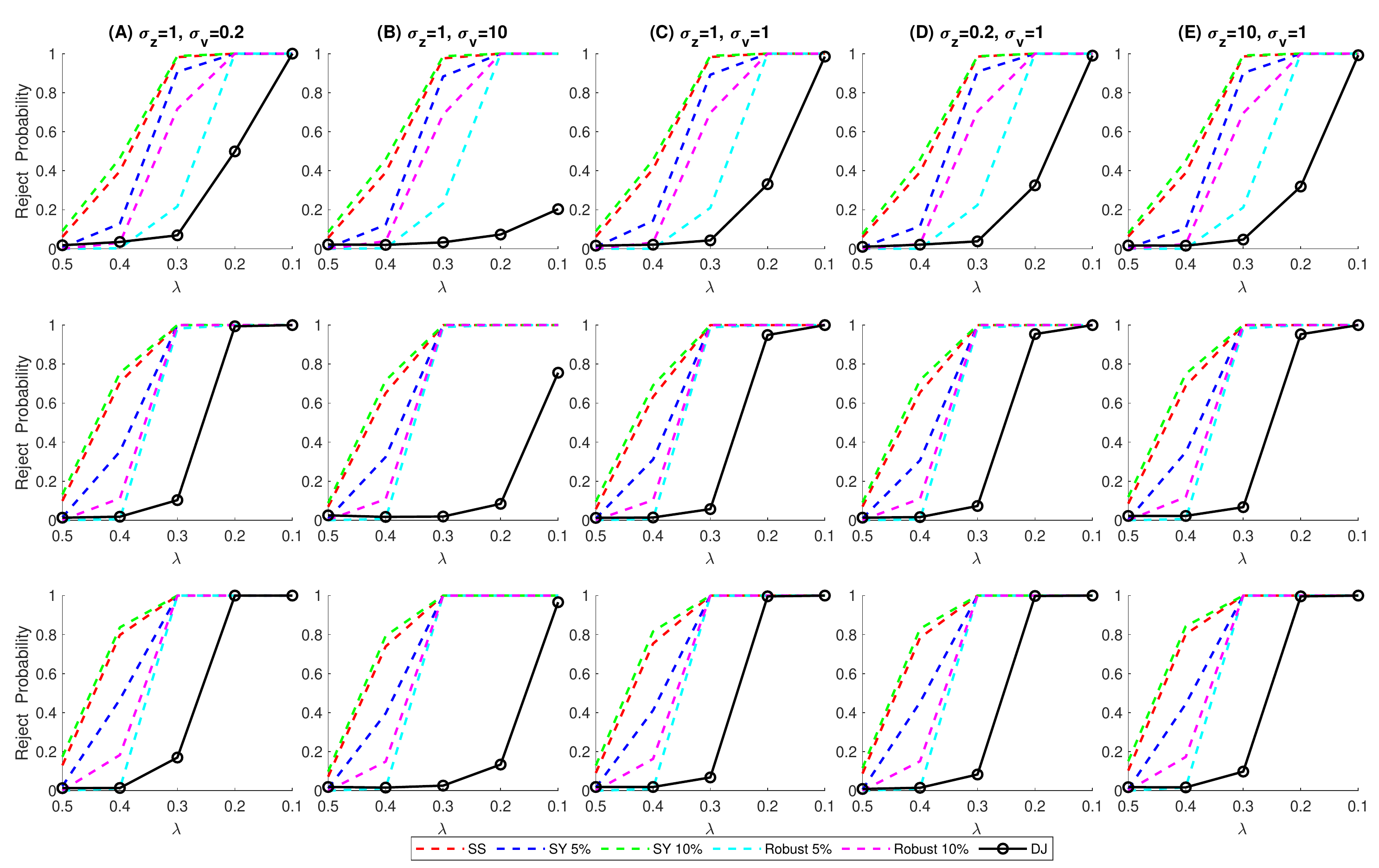}
	\end{center}
	\footnotesize
	Note: x-axis is IV strength $\lambda$. First row $n=500$, second row $n=5000$, third row $n=10000$. The reject rates are computed using critical value $\chi^2_{0.95}(2)=5.99$.
\end{sidewaysfigure}

\begin{sidewaysfigure}[p]
	\begin{center}
		\caption{Rejection Rates under $\lambda<0.5$ ($\rho=0.95$)}\label{fig:power0.95}
		\includegraphics[width=1.02\textwidth]{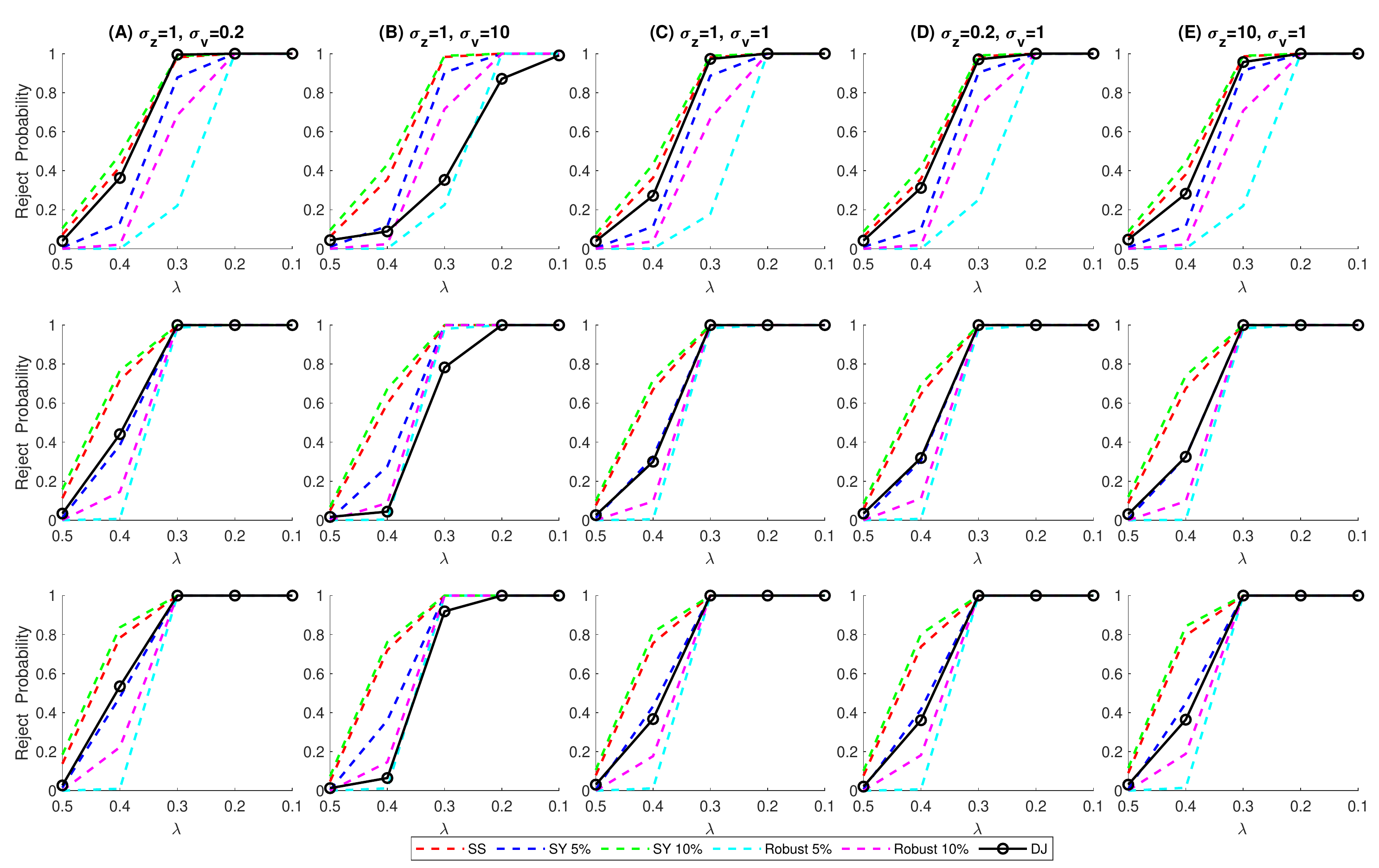}
	\end{center}
	\footnotesize
	Note: x-axis is IV strength $\lambda$. First row $n=500$, second row $n=5000$, third row $n=10000$. The reject rates are computed using critical value $\chi^2_{0.95}(2)=5.99$.
\end{sidewaysfigure}

\begin{sidewaysfigure}[p]
	\begin{center}
		\caption{Size Adjusted Rejection Rates under $\lambda<0.5$ ($\rho=0.50$)}\label{fig:adj_power0.5}
		\vspace{-1cm}
		\includegraphics[width=1.02\textwidth]{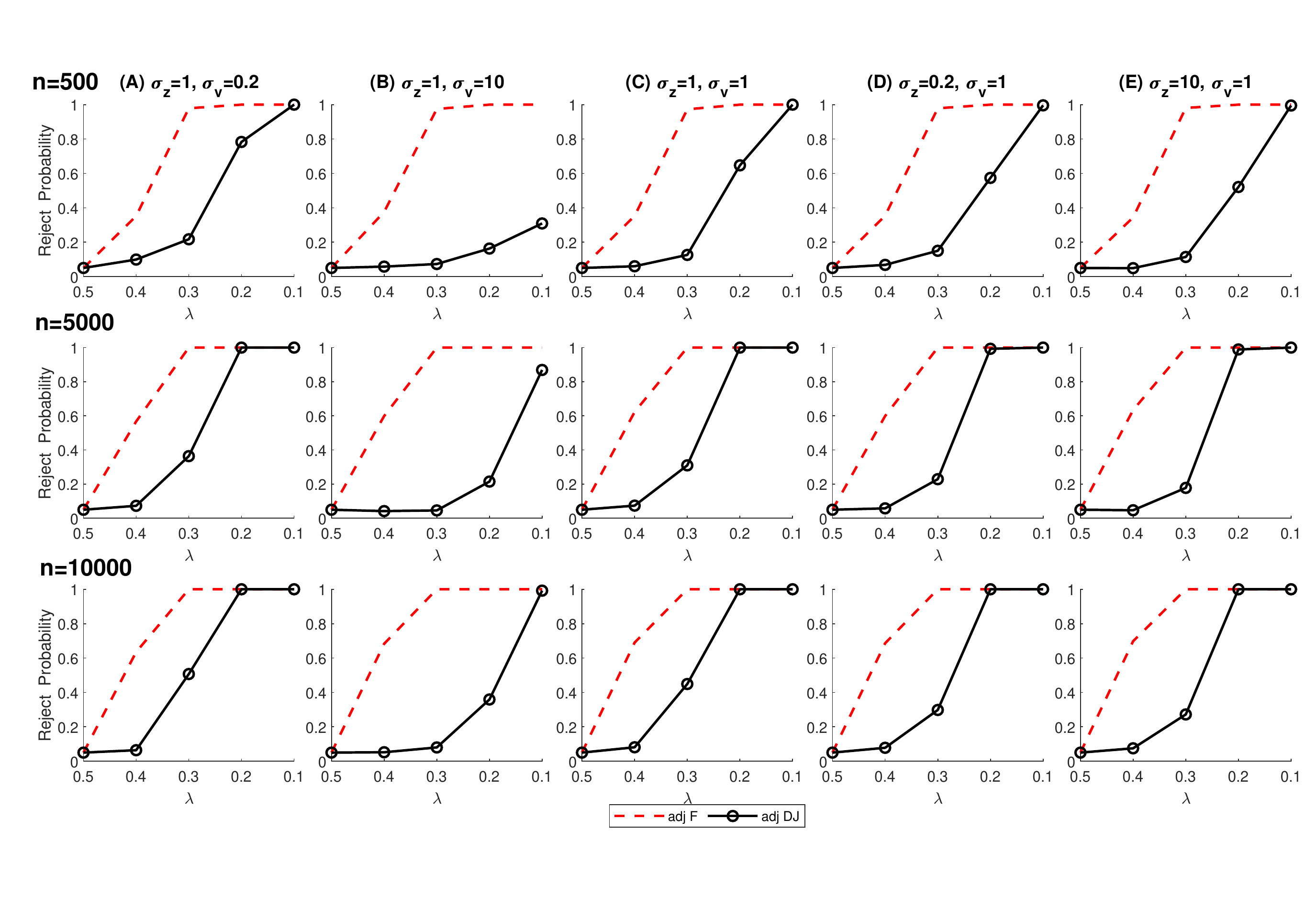}
	\end{center}
	\vspace{-1.5cm}
	\footnotesize
	Note: x-axis is IV strength $\lambda$. First row $n=500$, second row $n=5000$, third row $n=10000$. The test statistic of SS, SY and Robust under one endogenous regressor, one instrument and homoskedastic errors, are the same, i.e. the reduced form regression $F$-stat. The size adjusted power curve is therefore the same for SS, SY and Robust.
\end{sidewaysfigure}

\begin{sidewaysfigure}[p]
	\begin{center}
		\caption{Size Adjusted Rejection Rates under $\lambda<0.5$ ($\rho=0.95$)}\label{fig:adj_power0.95}
		\vspace{-1cm}
		\includegraphics[width=1.02\textwidth]{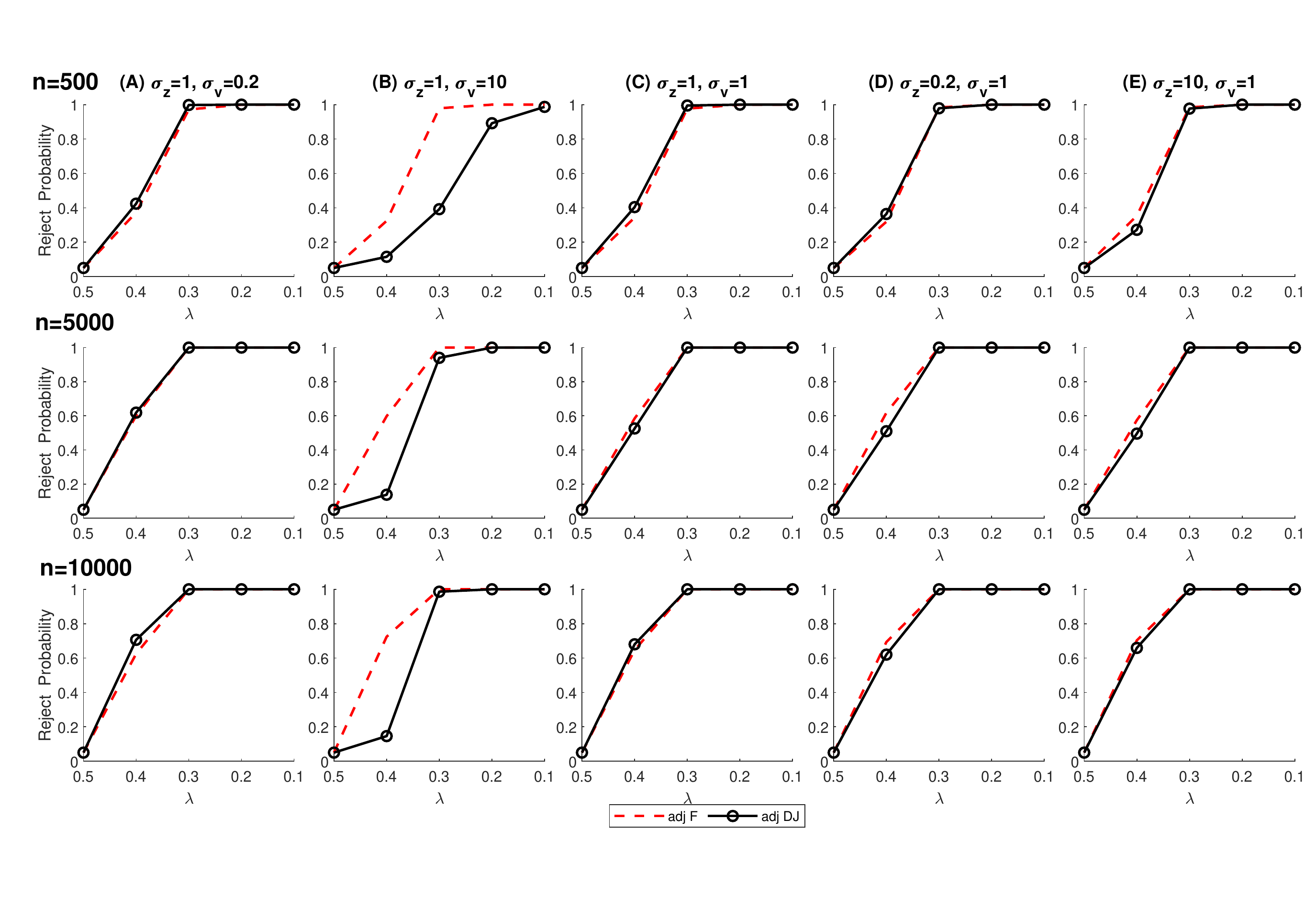}
	\end{center}
	\vspace{-1.5cm}
	\footnotesize
	Note: x-axis is IV strength $\lambda$. First row $n=500$, second row $n=5000$, third row $n=10000$. The test statistic of SS, SY and Robust under one endogenous regressor, one instrument and homoskedastic errors, are the same, i.e. the reduced form regression $F$-stat. The size adjusted power curve is therefore the same for SS, SY and Robust.
\end{sidewaysfigure}

\begin{sidewaysfigure}[p]
	\begin{center}
		\caption{Kernel Density of Standardized CUE for $\alpha$ ($n=10000,\rho=0.50$)}\label{fig:emp0.5}
		\vspace{-1cm}
		\includegraphics[width=1.02\textwidth]{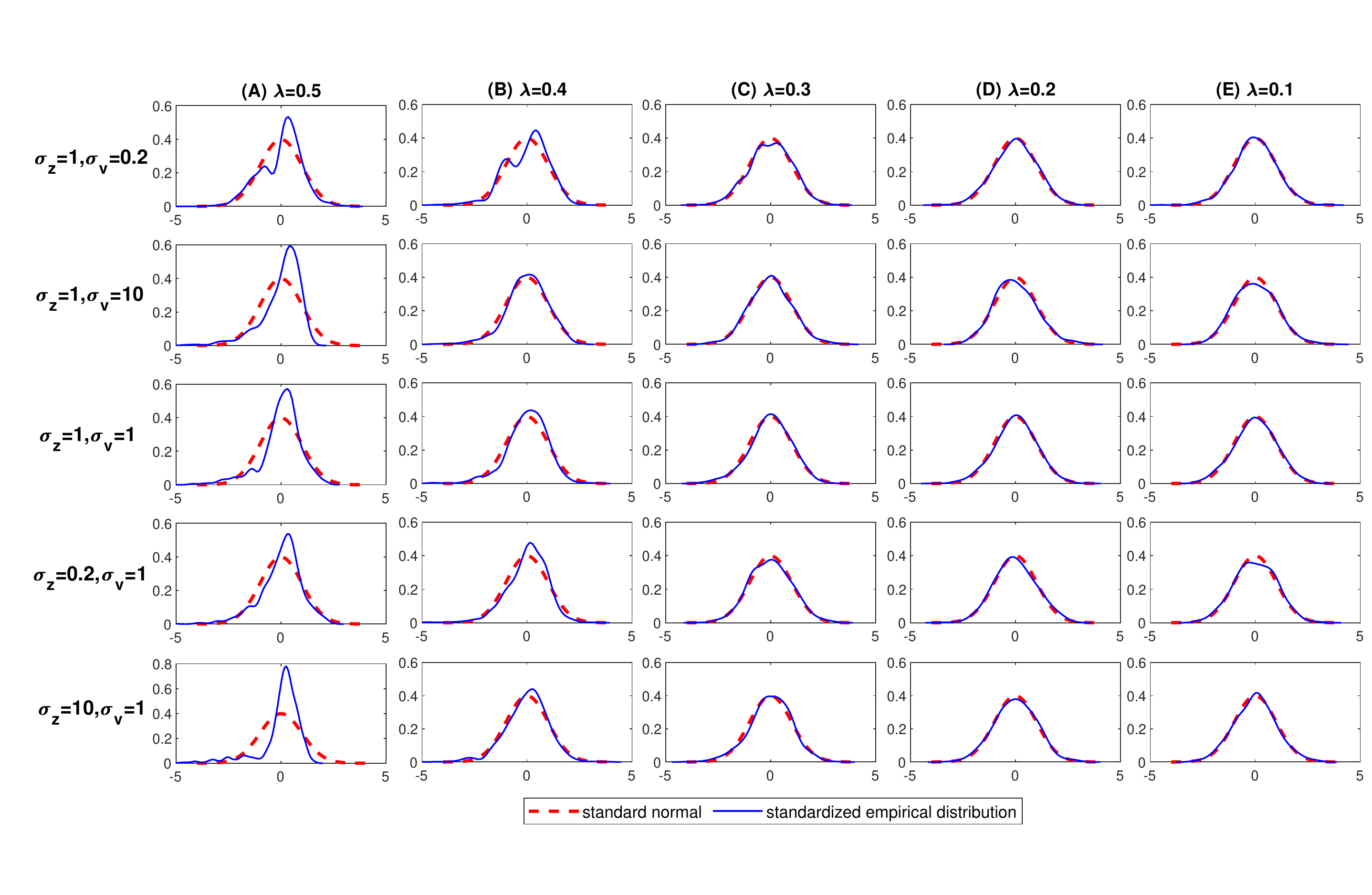}
	\end{center}
	\vspace{-1.5cm}
	\footnotesize
	Note: The standardized CUE for $\alpha$ is $(\hat{\alpha}-\bar{\hat{\alpha}})/s.d(\hat{\alpha})$, where $\bar{\widehat{\alpha}}=1/N\sum_{l=1}^N\widehat{\alpha}_{l}$, $\widehat{\alpha}_{l}$ stands for the $l$-th Monte Carlo CUGMM estimates, and $s.d(\hat{\alpha})$ is the standard deviation defined in \eqref{criterion1}.
\end{sidewaysfigure}

\begin{sidewaysfigure}[p]
	\begin{center}
		\caption{Kernel Density of Standardized CUE for $\alpha$ ($n=10000,\rho=0.95$)}\label{fig:emp0.95}
		\vspace{-1cm}
		\includegraphics[width=1.02\textwidth]{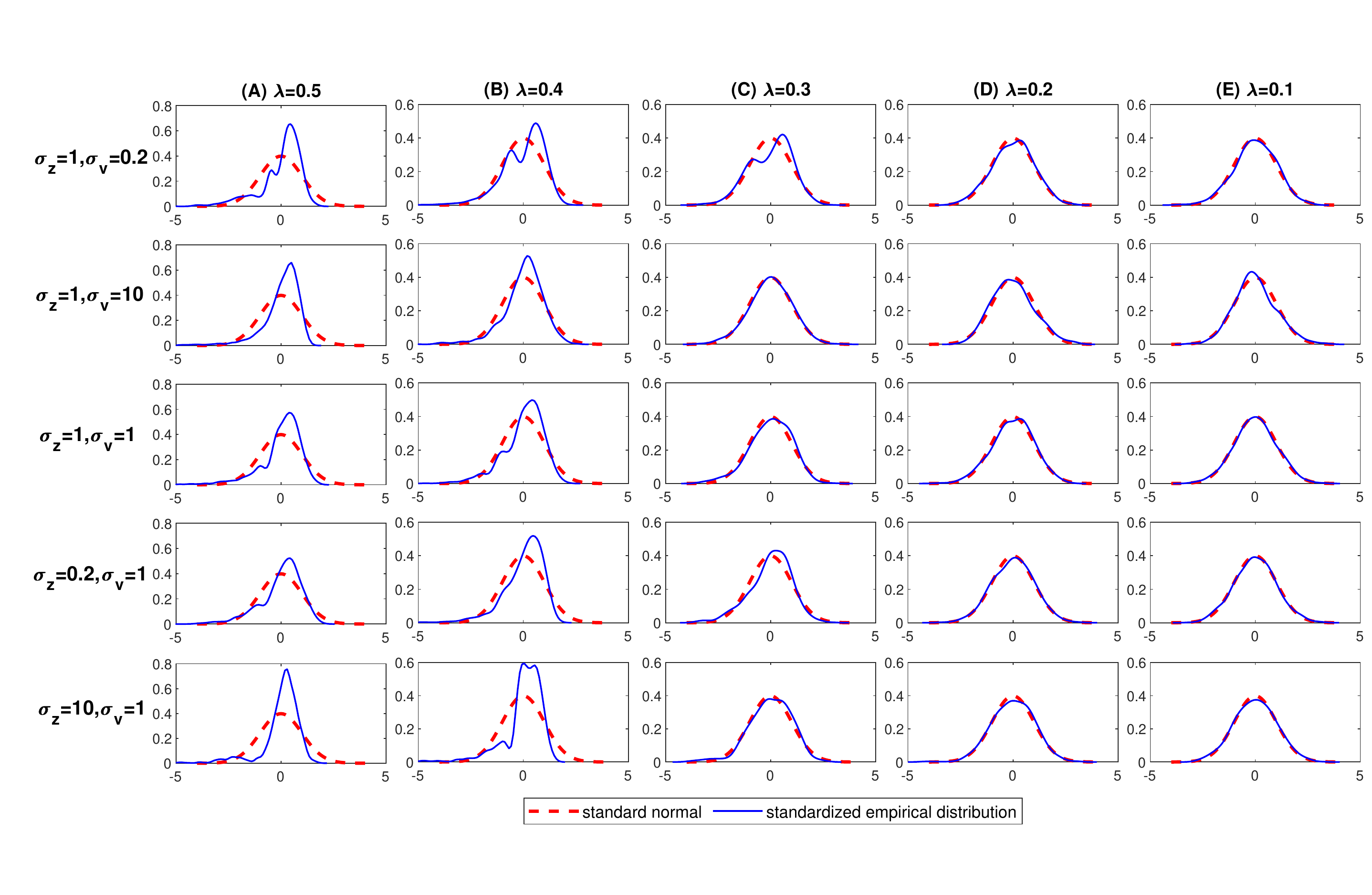}
	\end{center}
	\vspace{-1.5cm}
	\footnotesize
	Note: The standardized CUE for $\alpha$ is $(\hat{\alpha}-\bar{\hat{\alpha}})/s.d(\hat{\alpha})$, where $\bar{\widehat{\alpha}}=1/N\sum_{l=1}^N\widehat{\alpha}_{l}$, $\widehat{\alpha}_{l}$ stands for the $l$-th Monte Carlo CUGMM estimates, and $s.d(\hat{\alpha})$ is the standard deviation defined in \eqref{criterion1}.
\end{sidewaysfigure}

\end{document}